\documentclass{article}
\usepackage{microtype}
\usepackage[table,svgnames,dvipsnames]{xcolor}
\usepackage{soul}
\usepackage{url}
\usepackage[hidelinks]{hyperref} 
\usepackage[small]{caption}
\usepackage{graphicx}
\usepackage{booktabs}
\usepackage{amsmath,multirow}
\usepackage{amsthm,thm-restate,thmtools}
\usepackage{amssymb}
\usepackage{enumerate}
\usepackage{subcaption}
\usepackage{xcolor,xfrac}
\usepackage[most]{tcolorbox}
\usepackage[sort,numbers]{natbib} 
\usepackage[english]{babel}
\usepackage{charter}
\usepackage{helvet}
\usepackage{framed}
\usepackage[normalem]{ulem}  
\usepackage{tikz}
\usetikzlibrary{shapes.arrows}
\usetikzlibrary{shapes.geometric} 
\usepackage{amsfonts} 
\usepackage[T1]{fontenc}
\usepackage[utf8]{inputenc}
\usepackage[top=1 in,bottom=1in, left=1 in, right=1 in]{geometry}  
\usepackage{enumitem}
\usepackage{tabularx}
\usepackage{authblk}
\usepackage{algorithm}
\hypersetup{
     linkcolor    = red, 
     urlcolor     = teal, 
	 citecolor    = blue 
}
\usepackage{mathtools}
\usepackage{nicefrac} 
\usepackage{pifont}

\usepackage{algorithm}
\usepackage{graphicx}
\usepackage[noend]{algpseudocode}
\usepackage{subcaption}

\usepackage{bbm}
\usepackage{comment}
\usepackage{xr}

\usepackage{tkz-graph}

\newcommand{\shortcite}[1]{\citep{#1}}

\usepackage{xspace}

\newtheorem{theorem}{Theorem}
\newtheorem{lemma}{Lemma}
\newtheorem{definition}{Definition}
\newtheorem{corollary}{Corollary}
\newtheorem*{proofsketch}{Proof Sketch}
\renewenvironment{abstract}
{\par\noindent\begin{center}\bfseries\abstractname\end{center}\normalsize\noindent}
{\par}
\newcommand{\ictpp}{\textsf{IC2PP}\xspace}

\begin{document}
\title{Intermittent Strategic Cooperation of Two Selfish Agents on Graphs}

\author{
\begin{minipage}[t]{0.45\textwidth}
\centering
Itay Shedlezki\\
Bar-Ilan University\\
Ramat Gan, Israel\\
\texttt{itay.shedlezki@biu.ac.il}
\end{minipage}
\hfill
\begin{minipage}[t]{0.45\textwidth}
\centering
Noa Agmon\\
Bar-Ilan University\\
Ramat Gan, Israel\\
\texttt{noa.agmon@biu.ac.il}
\end{minipage}
}

\date{}
\maketitle
\begin{abstract}
We study strategic space- and time-constrained cooperation between two self-interested agents through the \emph{Intermittent Strategic Cooperation-Based Two-Agent Path Planning} (\ictpp) problem, a shortest-path game on graphs in which agents navigate toward individual targets while optionally cooperating at specific nodes to reduce their own travel times. Although such cooperation can strictly benefit both agents, it is strategically fragile: agents may deviate at any point along their paths.
Modeled as a 2-player game, we characterize the structure of Pure Nash Equilibrium (PNE) joint strategies in \ictpp, and show that stable cooperation must follow a highly constrained form. We further prove that at least one PNE exists in every instance of \ictpp, and present a polynomial-time algorithm for enumerating all relevant PNEs. 
When multiple equilibria arise, we study coordination mechanisms based on bargaining-theoretic selection concepts and empirically compare equilibrium outcomes in terms of individual travel times and social welfare.
\end{abstract}
\section{Introduction}
Consider the following real-world scenario: two individuals commute from their homes to separate workplaces. Each can independently follow a shortest path, yet under certain circumstances they may benefit from coordinating their routes. For instance, driving together through certain intersections may grant them precedence, enabling faster or safer travel and incentivizing deviation from individually optimal paths in favor of cooperation. 
Such coordination, however, is inherently precarious. It is feasible and beneficial only if both agents arrive in time relevant for forming cooperation, and neither agent may gain by deviating (e.g. by initiating cooperation earlier or later, or by leaving prematurely once its individual benefit diminishes). As a result, cooperation that appears advantageous in principle may fail in practice.
This creates a fundamental tension: cooperation can strictly improve individual outcomes, but only if it is precisely timed and strategically stable. From a game-theoretic perspective, this raises a nontrivial challenge: although cooperation opportunities are local and time-dependent, agents commit to complete paths in advance, and deviations at any point can undermine cooperation. 
In this work, we show that despite this apparent fragility, cooperation at equilibrium is possible, and precisely because of it, such cooperation must follow a highly structured form.
This structure provides a foundation for studying richer multi-agent interactions through pairwise cooperation.
We study this phenomenon through the {\bf Intermittent Strategic Cooperation-Based Two-Agent Path Planning (\ictpp)} problem, a two-player game on a graph in which self-interested agents navigate from individual start nodes to individual targets, seeking to minimize their arrival times while cooperation opportunities in the graph can be leveraged to improve their outcomes. 
These cooperation opportunities may undermine system stability, as they allow agents to improve their path times, and consequently agents following their shortest independent paths do not necessarily form a Nash equilibrium, nor is this outcome necessarily advantageous.
We fully characterize the Pure Nash Equilibria (PNE) of this game and show that any cooperative equilibrium must follow a rigid structure: a single contiguous cooperation segment, preceded and followed by independent paths, with cooperation beginning and ending only at nodes {\em jointly stable} against unilateral deviation.
This structural characterization has two key implications. First, a PNE always exists. 
Second, equilibrium computation becomes tractable: for each cooperation node, there exists at most one non-dominated PNE ending cooperation at that node, meaning that no other PNE yields strictly lower travel times for both agents. Leveraging these insights, we present a polynomial-time algorithm that enumerates all non-dominated equilibrium joint strategies.

Finally, when multiple equilibria arise, they reflect different trade-offs between the agents’ individual benefits. We examine how agents may select among such equilibria using standard bargaining-based solution concepts, and empirically evaluate how cooperation opportunities and path alignment influence equilibrium efficiency.

This work lays the foundation for a novel class of path-planning problems with broad real-world applicability. By focusing on two-agent interactions, we expose fundamental strategic challenges that arise when self-interested agents may benefit from temporary cooperation while still optimizing their own paths. Since many multi-agent interactions can be decomposed into, or approximated by, sequences of pairwise encounters, the two-agent setting provides a principled starting point for studying more sophisticated coalition structures.

\section{Related Work}
\ictpp is related to a diverse set of two-player games, as well as multi-robot and multi-agent problems that arise across markedly different domains, each with distinct assumptions, objectives and interaction models. 
Among those are Multi-Agent Path Finding (MAPF), Autonomous Intersection Management (AIM), congestion games, games on graphs, Dynamic Coalition Formation (DCF), Task allocation, Timed Network Games (TNG), and shared transportation. 

Path planning is a fundamental problem in robotics, underlying nearly all tasks in mobile robot systems, and concerns determining efficient trajectories for the robot from a start to a target location. Multi-Robot Path Planning (MRPP) \cite{antonyshyn2023multiple, lin2022review} extends this to the concurrent navigation of multiple robots in a shared environment, where the presence of other robots greatly increases complexity. Various models of MRPP have been proposed, differing, for example, in system structure (centralized or decentralized), robot characteristics (homogeneous or heterogeneous), and world representation. Various approaches used game-theoretic formulations, accounting for other entities' intent and plans in both cooperative and competitive settings \cite{laine2021multi,burger2022interaction,le2021lucidgames,wang2021game}.
However, all typically share a core assumption: close encounters by robots should be avoided, and in most non-strictly-competitive settings, robots act cooperatively to optimize a common global objective, such as minimizing total cost or makespan. 

Multi-Agent Path Finding (MAPF), Autonomous Intersection Management (AIM), and congestion games typically model 
shared locations as detrimental phenomena to be avoided. MAPF focuses on collision-free path planning, usually under centralized or fully cooperative assumptions~\cite{stern2019multi,bertolini2022decentralized,ho2020decentralized}.
\cite{BonhommeGLDG24} study a MAPF variant that incorporates a limited form of robot cooperation, where certain hazardous target locations can be accessed only when another robot occupies a supporting location from which it can monitor the area. \cite{zhou2026multi} consider another cooperation-involving variant, in which some tasks require teams of multiple agents. Both settings, however, remain fully cooperative: they do not address self-interested robots, nor the possibility of intermittent, time-limited cooperation that agents may strategically initiate or terminate along their paths.
Congestion games capture competition over shared resources whose costs increase with usage~\cite{vocking2006congestion,gourves2015congestion,gairing2020existence,10.1145/1132516.1132529}.
These games are commonly used to model scenarios such as traffic routing, where shared resources impact strategic decisions. 
Specifically, in \cite{HOEFER20115420} the authors study temporal network congestion games with time-dependent costs and coordination mechanisms, analyzing agents' behavior from a game-theoretic perspective. However, this line of work, like the broader congestion game literature, treats interactions between agents as detrimental, and does not account for their potential beneficial effects.
AIM plans vehicle trajectories to avoid conflicts at intersections~\cite{dresner2008multiagent,cederle2024distributedapproachautonomousintersection,6094668}, including distributed methods \cite{cederle2024distributedapproachautonomousintersection}, multi-intersection settings \cite{6094668}, and vehicle platooning~\cite{bashiri2017,kang2025simultaneous,jurasko2025}, where coordinated arrivals improve efficiency. However, these typically assume fully cooperative behavior. In contrast, our work models nodes as opportunities for strategic coordination among self-interested agents, introducing a fundamentally different game-theoretic structure.

Several games on graphs study strategic interactions in spatial settings,
including Stackelberg security games
\cite{clepnerstackelberg,bucarey2021coordinating,ZHANG2022107840}
and assignment-based models such as 
the dinner party problem
\cite{berriaud2023stable,bodlaender2020hedonic,ceylan2023optimal,aziz2024neighborhood},
and topological distance games
\cite{bullinger2024topological,deligkas2024individual}.
However, they typically focus on static assignments or partitions and do not address path-based strategies, where decisions consist of sequences of nodes subject to spatial 
constraints and hence may involve intermittent dynamics.


Coalition formation research \cite{aziz2013computing,aziz2016hedonic,dreze1980hedonic} studies how agents form cooperative groups to pursue mutual benefits, typically focusing on stable partitions or maximizing global utility. Dynamic coalition formation \cite{ye2013self,ARNOLD2002363,1005630,hoefer2018dynamics} extends this line of work to settings in which collaboration can evolve over time, but it overlooks the trajectory coordination or the planning of sequences of short-term, intermittent cooperation.

Probabilistic physical search problems \cite{hazon2013physical,hazon2020probabilistic} address the challenge of locating items in uncertain environments with unknown acquisition costs. Collaborative variants \cite{hazon2009collaborative} study teams searching together, but do not consider agents dynamically rearranging into different teams or influencing each other’s costs.
In \cite{rochlin2016efficiency}, efficiency and fairness are explored in collaborative search involving self-interested agents. Their research examined various types of fairness and cost-sharing mechanisms among self-interested agents cooperating toward a shared goal. However, the work did not consider spatial factors that may influence agents' willingness to cooperate, nor did it address the possibility of agents resigning from the team and re-engaging in other teams.

Task allocation problems \cite{jiang2015survey} address efficient execution of tasks by agents focusing on optimizing a global utility. For example, in Coalition Formation with Spatial and Temporal Constraints (CFSTP) \cite{ramchurn2010coalition} agents are assigned to time-critical tasks across locations, yet this framework assumes fully cooperative agents maximizing collective utility, rather than self-interested agents forming temporary coalitions for individual objectives. Ad-Hoc Teamwork \cite{stone2010ad} refers to the problem of enabling agents to collaborate without prior coordination, focusing on adaptive behaviors and adjustment skills developed through training and experience. However, this line of work also focuses on fully cooperative settings. 

Cost-sharing games study how agents split the cost of jointly used network resources, focusing on equilibrium existence and efficiency under fixed cost-allocation mechanisms~\cite{gupta2008cost,anshelevich2008price}. While this line of work captures strategic incentives in network design, it assumes static participation in shared resources and does not consider path-based strategies in which agents endogenously decide when to cooperate and when to act independently along their trajectories, as in our setting.

Timed Network Games (TNG)~\cite{AVNI2023104996} provide a broad game-theoretic framework for studying timed interactions among agents moving through a network, including cases in which encounters between agents may be beneficial. Our work takes a path-planning perspective on such interactions, focusing on time-weighted graphs in which self-interested agents explicitly construct routes and cooperation is induced by the structure of the selected paths. This viewpoint shifts the analysis from general timed resource usage to cooperation-oriented path construction, allowing us to characterize when cooperation can emerge, remain stable, and terminate without creating profitable deviations.

Research on urban mobility and shared transportation
\cite{kimms2017consideration,bistaffa2015sharing,guajardo2018collaborative,CLEOPHAS2019801,zhao2022dynamic,chen2019dispatching,rheingans2019ridesharing}
primarily focuses on 
coordination mechanisms for ride-sharing and fleet management, often aiming to balance operational efficiency with social welfare. These models typically abstract away from agent-level strategic path choices and do not consider cooperation as a decision made by self-interested agents.
While existing work, such as~\cite{Drews_Luxen_2021,hou2012tictac}, consider cooperation switching with one or multiple hops, they focus on matching drivers and riders under centralized or algorithmic coordination, rather than modeling the strategic decision-making process of self-interested agents from a game-theoretic perspective.


\section{\ictpp: Problem Definition}
We formalize the {\bf Intermittent Strategic Cooperation-Based Two-Agent Path Planning (\ictpp) problem} as follows: 
\begin{tcolorbox}
Given two self-interested agents simultaneously navigating a shared graph-based environment from their respective source to target nodes, where cooperation at specific interaction points may reduce the delays they incur. Each agent seeks to find a path minimizing its individual travel time, given the strategy of the other agent. Their strategic interaction is modeled as a two-player game, and the goal is to identify, characterize, and compare resulting Pure Nash Equilibrium (PNE) outcomes.
\end{tcolorbox}
Consider an environment with two self-interested agents, $a_1$ and $a_2$, simultaneously navigating an undirected graph $G=(V,E)$ representing a shared physical environment.
Each agent $a_i$ starts at its initial node $s_i \in V$, aiming to reach its target node $g_i \in V$ as early as possible.
Each edge $(v,u) \in E$ has an associated travel time $\tau_{vu} > 0$.
Each node $v \in V$ may impose a {\em travel delay} on an agent visiting it, representing, for example, the time required to perform a local task (e.g., opening a gate) before proceeding.
The travel delay incurred by a single agent passing alone at node $v$ is denoted $\tau_v^1$, whereas cooperation may reduce it to $\tau_v^2$ ($\tau_v^1 \geq \tau_v^2 \geq 0$).
The set $V_C = \{v \in V \mid \tau_v^2 < \tau_v^1\}$ includes nodes in which cooperation strictly reduces delay (referred to as \emph{cooperation nodes}).

If agent $a_1$ (w.l.o.g.) arrives at a cooperation node $c \in V_C$ at time $t_1$ before agent $a_2$ arrives at time $t_2$, cooperation requires the earlier agent to wait until the other arrives. As a result, if $a_1$ waits for $a_2$, its total delay at $c$ is $(t_2 - t_1) + \tau_c^2$. In this case, agent $a_2$ incurs a total delay of $\tau_c^2$ at node $c$. Otherwise, both agents incur the non-cooperative delay $\tau_c^1$.
We note that while cooperating at $c$ only improves $a_1$'s local travel delay at $c$ when $t^2 - (\tau^1_c - \tau^2_c) \leq t^1_c$, considering the potential for ongoing cooperation, $a_1$ may choose to wait upon early arrival (before $t^2 - (\tau^1_c - \tau^2_c)$) to cooperate with $a_2$, thereby aiding $a_2$ in improving its departure time, which could also benefit $a_1$ at subsequent nodes (see example in  Figure~\ref{fig:relevant}).
Consequently, an agent’s decision to wait at a node depends on its own arrival time, the arrival time of the other agent, and the agents' planned future path.
\begin{figure}[ht]
    \centering
    \includegraphics[width=0.2\linewidth]{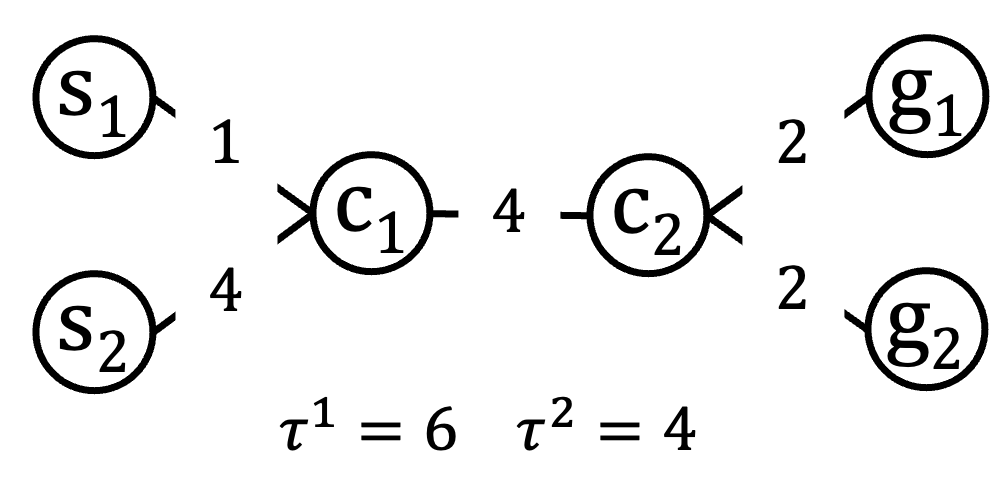}
    \caption{Agent $a_1$ reaches node ${c_1}$ in time for cooperation. While it could depart earlier without cooperating ($t=7$ vs.\ $t=8$), cooperation reduces the total path time (18 vs.\ 19).}
    
    \label{fig:relevant}
\end{figure}

Formally, when agent \( a_i \) arrives at node \( v \) at time \( t^i_v \), its waiting time for the other agent \( a_{-i} \), denoted \( w^i_v \), and its incurred travel delay \( \delta^i_v \), both at node $v$, depend on the arrival time \( t^{-i}_v \) of \( a_{-i} \). Specifically
\[
w^i_v(t^i_v, t^{-i}_v) =
\begin{cases}
t^{-i}_v - t^i_v 
& \text{if } t^i_v \in \left[\,t^{-i}_v - (\tau^1_v - \tau^2_v),\; t^{-i}_v\,\right] 
\text{ or if } a_i \text{ explicitly waits at } v \text{ for } a_{-i}, \\
0 
& \text{otherwise.}
\end{cases}
\]

\[
\delta^i_v(t^i_v, t^{-i}_v) =
\begin{cases}
\tau^2_v 
& \text{if } t^i_v \in \left[\,t^{-i}_v - (\tau^1_v - \tau^2_v),\; t^{-i}_v + (\tau^1_v - \tau^2_v)\,\right] \\
& \quad \text{or if either agent deliberately waits at } v \text{ to synchronize}, \\
\tau^1_v 
& \text{otherwise.}
\end{cases}
\]

The total delay incurred at node \( v \), combining waiting time and the node's travel delay, is denoted by
\[
\lambda^i_v(t^i_v, t^{-i}_v) = w^i_v(t^i_v, t^{-i}_v) + \delta^i_v(t^i_v, t^{-i}_v)
\]

A path \( \pi \) in the graph is a sequence of nodes 
\( \pi_{(1)}, \pi_{(2)}, \ldots, \pi_{(p)} \) such that 
\( (\pi_{(i)}, \pi_{(i+1)}) \in E \) for every \( i \).
We model deliberate waiting of an agent at a node \( v \in V \), for the purpose of synchronizing with another agent, by augmenting the node with the notation \( \dot{v} \). For example, an agent following the path \( s_1, v_1, \dot{v_2}, g_1 \) waits at \( v_2 \) until the other agent arrives before proceeding to \( g_1 \).
The partial path between nodes \( v \) and \( u \) is denoted by \( \pi_{v,u} \).
We denote with $\Pi_{p_1 \ldots p_l}$ the set of all paths starting at $v_{p_1}$, ending at $v_{p_l}$, and visiting $v_{p_2}, \ldots, v_{p_{l-1}}$ in order, possibly including additional nodes (for example, the path $v_1,v_{7},v_2,v_3,v_4$ is in $\Pi_{1,2,4}$).


The \emph{path time} between nodes \( v \) and \( u \) for agent \( a_i \), when following path \( \pi^i \) while the other agent \( a_{-i} \) follows path \( \pi^{-i} \), is defined as the sum of all edge traversal times, travel delays at nodes, and waiting times along \( \pi^i \), excluding the travel delays at the first and last nodes. We denote this time as \( T_{v,u}(\pi^i \mid \pi^{-i}) \).

We use \( D_{v,u}(\pi^i \mid \pi^{-i}) \) to denote the \emph{departure time} of agent \( a_i \) from node \( u \), defined as the path time from $v$ to $u$ plus the waiting and travel dealys incurred at \( u \).
The total path time of the full path $\pi$, from its starting node to its ending node, is denoted by $T(\pi^i \mid \pi^{-i})$. Similarly, the departure time is represented by $D(\pi^i \mid \pi^{-i})$.

Formally,
\[
T_{v,u}(\pi^i \mid \pi^{-i}) =
\begin{cases}
0 & \text{if } v = u, \\[4pt]
\displaystyle
\sum_{(x,y) \in \pi^i_{v,u}} \tau_{xy}
+
\sum_{x \in \pi^i_{v,u} \setminus \{v,u\}}
\lambda^i_x\!\bigl(T_{v,x}(\pi^i \mid \pi^{-i}),\; T_{v,x}(\pi^{-i} \mid \pi^i)\bigr)
& \text{otherwise.}
\end{cases}
\]

\[
D_{v,u}(\pi^i \mid \pi^{-i}) =
T_{v,u}(\pi^i \mid \pi^{-i})
+
\lambda^i_u\!\bigl(T_{v,u}(\pi^i \mid \pi^{-i}),\; T_{v,u}(\pi^{-i} \mid \pi^i)\bigr)
\]

For agent $a_i$, the individual path time along $\pi^i_{v,u}$, disregarding any interactions or cooperative dynamics with the other agent, is denoted by $T_{v,u}(\pi^i)$. The departure time from node $u$ in this independent context is denoted by $D_{v,u}(\pi^i)$.

The \emph{shortest independent path} ($SIP_{v,u}$) denotes the shortest path between nodes $v$ and $u$, assuming the traversing agent does not cooperate with the other agent at any node along the path.  
The \emph{shortest cooperation path} ($SCP_{v,u}$) denotes the shortest path assuming the traversing agent cooperates at all cooperation nodes without waiting for the other agent, incurring delay $\tau^2_w$ at each node $w \in SCP_{v,u}$.

A path $\pi$ for agent $a$ is considered \emph{rational} if, for any two nodes $v$ and $u$ along the path, the condition $T_{v,u}(\pi | \pi) \leq T_{v,u}(SIP_{v,u})$ holds. This implies that if the agent cooperates in all nodes along the path, the route between any two nodes is the shortest possible.

A strategy for an agent is defined as a full path from its start node to its target, including any waiting decisions. We consider a strategic setting in which each agent selects such a full-path strategy \emph{ex ante}, while accounting for the other agent’s start and target nodes, as well as its possible paths, with the goal of minimizing its individual path time.
A \emph{joint strategy} is a pair of paths $(\pi^1, \pi^2)$ specifying the strategies chosen by agents $a_1$ and $a_2$, respectively.
A joint strategy is a Pure Nash Equilibrium (PNE) if neither agent can unilaterally deviate so as to strictly improve its individual travel time.

\section{PNE Properties in \ictpp}

We seek to identify all joint strategies constituting a PNE: 
joint strategies where no agent has an incentive to unilaterally 
deviate.  
When considering a joint strategy in which neither agent’s path includes any cooperation nodes, verifying whether the strategy constitutes a PNE is straightforward: it suffices to check that each agent follows its shortest independent path, which can be done using a standard shortest-path algorithm. 
However, when cooperation nodes are part of the path, the situation becomes more complex, as these cooperation opportunities may incentivize the agents to deviate from their paths in order to improve their arrival times. If there are $m>0$ cooperation nodes, there may exist up to $2^m$ distinct cooperation patterns. Consequently, determining whether a joint strategy that involves cooperation opportunities forms a PNE is nontrivial.
This raises several structural questions: \emph{when should cooperation begin, what form should it take, and when should it end?} We show that cooperation-involving equilibrium strategies exhibit strong structural properties that sharply restrict these possibilities.

We divide our analysis into two classes of joint strategies: those that involve cooperation and those that do not. For each class, we characterize the structural conditions required for a joint strategy to constitute a PNE. We begin with joint strategies that involve cooperation. These can be partitioned into three distinct segments (see Figure \ref{fig:segments}):
\begin{enumerate}
\item \textbf{Joining Segment:} Each agent independently approaches the first cooperation node at which the agents cooperate, referred to as the \emph{cooperation starting node} \( c_s \).
    \item \textbf{Cooperation Segment}: The agents intermittently cooperate at a subset of cooperation nodes from \( {c_s} \) to a later node \( {c_d} \), referred to as the \emph{cooperation departure node}.  
    \item \textbf{Departure Segment}: From \( {c_d} \) onward, each agent follows an independent path to its target node.
\end{enumerate}

\begin{figure}[H]
    \centering
    \includegraphics[width=0.3\linewidth]{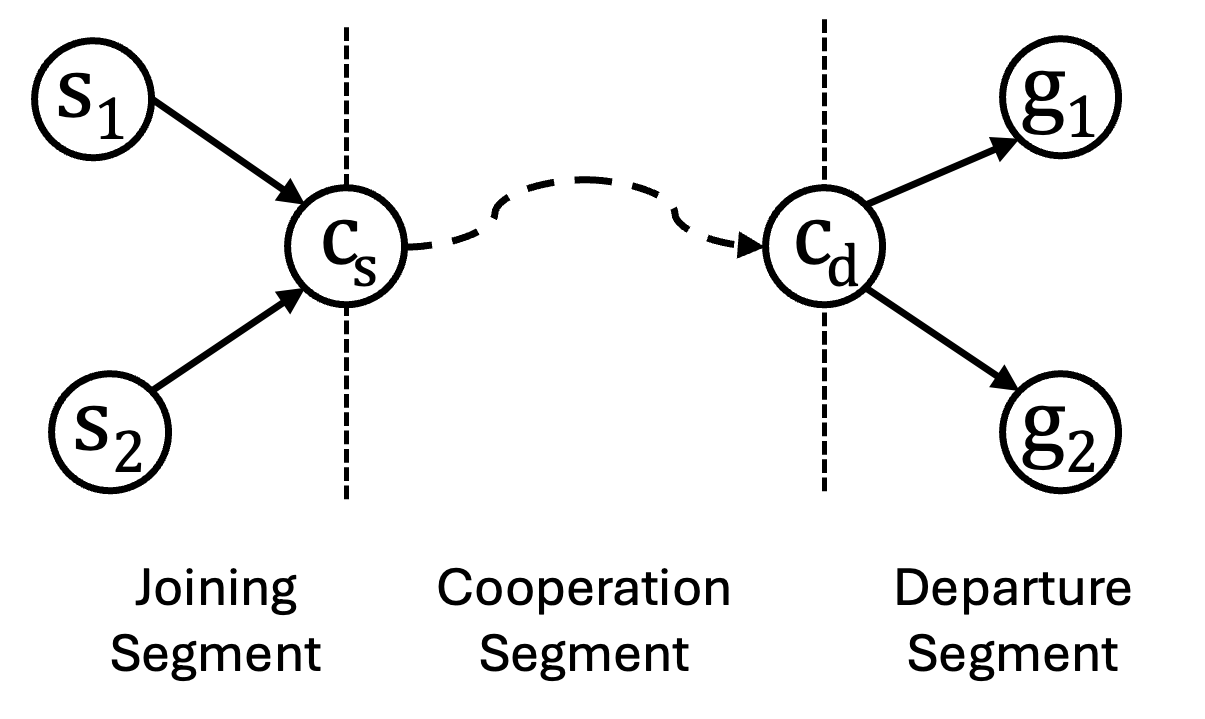}
\caption{Structure of a cooperative joint strategy}
    \label{fig:segments}
\end{figure}

In the following sections, we analyze each of these segments and formulate a set of structural conditions that must hold to prevent unilateral deviations\footnote{Although described in dynamic terms for clarity, deviations in our model correspond to ex ante selection of an alternative full-path.}
 and ensure the joint strategy constitutes a PNE.

\subsection{Joining segment}\label{joining-segment}
In this segment, each agent follows an independent path, involving no cooperation, toward the cooperation starting node \( {c_s} \).
To analyze this segment, we first show that, when the path of one agent is fixed, the optimal cooperative path of the other agent is achieved by joining and initiating cooperation at the earliest cooperation node along the fixed path that is reachable at a \emph{cooperation-relevant time}.
\begin{definition}
Let \(\pi^2\) denote the path of agent \(a_2\), and let \(c \in \pi^2\) be a cooperation node.
We say that \(c\) is reachable by \(a_1\) at a \emph{cooperation-relevant time} if \(a_1\) can arrive at \(c\) within a time interval that still allows cooperation with \(a_2\) at \(c\) \footnote{In our model, although the strategy of \(a_2\) is fixed, if cooperation is locally beneficial, \(a_2\) will trivially wait for \(a_1\). In particular, even if \(a_2\) has already started executing the task at \(c\), if cooperation upon the arrival of \(a_1\) still enables an earlier departure from the node, \(a_2\) will cooperate.}, namely:
\[
T(SIP_{s_1,c}) \leq T(\pi^2_{s_2,c}) + \tau^1_c - \tau^2_c
\]
\end{definition}
\begin{lemma}[Early Cooperation]\label{early}
Each agent prefers to initiate cooperation at the earliest possible cooperation node along the other agent's path.
That is, given an agent $a_i$'s path, the best response for the other agent, $a_{-i}$, is to join $a_i$ as early as possible along its path.
\end{lemma}
\begin{proof}
Consider a two-agent system in which agent \(a_2\) follows a fixed path \(\pi^2\). Let $V_{C^*}$ be the set of all cooperation nodes in $\pi^2$ that $a_1$ can reach within a cooperation-relevant time frame:
\[
    V_{C^*}=\{v\in V_C \mid  T(SIP_{s_1,v}) \leq T_{s_2,v}(\pi^2) + \tau^1_{v} - \tau^2_v \}
\]
The path $\pi^2$ dictates a chronological order for the cooperation nodes. Let ${c^*_i}$ denote the $i$-th cooperation node in $V_{C^*}$ that $a_2$ visits according to the path $\pi^2$.

Then, to prove the lemma we show that the optimal cooperation node to start cooperation at is ${c^*_1} \in V_{C^*}$. \\
Formally we prove that for any node ${c^*_i} \in V_{C^*}$ with $i > 1$, it holds that:
\[
\forall \pi_{c^*_i, g_1} \in \Pi_{c^*_i, g_1}, \quad
T_{s_1, g_1}(SIP_{s_1, c^*_1} \circ \pi^2_{c^*_1, c^*_i} \circ \pi_{c^*_i, g_1}|\pi^2) \leq T_{s_1, g_1}(SIP_{s_1, c^*_i} \circ \pi_{c^*_i, g_1}|\pi^2) 
\]

Given the path $\pi_{c^*_i, g_1}$, we construct the path $\pi'_{c^*_1, g_1} = \pi^2_{c^*_1, c^*_i} \circ \pi_{c^*_i, g_1}$ and examine the full paths $\hat{\pi} = SIP_{s_1, c^*_i} \circ \pi_{c^*_i, g_1}$ and $\hat{\pi}' = SIP_{s_1, c^*_1} \circ \pi'_{c^*_1, g_1}$. We need to show that:
\[
T_{s_1, g_1}(\hat{\pi}'|\pi^2) \leq T_{s_1, g_1}(\hat{\pi}|\pi^2) 
\]
Since the path from ${c^*_i}$ to ${g_1}$, denoted $\pi_{c^*_i, g_1}$, is the same in both cases, it suffices to show that the departure time from ${c^*_i}$ using the path 
$SIP_{s_1, c^*_1} \circ \pi'_{c^*_1, g_1}$ is earlier than the departure time using the path $SIP_{s_1, c^*_i} \circ \pi_{c^*_i, g_1}$.
When using $SIP_{s_1, c^*_i}$, agent $a_1$ arrives at ${c^*_i}$ in a cooperation-relevant time, and its departure time from this node is given by:
\[
D_{s_1, c^*_i}(SIP_{s_1, c^*_i}|\pi^2) = \max\left(T_{s_1, c^*_i}(SIP_{s_1, c^*_i}|\pi^2), T_{s_2, c^*_i}(\pi^2|SIP_{s_1, c^*_i})\right) + \tau^2_{c^*_i}
\]
    However, using $\hat{\pi}'$, $a_1$ arrives at ${c^*_1}$ at time to cooperate and therefore the delay of $a_2$ at ${c^*_1}$ is reduced and the arrival time of $a_1$ at ${c^*_i}$ (which is the same as the arrival time of $a_2$) is $T_{s_1,c^*_i}(\hat{\pi}'|\pi^2) = T_{s_2,c^*_i}(\pi^2|\hat{\pi}') \leq T_{s_2,c^*_i}(\pi^2|SIP_{s_1,c^*_i})$.
    The departure time of $a_1$ from ${c^*_i}$ is then:
    $D_{s_1,c^*_i}(\hat{\pi}'|\pi^2)=T_{s_2,c^*_i}(\pi^2|\hat{\pi}') + \tau^2_{c^*_i}$ and it holds that:
    \[T_{s_1,c^*_i}(\hat{\pi}'|\pi^2) + \tau^2_{c^*_i}=T_{s_2,c^*_i}(\pi^2|\hat{\pi}') + \tau^2_{c^*_i}\ \leq \]
\[\max(T_{s_1,c^*_i}(SIP_{s_1,c^*_i}|\pi^2), T_{s_2,c^*_i}(\pi^2|\hat{\pi}')) + \tau^2_{c^*_i} \leq\]
\[ \max(T_{s_1,c^*_i}(SIP_{s_1,c^*_i}|\pi^2),T_{s_2,c^*_i}(\pi^2|SIP_{s_1,c^*_i}))+\tau^2_{c^*_i}\]
\[
   D_{s_1,c^*_i}(\hat{\pi}'|\pi^2) \leq D_{s_1,c^*_i}(SIP_{s_1,c_i^*}|\pi^2)
\]

\end{proof}

Following Lemma~\ref{early}, one might expect both agents to prefer initiating cooperation at the earliest cooperation node that is mutually reachable. However, doing so may enable one agent to exploit the other by leaving the cooperation prematurely, thereby imposing an inferior path time on the remaining agent. Anticipating such behavior, the other agent may instead prefer to avoid cooperation nodes that admit exploitative deviations, even if this delays the onset of cooperation. Figure \ref{fig:non-cooperative} illustrates such a scenario.

\begin{figure}[ht]
    \centering
    \includegraphics[width=0.4 \linewidth]{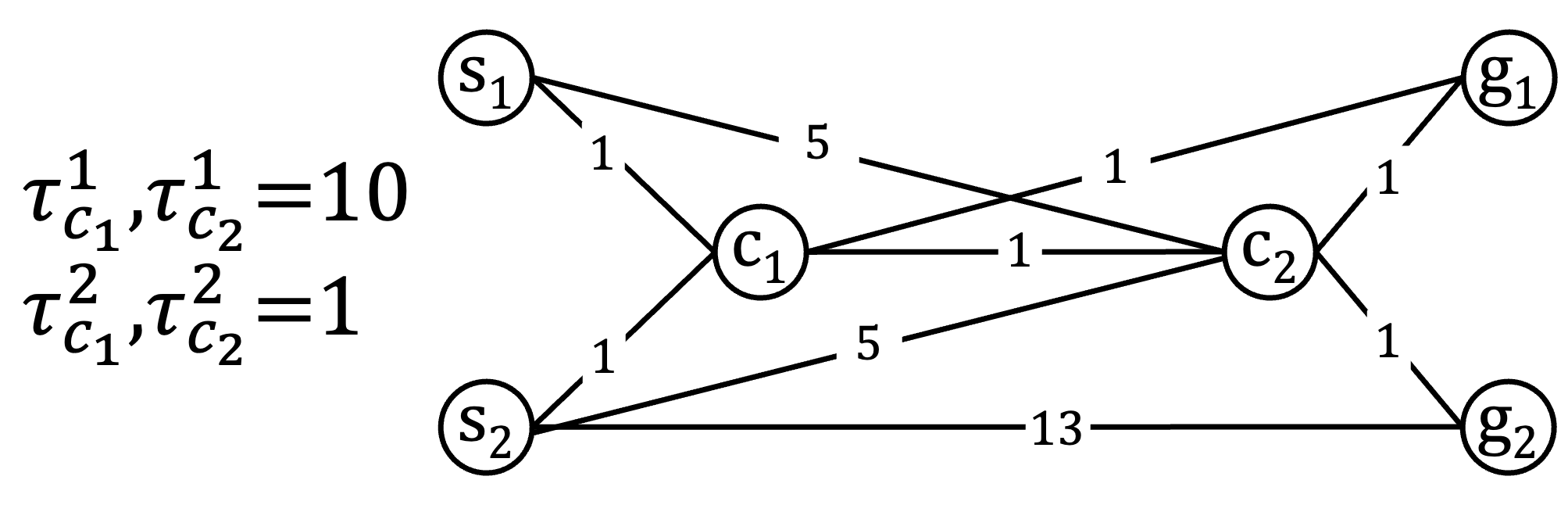}
    \caption{Although ${c_1}$ is the earliest cooperation node reachable by both agents, cooperation there allows agent $a_1$ to deviate and leave agent $a_2$ with an inferior outcome, leading $a_2$ to avoid ${c_1}$.}
     \label{fig:non-cooperative}
\end{figure}

To ensure that neither agent has an incentive to deviate from its intended path during the joining segment, we explicitly consider two types of potential deviations:
\begin{enumerate}
    \item \textbf{Cooperation deviation:} If agent $a_1$’s (w.l.o.g.) path includes a cooperation node reachable by agent $a_2$ at a cooperation-relevant time, then by the Early Cooperation Lemma, agent $a_2$ would prefer to deviate from its intended in order to initiate cooperation at that earlier node (see Figure~\ref{fig:joining-deviation:A}).
    \item \textbf{Arrival time deviation:} If the last-arriving agent at ${c_s}$ can adjust its path to arrive earlier (thereby enabling cooperation to begin sooner) this adjustment benefits both agents (see Figure~\ref{fig:joining-deviation:B}).
\end{enumerate}
\begin{figure}[H]
    \centering
    \begin{minipage}{1\textwidth} 
        \centering
        \begin{subfigure}[b]{0.15\textwidth}
            \includegraphics[width=\textwidth]{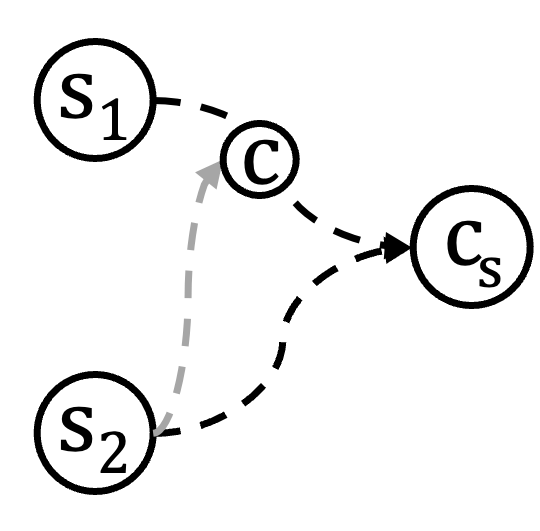}
            \caption{Starting cooperation at an earlier node.}
            \label{fig:joining-deviation:A}
        \end{subfigure}
        \hspace{3cm} 
        \begin{subfigure}[b]{0.15\textwidth}
            \includegraphics[width=\textwidth]{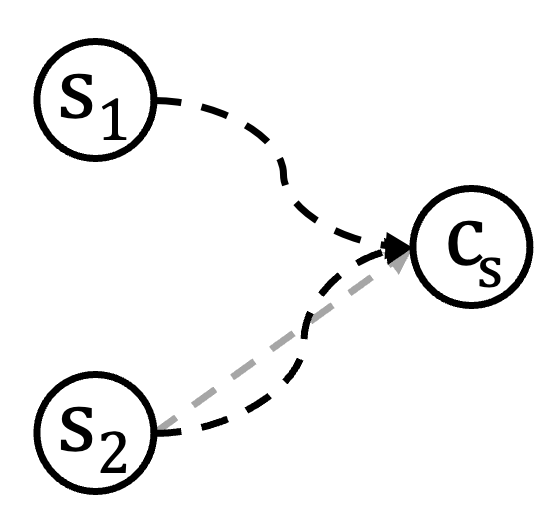}
            \caption{Arriving earlier at $v_{c_s}$ to initiate cooperation sooner.}
            \label{fig:joining-deviation:B}
        \end{subfigure}
    \end{minipage}
    \caption{Possible deviations in the joining segment.}
    
    \label{fig:joining-deviation}
\end{figure}

\begin{definition}[Non-Cooperative Partial Path]\label{def:non-cooperative}
A partial path $\pi_{s_i,v}$ from agent $a_i$'s starting node ${s_i}$ to a node $v \in V$ is defined as a {\em Non-Cooperative (NC) Partial Path} if the first cooperation node along the path that the other agent, $a_{-i}$, can reach from its starting node ${s_{-i}}$ within a cooperation-relevant time is $v$. Formally, for every cooperation node $c \in \pi_{s_i,v} \setminus \{v\}$, $T(SIP_{s_{-i}, c}) > T(\pi_{s_i, c}) + \tau^1_{c} - \tau^2_c$.
\end{definition}
We denote the set of all Non-Cooperative Partial Paths from the starting node ${s_i}$ to node $v$ as $\Pi^{NC}_{s_i,v}$, and the {\em shortest} Non-Cooperative Partial Path from ${s_i}$ to node $v$ as $SIP^{NC}_{s_i,v}$.
Algorithm \ref{algorithm:shortest-non-cooperative-partial-paths} finds the shortest Non-Cooperative partial path from a given starting node ${s_i}$ to all nodes in the graph. The algorithm is identical to Dijkstra's algorithm, except for one modification: if a cooperation node $c \in V_C$ is reachable by the other agent at a time suitable for cooperation, it is removed from the graph once its shortest non-cooperative partial path is identified, preventing further expansions.
This ensures that $c$ is not considered as part of the non-cooperative path to any other node in the graph.
The algorithm initializes by setting the paths from ${s_i}$ to all nodes in the graph to infinity (except ${s_i}$, which is set to $0$ with a trivial path) and defining $Q \gets V$ as the set of unvisited nodes [lines 1–3].
The algorithm then evaluates all unvisited nodes \(v\) that are reachable from ${s_i}$ in ascending order of their path time from \({s_i}\) [lines 4-11]. For each node \(v\), ignoring nodes that can be leveraged by the other agent to initiate an earlier cooperation [line 7], the algorithm iterates over its neighbors. For each neighbor, if the shortest Non-Cooperative path from ${s_i}$ to it via \(v\) is shorter than its current path, the algorithm updates the path and its associated time [lines 8-11].
Once all reachable nodes have been evaluated, the algorithm returns a mapping of each node in the graph to its corresponding shortest Non-Cooperative partial path from \({s_i}\) [line 12].

If a cooperation node $c \in V_C$ is reachable by the other agent at a time suitable for cooperation, it is removed from the graph once its shortest non-cooperative partial path is identified, preventing further expansions.

\begin{algorithm}[ht]
\caption{\textsc{Shortest Non-Cooperative Partial Paths($G, V_C, {s_i}, SIP_{s_{-i}}$)}}\label{algorithm:shortest-non-cooperative-partial-paths}
\begin{algorithmic}[1]
\State $T_{s_i,v} \gets \infty, \pi_{s_i,v} \gets \emptyset$ for all $v \in V$
\State $T_{s_i,s_i} \gets 0, \pi_{s_i,s_i} \gets {s_i}$
\State $Q \gets V$
\While{$Q$ has a node $v$ s.t. $T_{s_i,v}<\infty$}
    \State $v \gets$ node in $Q$ with smallest $T_{s_i,v}$
    \State Remove $v$ from $Q$
    \If{$v \notin V_C$ or $SIP_{s_{-i},v} > T_{s_i,v} +\tau^1_{v} - \tau^2_{v}$}
        \For{each neighbor $u$ of $v$}
            \If{$T{s_i,v} + \tau^1_{v} + \tau_{v,u} < T{s_i,u}$}
                \State $T{s_i,u} \gets T{s_i,v} + \tau^1_{v} + \tau_{v,u}$
                \State $\pi_{s_i,u} \gets \pi_{s_i,v} \circ u $
            \EndIf
        \EndFor
    \EndIf
\EndWhile
\State \textbf{Return} a dictionary from each $v \in V$ to $\pi_{s_i,v}$
\end{algorithmic}
\end{algorithm}

\begin{lemma}\label{lemma:non-cooperative-algorithm}
    Algorithm~\ref{algorithm:shortest-non-cooperative-partial-paths} computes the \textbf{shortest} \emph{Non-Cooperative Partial Paths} from a given starting node to all nodes in the graph in polynomial time.
\end{lemma}

The correctness of Lemma \ref{lemma:non-cooperative-algorithm} follows from the optimality of Dijkstra's algorithm. Since the removal of nodes from the graph is done only when the shortest non-cooperative partial path to them is found, and the removed nodes are not part of any shortest non-cooperative path to any other node when this removal is performed, the path time and non-cooperation conditions are preserved.
Similar to Dijkstra's algorithm, the complexity of Algorithm \ref{algorithm:shortest-non-cooperative-partial-paths} is $\mathcal{O}(|E| + |V| \log |V|)$.
\begin{definition}[Mutually-Robust Non-Cooperative Partial Paths]\label{def:robust-non-cooperative}
Two non-cooperative partial paths, $\pi^1_{s_1,v}$ and $\pi^2_{s_2,v}$, are defined as {\em Mutually-Robust Non-Cooperative Partial Paths} if neither agent can improve the cooperation starting time at $v$ by unilaterally changing its path toward it.
Formally, one of the following conditions holds:  
    (1) Both agents arrive at node $v$ simultaneously: $T(\pi^{1}_{s_1,v}) = T(\pi^{2}_{s_2,v})$. 
    (2) If agent $a_1$ (w.l.o.g.) reaches $v$ first, then $a_2$ arrives at $v$ via its shortest independent path, i.e.,
$T(\pi^{2}_{s_2,v}) = T(SIP_{s_2,v})$.
\end{definition}

Following Definitions~\ref{def:non-cooperative} and~\ref{def:robust-non-cooperative}, to eliminate unilateral deviations within the joining segment, the agents’ \emph{joining segments} must form Mutually-Robust NC Partial Paths. Figure \ref{fig:robust-non-coop} illustrates the necessity of both requirements.

\begin{figure}[ht]
    \centering
    \includegraphics[width=0.4 \linewidth]{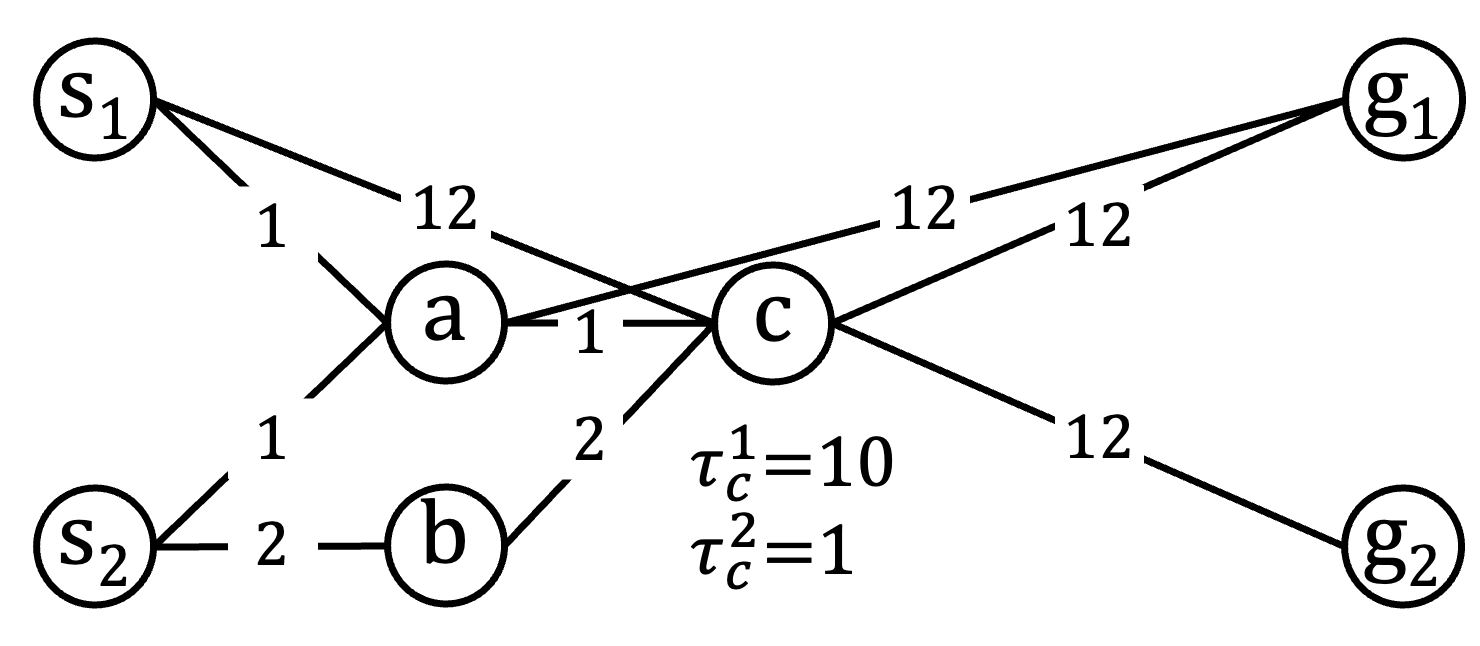}
\caption{The shortest independent path from $s_2$ to $c$ passes through $a$. However, since this path allows $a_1$ to enforce cooperation and then deviate early, it is not considered non-cooperative. The two non-cooperative paths to $c$, namely $(s_1,c)$ and $(s_2,b,c)$, are nevertheless not mutually robust: agent $a_2$, which arrives later at $c$, would prefer to deviate to its shortest independent path $(s_2,a,c)$.}
    \label{fig:robust-non-coop}
\end{figure}

\subsection{Cooperation segment}\label{cooperation-segment}
In this segment, an agent may deviate by seeking additional cooperation opportunities or by avoiding cooperation at certain nodes. 
To analyze this segment, we first show that when two agents cooperate along a partial path, neither agent has an incentive to leave the cooperation and rejoin it later, implying that the cooperation is continuous.
\begin{lemma}[Cooperation Continuity]\label{cont}
Consider two cooperation nodes. If both agents cooperate between them along a \emph{rational} path, then neither agent can strictly improve its path time by unilaterally deviating to an individual path and rejoining later.
\end{lemma}
\begin{proof}
Without loss of generality, assume that agent $a_2$ follows a rational path $\pi$, and that the agents initiate cooperation at a cooperation node $c_i \in \pi$. 
We show that the best response of the other agent, $a_1$, yielding the fastest path from $c_i$ to any subsequent cooperation node $c_j \in \pi$, is to follow the subpath $\pi_{c_i,c_j}$.
Formally, we prove that for any path
\[
\pi' \in \Pi_{s_1,c_i,c_j}
\quad \text{satisfying} \quad
D_{s_1,c_i}(\pi' \mid \pi)
=
D_{s_2,c_i}(\pi \mid \pi')
\]
namely, that the agents cooperate at \(c_i\), the following holds:
\[
T_{c_i,c_j}\!\left(\pi'_{s_1,c_i} \circ \pi_{c_i,c_j} \mid \pi\right)
\;\le\;
T_{c_i,c_j}(\pi' \mid \pi)
\]

Since $a_1$ and $a_2$ depart from node $c_i$ at the same time, traveling together on the same path from $c_i$ to any subsequent node $v$ ensures they arrive at $v$ simultaneously. If $v$ is a cooperation node, this synchronization allows them to cooperate immediately upon arrival, minimizing the latency at that node.

Assume, by contradiction, that there exists a path $\pi' \in \Pi_{s_1, c_i, c_j}$ such that $D_{s_1, c_i}(\pi'|\pi) = D_{s_2, c_i}(\pi|\pi')$ and $T_{c_i, c_j}(\pi'_{s_1, c_i} \circ \pi_{c_i, c_j}|\pi) > T_{c_i, c_j}(\pi'|\pi)$. Since the partial path to ${c_i}$ is the same in both paths $\pi'$ and $\pi'_{s_1, c_i} \circ \pi_{c_i, c_j}$, the deviation between the two paths must occur along the partial path starting at ${c_i}$, that is, $\pi'_{c_i, c_j} \neq \pi_{c_i, c_j}$.

For any deviation between the paths, let \(d\) denote the last node preceding the deviation and let \(r\) denote the first node at which the paths rejoin. Because $a_2$ is rational, it holds that $T_{d, r}(\pi|\pi) \leq T_{d, r}(SIP_{d, r})$. Since there are only two agents, and their partial path from $d$ to $r$ doesn't intersect, no cooperation can occur along the path $\pi'_{d, r}$.
Therefore,
\[
    T_{d, r}(\pi|\pi) \leq T_{d, r}(SIP_{d, r}) \leq T_{d, r}(\pi'|\pi)
\]
This implies that
\[
    T_{c_i, r}(\pi|\pi) \leq T_{c_i, r}(\pi'|\pi)
\]
for any deviation of $\pi'_{c_i, c_j}$ from $\pi_{c_i, c_j}$. 
Thus, 
\[
    T_{c_i, c_j}(\pi'_{s_1, c_i} \circ \pi_{c_i, c_j}|\pi) \leq T_{c_i, c_j}(\pi'|\pi)
\]
which contradicts our initial assumption.\\
Therefore, for any node $c_j \in \pi$ visited after $c_i$, the \emph{fastest path} for $a_1$ from $c_i$ to $c_j$ is achieved by following the path $\pi_{c_i, c_j}$.
\end{proof}

We next show, in~\ref{sec:coop-brp}, how Lemmas~\ref{early} and~\ref{cont} can be used to efficiently compute the best response of an agent to a fixed strategy of the other agent. This characterization will later be used to verify whether the fully independent joint strategy is a PNE in Section~\ref{section:independentPNE}. We then return, in~\ref{section:stability}, to the cooperation segment and introduce the stability condition required for it to constitute part of a cooperative PNE.

\subsubsection{Best Response}\label{sec:coop-brp}
Leveraging Lemmas~\ref{early} and~\ref{cont}, we show that the best cooperation-involving response of an agent to a fixed strategy of the other agent has a simple structure. Rather than considering all possible subsets of cooperation opportunities along the path, it suffices to consider responses that start cooperation as early as possible and maintain it continuously until an optimal departure node. This reduces the relevant strategy space for agent \(a_1\) from exponentially many cooperation patterns to a number linear in \(m\).

The following theorem formalizes this structure for the best cooperation-involving response of agent \(a_1\) to a \emph{fixed} path \(\pi\) of agent \(a_2\), and follows directly from Lemmas~\ref{early} and~\ref{cont}.
\begin{theorem}\label{cor1}
The best cooperation-involving response of agent \(a_1\) to a fixed path \(\pi\) of agent \(a_2\) is the path that initiates cooperation at the earliest cooperation node along \(\pi\) that \(a_1\) can reach at a cooperation-relevant time, denoted \(c^*_1\), and then follows \(\pi\) until the optimal departure node \(d\), defined by
\begin{equation}\label{vd_def}
d =
\arg\min_{v \in \pi}
T(\pi_{c^*_1, v} \circ SIP_{v, g_1} \mid \pi)
\end{equation}

That is,
\[
\pi^* =
SIP_{s_1,c^*_1}
\circ
\pi_{c^*_1,d}
\circ
SIP_{d,g_1}
 \]
\end{theorem}

Intuitively, \(a_1\) first follows the shortest path to the earliest node at which cooperation with \(a_2\) is feasible, then continues jointly with \(a_2\) along \(\pi\), leveraging additional cooperation opportunities along the way. Finally, \(a_1\) departs from \(a_2\) only once, at node \(d\), and follows its shortest independent path to \(g_1\).

Thus, given that \(a_2\)'s strategy \(\pi\) is fixed and known in advance, the best response of \(a_1\) is one of the following:
\textbf{(a)} an \textbf{Independent Shortest Path:}
    The path $SIP_{s_1,g_1}$ from ${s_1}$ to ${g_1}$ without cooperation, computed by a shortest-path algorithm with node weights $\tau^1_v$;  or \textbf{(b)} a \textbf{Cooperation-Assisted Path:}
    The path $SIP_{s_1,c^*_1} \circ \pi_{c^*_1,d} \circ SIP_{d,g_1}$, where \(c^*_1\) and \(d\) denote the earliest cooperation node and the optimal departure node, respectively, as defined in Theorem~\ref{cor1}. 
    The cooperation-assisted path can be constructed efficiently as follows:
    \begin{enumerate}
        \item Calculate all shortest paths from ${s_1}$ to all $v \in V_C$.
        \item Identify the first cooperation node in $\pi$ that $a_1$ can reach at a cooperation-relevant time, denoted as ${c^*_1}$.
        \item Compute the shortest paths from every node along $\pi$ to ${g_1}$.
        \item Determine the optimal node $d$ along $\pi$ for $a_1$ to leave the cooperation, minimizing the path time $T_{c^*_1,g_1}(\pi_{c^*_1,d} \circ SIP_{d,g_1}, \pi)$:
        $
        d=\arg\min_{d \in \pi_{c^*_1,g_2}} T_{c^*_1,g_1}(\pi_{c^*_1,d} \circ SIP_{d,g_1}|\pi)
        $
    \end{enumerate}
It is straightforward to verify that this procedure operates in polynomial time.
Algorithm~\ref{algorithm:best-response} (\textsc{Best Response Path}), provided in Appendix~\ref{app:br-algo}, formalizes the procedure for computing the best response of one agent to a fixed strategy of the other.

\subsubsection{Cooperation Stability}\label{section:stability}
From Lemma~\ref{cont} it follows that the cooperation segment consists of a single continuous cooperation path from the cooperation starting node to the cooperation departure node. However, this does not imply that the segment should simply coincide with the shortest cooperative path \(SCP_{c_s,c_d}\). Since the agents have different individual objectives, their preferred departure nodes may differ. Consequently, an agent may benefit from leaving the cooperation earlier and continuing independently toward its target.
If an agent benefits from leaving cooperation before \( {c_d} \), such a deviation violates the Nash equilibrium conditions.
Moreover, anticipating this, the other agent, may decide to depart even earlier to optimize its own departure node. This process could repeat multiple times, potentially resulting in suboptimal outcomes for both agents (see example in Figure \ref{fig:retreats}). 
\begin{figure}[ht]
    \centering
    \includegraphics[width=0.4 \linewidth]{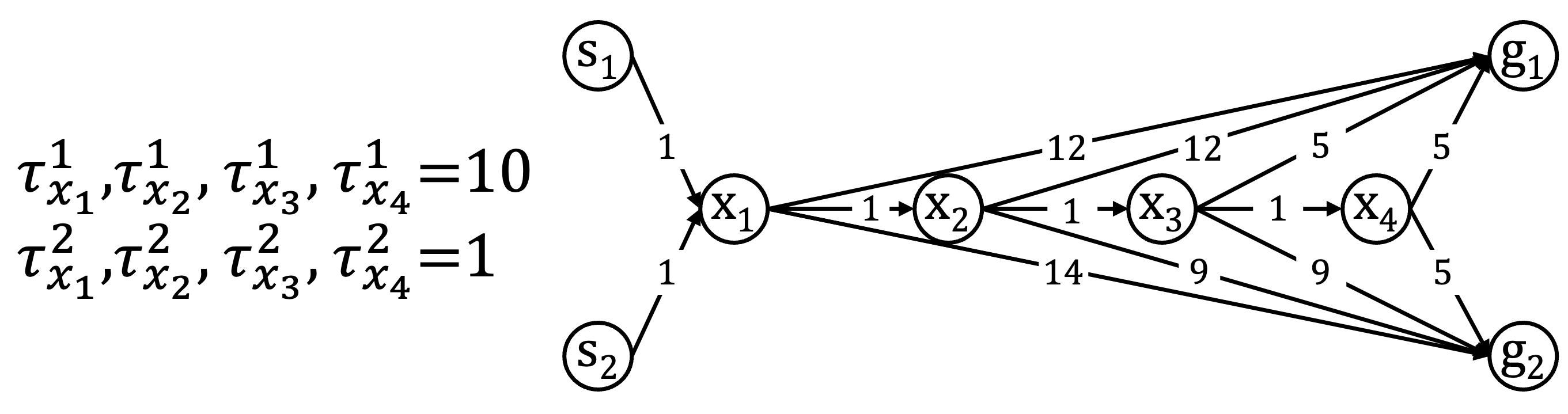}
    \caption{The optimal departure node for $a_2$ along the path $x_1, x_2, x_3, x_4$ is $x_4$, while for $a_1$, it is $x_3$. Along the path $x_1, x_2, x_3$, the optimal departure node for $a_2$ is $x_2$, and for the path $x_1, x_2$, the optimal departure node for $a_1$ is $x_1$. Thus, the only PNE in this scenario is the path containing only $x_1$.}
    \label{fig:retreats}
\end{figure}
Alternatively, an agent may enforce a detour within the cooperation segment to avoid nodes at which the other agent could deviate prematurely, thereby preventing outcomes that would result in an inferior path time. Figure~\ref{fig:suboptimal} illustrates such a scenario.
\begin{figure}[ht]
    \centering
    \includegraphics[width=0.4\linewidth]{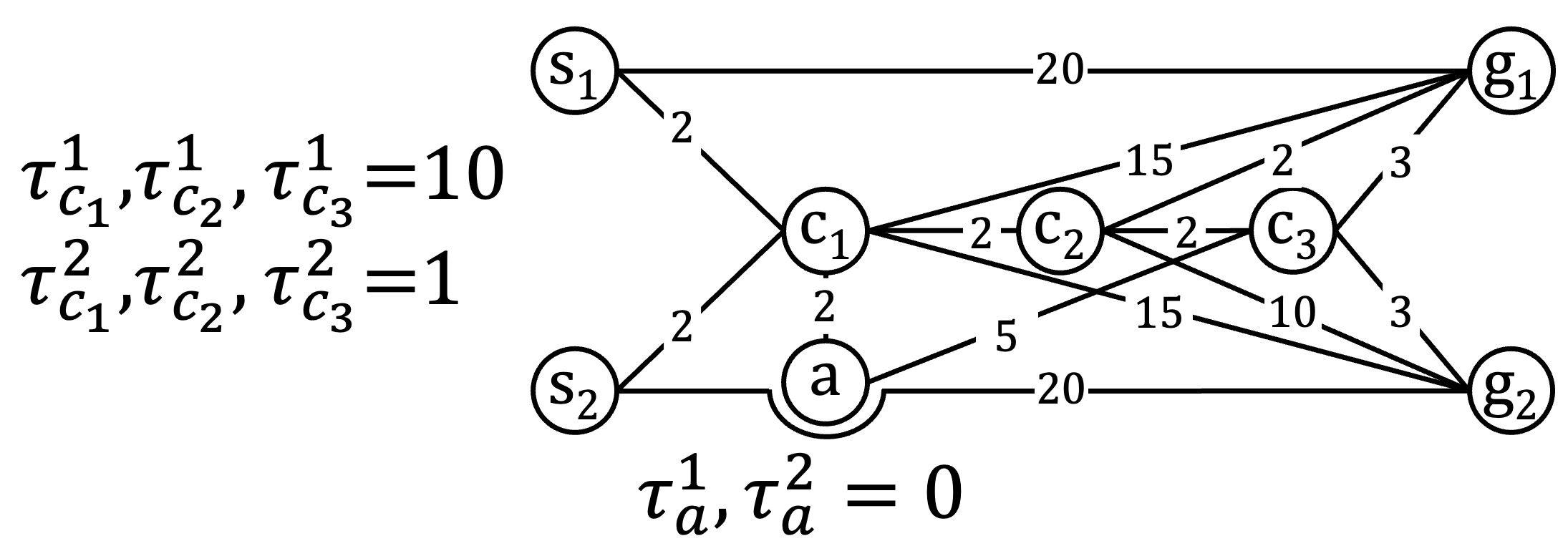}
    \caption{A PNE exists from \({c_1}\) to \({c_3}\), but not when both agents follow \(SCP_{c_1,c_3}\), since \(a_1\) prefers to deviate at \({c_2}\). Equilibrium is obtained when both instead follow \({c_1}, a, {c_3}\).}
    \label{fig:suboptimal}
\end{figure}
Therefore we make the following vital definition: 
\begin{definition}[Stable Cooperation Partial Path]\label{def:stable-cooperation}
A cooperation partial path $\pi_{c_s,c_d}$ between cooperation nodes $c_s, c_d \in V_C$ is a {\em stable cooperation partial path} if the departure node ${c_d}$ is the optimal departure node for both agents along the path.
Formally, let $v_{d^*}^i(\pi_{c_s,c_d})$ denote the optimal departure node for agent $a_i$ along 
$\pi_{c_s,c_d} \in \Pi_{c_s,c_d}$, defined as
\[
v_{d^*}^i(\pi_{c_s,c_d})=\arg\min_{c \in \pi_{c_s,c_d}}T\!\left(\pi_{c_s,c} \circ SIP_{c,g_i}\;\middle|\;\pi_{c_s,c_d} \circ SIP_{c_d,g_{-i}}\right)
\]
then,
$c_d = v_{d^*}^1(\pi_{c_s,c_d}) = v_{d^*}^2(\pi_{c_s,c_d})$.

\end{definition}
This definition implies that once cooperation begins at node ${c_s}$, both agents would prefer to continue cooperating along $\pi_{c_s,c_d}$ rather than deviating earlier towards their respective target nodes. 
The set of all stable cooperation partial paths between two nodes ${v}$ and ${u}$ is denoted by $\Pi^S_{v,u}$, and the {\em shortest} stable cooperation partial path between these nodes is denoted by $SCP^S_{v,u}$.

The stability of a partial path between two cooperation nodes $c_1$ and $c_2$ depends solely on the path times from $c_1$ onward and is independent of the path segment leading to $c_1$. 
Leveraging this property, the shortest stable cooperation partial paths ending at a cooperation node $c \in V_C$
can be computed efficiently via a backward variant of Dijkstra’s algorithm (Algorithm~\ref{algorithm:shortest-stable-partial-paths})
which explores the graph backward from $c$ while excluding nodes that violate stability.

Algorithm~\ref{algorithm:shortest-stable-partial-paths} initializes all paths to \(c_d\) as empty with infinite time, except \(c_d\) itself, which is set to time \(0\) with a trivial path, and defines \(Q \gets V\) as the set of unvisited nodes [lines 1–3]. The algorithm then evaluates all unvisited nodes \(v\) that can reach \(c_d\) in ascending order of their distance to \(c_d\) [lines 4-10]. 
For each node \(v\), only its neighbors that preserve path stability are considered as its neighbors [lines 7-8]. If the shortest stable cooperation path from a neighbor \(u\) of \(v\) via \(v\) to $c_d$ is shorter than its current path, the algorithm updates the path and its associated time [lines 9-11]. 
Once all reachable nodes have been evaluated, the algorithm returns a mapping of each node in the graph to its corresponding \emph{shortest} stable cooperation partial path to \(c_d\) [line 12].
\begin{algorithm}[ht] 
\caption{\textsc{Shortest Stable Paths($G, V_C, c_d, SIP_{g_1}, SIP_{g_2}$)}}
\label{algorithm:shortest-stable-partial-paths}
\begin{algorithmic}[1]
\State $T_{v,c_d} \gets \infty, \pi_{v,c_d} \gets \emptyset$ for all $v \in V$
\State $T_{c_d,c_d} \gets 0, \pi_{c_d,c_d} \gets c_d$
\State $Q \gets V$
\While{$Q$ has a node $v$ s.t. $T_{v,c_d}<\infty$}
    \State $v \gets$ node in $Q$ with smallest $T_{v,c_d}$
    \State Remove $v$ from $Q$
       
    \For{each neighbor $u$ of $v$ s.t. $u \in Q$ and $\tau_{u,v}+\tau^2_{v}+T_{v,c_d} + SIP_{c_d,g_1} \leq SIP_{u,g_1}$
   \State \hspace{12.5em} and
   $\tau_{u,v}+\tau^2_{v}+T_{v,c_d} + SIP_{c_d,g_2} \leq SIP_{u,g_2}$}
        \If{$T{v,c_d} + \tau^2_{v} + \tau_{u,v} < T{u,c_d}$}
            \State $T{u,c_d} \gets T{v,c_d} + \tau^2_{v} + \tau_{u,v}$
            \State $\pi_{u,c_d} \gets u \circ \pi_{v,c_d}$
        \EndIf
    \EndFor
\EndWhile
\State \textbf{Return} a dictionary from each $v \in V$ to $\pi_{v,c_d}$
\end{algorithmic}
\end{algorithm}
\begin{lemma}\label{lemma:stable-cooperation-algorithm}
    Given a cooperation ending node $c_d$, Algorithm~\ref{algorithm:shortest-stable-partial-paths} computes the \textbf{shortest} \emph{Stable Cooperation Partial Path} concluding cooperation at $c_d$ from every node in the graph in polynomial time.
\end{lemma}

Similar to Lemma \ref{lemma:non-cooperative-algorithm}, the correctness of Lemma \ref{lemma:stable-cooperation-algorithm} follows from the optimality of Dijkstra's algorithm, with one adjustment. Since the neighbors of a node are filtered to maintain stability (line 7), we must also ensure that when a node $v$ is pulled from $Q$, its set of stable neighbors, as filtered in line 7, is final and reflects all and only the stable neighbors of $v$. Specifically, we need to ensure that once node $v$ is removed from $Q$, for all its neighbors $u \in V$ such that $(v,u) \in E$, the value $\tau_{u,v} + \tau^2_{v} + T_{v, c_d}$ is fixed and will not change from that point on. 
The only element in this value that can change during the algorithm's execution is $T_{v, c_d}$, but once $v$ is pulled from $Q$, its value remains fixed.
Thus, the path times and stability conditions are correctly computed and finalized for each node upon removal from $Q$, ensuring the overall optimality of the algorithm.
Similar to Dijkstra's algorithm, the complexity of Algorithm \ref{algorithm:shortest-stable-partial-paths} is $\mathcal{O}(|E| + |V| \log |V|)$.\\

Following Definition~\ref{def:stable-cooperation}, to eliminate unilateral deviations within the cooperation segment, this segment must form Stable Cooperation Partial Path.

\subsection{Departure segment}
In this segment, since cooperation ceases and both agents aim to reach their respective targets as quickly as possible, each agent will necessarily follow its shortest independent path towards its respective target: $SIP_{c_d, g_1}$ and $SIP_{c_d, g_2}$. However, if the shortest independent path of one agent, (w.l.o.g) $a_1$,  includes a cooperation node ${c_d'}$ that serves as a better departure point for $a_2$ than ${c_d}$, then $a_2$ would prefer to continue cooperating along $SIP_{c_d,c_d'}$ instead of taking its independent shortest path directly to its target. This violates the Nash equilibrium property.  
Thus, the cooperation departure node ${c_d}$ must also be the optimal exit node for both agents along the departure segment $(SIP_{c_d, g_1},SIP_{c_d, g_2})$.

\subsection{Full Path Equilibrium Cooperation}
To integrate the above findings into a full-path structural characterization of cooperative joint strategies, we identify five conditions that a cooperative joint strategy must satisfy, proven to be both necessary and sufficient for it to constitute a Nash equilibrium. This characterization applies only to joint strategies that involve cooperation.
The \emph{independent joint strategy} \((SIP_{s_1,g_1}, SIP_{s_2,g_2})\), which is the optimal joint strategy involving no cooperation, may also constitute a PNE and is treated separately in Section~\ref{section:independentPNE}.

\begin{definition}[Equilibrium Cooperation Joint Strategy]
\label{def:equilibrium_cooperation_joint_strategy}
A cooperation joint strategy \((\pi^1,\pi^2)\) is called an
\emph{Equilibrium Cooperation Joint Strategy} (ECJS) if it satisfies the following conditions:
    \begin{enumerate}
        \item \label{item:cond-stable} The agents cooperate along exactly one cooperation segment, $\pi^S_{c_s, c_d}\in \Pi^S_{c_s, c_d}$, which is stable.
       \item \label{item:cond-non-coop} Both agents reach the cooperation start node ${c_s}$ via
    Mutually-Robust non-cooperative partial paths
    $\pi^{1(NC)}_{s_1,c_s} \in \Pi^{NC}_{s_1,c_s}$ and
    $\pi^{2(NC)}_{s_2,c_s} \in \Pi^{NC}_{s_2,c_s}$.
        \item \label{item:cond-departure} The cooperation end node ${c_d}$ is jointly optimal for departure:
    ${c_d} = v^{1}_{d^*}(\pi^2_{c_s,g_2}) = v^{2}_{d^*}(\pi^1_{c_s,g_1})$.
\item \label{item:cond-shortest-departure}
From ${c_d}$ onward, both agents follow their shortest independent
    paths to their targets, $SIP_{c_d,g_1}$ and $SIP_{c_d,g_2}$.
    \item \label{item:cond-cooperation-over-independent}
    Cooperation weakly dominates the shortest independent path for both agents:
    $
    T(\pi^1 \mid \pi^2) \le T(SIP_{s_1,g_1}), \quad$
    $T(\pi^2 \mid \pi^1) \le T(SIP_{s_2,g_2})$.
    \end{enumerate}
\end{definition}
Consequently, $(\pi^1, \pi^2)$ can be formally expressed as follows:
\[
\pi^1 = \pi^{1(NC)}_{s_1, c_s} \circ \pi^S_{c_s, c_d} \circ SIP_{c_d, g_1}, \quad 
\pi^2 = \pi^{2(NC)}_{s_2, c_s} \circ \pi^S_{c_s, c_d} \circ SIP_{c_d, g_2}
\]

Condition \ref{item:cond-stable} ensures the validity of the Cooperation segment, while Condition \ref{item:cond-non-coop} guarantees the validity of the Joining segment. Conditions \ref{item:cond-departure} and \ref{item:cond-shortest-departure} verify the correctness of the Departure segment. Finally, Condition \ref{item:cond-cooperation-over-independent} ensures that the entire path is advantageous for both agents, confirming the overall benefit of the cooperation.
We denote the set of all ECJS by \(\mathbb{ECJS}\). Let \(\mathbb{ECJS}^{c_d}\) denote the set of ECJS that terminate cooperation at \(c_d \in V_C\), and let \(\mathbb{ECJS}^{c_s,c_d}\) denote the set of ECJS that initiate cooperation at \(c_s \in V_C\) and terminate it at \(c_d\).

\begin{theorem}\label{theorem:pne_equivalance}
A cooperation joint strategy constitutes a PNE\footnote{We assume that when faced with two strategies yielding the same path time, an agent will prefer the one that starts cooperation earlier and involves a single continuous cooperation segment.} if and only if it is an ECJS.
\end{theorem}

The equivalence is proved formally in Appendix~\ref{app:theorem:pne_equivalance}. We outline the main argument below.
\begin{proofsketch}
\textbf{($\Leftarrow$)} Assume that \((\pi^1,\pi^2)\) is an ECJS. By definition, the joint strategy satisfies the structural conditions ensuring that each agent minimizes its path time given the other agent's strategy. Hence, no agent can strictly improve by unilaterally deviating, and \((\pi^1,\pi^2)\) constitutes a PNE.

\textbf{($\Rightarrow$)} Assume that \((\pi^1,\pi^2) \ne (SIP_{s_1,g_1}, SIP_{s_2,g_2})\) is a cooperation PNE. We show that any violation of the ECJS conditions induces a profitable unilateral deviation for at least one agent, contradicting the PNE property. Therefore, all ECJS conditions must hold, and \((\pi^1,\pi^2)\) is an ECJS.
\end{proofsketch}

Further analysis of the ECJS structure yields the insight:
\subsubsection{Pareto Sub-optimality}
    The optimal ECJS in $\mathbb{ECJS}^{c_s,c_d}$ is not necessarily Pareto optimal. Specifically, since $T(SIP^{NC}_{s_i,c_s}) \geq T(SIP_{s_i,c_s})$ and $T(SCP^S_{c_s,c_d} \mid SCP^S_{c_s,c_d}) \geq T(SCP_{c_s,c_d} \mid SCP_{c_s,c_d})$, the full-path joint strategy
    $
    \left( SIP^{NC}_{s_1,c_s} \circ SCP^S_{c_s,c_d} \circ SIP_{c_d,g_1}, SIP^{NC}_{s_2,c_s} \circ SCP^S_{c_s,c_d} \circ SIP_{c_d,g_2} \right)
    $
    may be suboptimal compared to
$
(SIP_{s_1,c_s} \circ SCP_{c_s,c_d} \circ SIP_{c_d,g_1},\ SIP_{s_2,c_s} \circ SCP_{c_s,c_d} \circ SIP_{c_d,g_2}),
$
indicating that the optimal joint strategy in $\mathbb{ECJS}^{c_d}$ is not necessarily Pareto optimal.
Figure~\ref{fig:suboptimal} illustrates such a scenario.
In this example, the game presents three PNE: $((s_1, g_1), (s_2, g_2))$,\\
$((s_1, c_1, c_2, g_1)$, $(s_2, c_1, c_2, g_2))$, and $((s_1, c_1, a, c_3, g_1), (s_2, c_1, a, c_3, g_2))$. Among the three PNE, the joint strategy $((s_1, c_1, a, c_3, g_1), (s_2, c_1, a, c_3, g_2))$ minimizes the path time of $a_2$, though it remains suboptimal compared to the joint strategy where both agents follow the shortest path between $v_{c_1}$ and $v_{c_3}$: $((s_1, c_1, c_2, c_3, g_1), (s_2, c_1, c_2, c_3, g_2))$.

\subsection{Equilibrium in Independent Joint Strategies}\label{section:independentPNE}
A joint strategy with no cooperation can constitute a PNE only if both agents follow their shortest independent paths and neither agent has a profitable cooperation-involving deviation.
As shown in Section~\ref{sec:coop-brp}, an agent's best response to a fixed strategy of the other agent can be computed in polynomial time. Comparing this response with the agent's shortest independent path allows us to determine whether a profitable cooperation-involving deviation exists, or whether the shortest independent path is itself the agent's  best response. Applying this test to both agents therefore enables us to efficiently verify whether
$
(SIP_{s_1,g_1}, SIP_{s_2,g_2})
$
constitutes a PNE.

\subsection{PNE Joint Strategy Set}
We show that the number of ECJS that must be considered is at most linear in $m = |V_C|$, and that this set is guaranteed to be non-empty.

\subsubsection{Linear Upper Bound}
Since all joint strategies in $\mathbb{ECJS}^{c_d}$ involve the simultaneous departure of the agents from ${c_d}$, and their path times from $c_d$ to their respective target nodes remain constant, there exists one (or a few equivalent) joint strategy in $\mathbb{ECJS}^{c_d}$ that arrives earliest at $c_d$ and dominates all others. Accordingly, our approach seeks to identify this optimal joint strategy for each cooperation node $c_d \in V_C$, thereby reducing the number of relevant cooperation paths to be linear in $m$.
This task reduces to identifying the optimal cooperation starting node $c_{s^*}$ and the joint strategy in $\mathbb{ECJS}^{c_{s^*},c_d}$ that minimizes the arrival time at $c_d$. For a given cooperation starting node $c_s$, the optimal joint strategy in $\mathbb{ECJS}^{c_{s},c_d}$ is obtained when both agents follow their shortest non-cooperative partial paths to $c_s$ ($SIP^{NC}_{s_1,c_s}$, $SIP^{NC}_{s_2,c_s}$), with arrival times synchronized to ensure mutual robustness, and then jointly follow the shortest stable cooperation partial path $SCP^S_{c_s,c_d}$.

Consequently, if $\mathbb{ECJS}^{c_s,c_d}\neq\emptyset$, its optimal element is $(SIP^{NC}_{s_1,c_s} \circ SCP^S_{c_s,c_d} \circ SIP_{c_d,g_1}, SIP^{NC}_{s_2,c_s} \circ SCP^S_{c_s,c_d} \circ SIP_{c_d,g_2})$.
Finding the optimal joint strategy in $\mathrm{ECJS}_{c_d}$ therefore reduces to finding the optimal cooperation starting node ${c_s^*}$ minimizing the arrival time at ${c_d}$:
\begin{equation}
\begin{split}
    {c_s^*} = \arg \min_{{c_s} \in V_C} \big\{ & \max\big(T(SIP^{NC}_{s_1,c_s}), T(SIP^{NC}_{s_2,c_s})\big) \\
    & + \tau^2_{c_s} + T(SCP^S_{c_s,c_d} \mid SCP^S_{c_s,c_d}) \big\}
\end{split}
\label{eq:armin}
\end{equation}
In Section \ref{section:algorithms}, we use this characterization to introduce an efficient algorithm to map all ECJS an agent must consider when determining its strategy by examining all cooperation nodes $c \in V_C$, treating each as a potential cooperation \textbf{ending} node and identifying the optimal ECJS concluding cooperation at that node. 
However, not all cooperation nodes need to be considered in this process.  
The following lemma, proved in Appendix~\ref{app:lem:pruning}, shows that the concatenation of two stable cooperation paths remains stable, and that the resulting longer stable path necessarily dominates its shorter prefix. Consequently, partial cooperation paths can be ignored.
\begin{lemma} \label{lemma:strategy-pruning-claim}
Let \(c_{d_1},c_{d_2} \in V_C\) be two distinct cooperation nodes. 
Suppose there exists a stable cooperation partial path from \(c_{d_1}\) to \(c_{d_2}\), denoted 
\(\pi^S_{c_{d_1},c_{d_2}}\). 
Then, for any stable cooperation partial path \(\pi_{c_s,c_{d_1}}\) from \(c_s\) to \(c_{d_1}\), the concatenated path
\[
\pi'_{c_s,c_{d_2}}
=
\pi_{c_s,c_{d_1}} \circ \pi^S_{c_{d_1},c_{d_2}}
\]
is also stable. Moreover, the optimal ECJS whose cooperation segment is \(\pi'_{c_s,c_{d_2}}\) dominates the optimal ECJS whose cooperation segment is \(\pi_{c_s,c_{d_1}}\). That is,
\begin{align*}
\forall i \in \{1,2\}, \quad
&
T\!\left(
    \pi_{c_s,c_{d_1}} \circ SIP_{c_{d_1},g_i}
    \mid
    \pi_{c_s,c_{d_1}} \circ SIP_{c_{d_1},g_{-i}}
\right)
\\
&\geq
T\!\left(
    \pi'_{c_s,c_{d_2}} \circ SIP_{c_{d_2},g_i}
    \mid
    \pi'_{c_s,c_{d_2}} \circ SIP_{c_{d_2},g_{-i}}
\right).
\end{align*}
\end{lemma}

The following corollary follows directly from Lemma~\ref{lemma:strategy-pruning-claim}. Since the segment $SIP_{c_d,c_{d'}}$ is stable, appending it to the original cooperation segment preserves stability. Moreover, this extension does not violate Conditions~\ref{item:cond-stable}, \ref{item:cond-non-coop}, \ref{item:cond-shortest-departure}, or \ref{item:cond-cooperation-over-independent} in Definition~\ref{def:equilibrium_cooperation_joint_strategy}.
\begin{corollary}\label{cor:dominated-strategies}
    Given a cooperation joint strategy 
    \[
    (\pi^1_{s_1,c_s} \circ \pi^S_{c_s,c_d} \circ SIP_{c_d,g_1}, \pi^2_{s_2,c_s} \circ \pi^S_{c_s,c_d} \circ SIP_{c_d,g_2})
    \]
    that satisfies all conditions of Definition \ref{def:equilibrium_cooperation_joint_strategy} except Condition \ref{item:cond-departure}, i.e., for some agent $a_i$, the optimal departure node along $\pi^S_{c_s,c_d} \circ SIP_{c_d,g_{-i}}$ is:
    \[
    {c_{d'}} = v_{d^*}^i(\pi^S_{c_s,c_d} \circ SIP_{c_d,g_{-i}}) \neq {c_d}
    \]
    then the joint strategy
    \[
    (\pi^1_{s_1,c_s} \circ \pi^S_{c_s,c_d} \circ SIP_{c_d,c_{d'}} \circ SIP_{c_{d'},g_1}, \pi^2_{s_2,c_s} \circ \pi^S_{c_s,c_d} \circ SIP_{c_d,c_{d'}} \circ SIP_{c_{d'},g_2})
    \]
    satisfies the same conditions of Definition \ref{def:equilibrium_cooperation_joint_strategy} and dominates the original joint strategy.
\end{corollary}


\subsubsection{Existence Guarantee}

To conclude the insights on the size of the ECJS set we show that in any scenario, there always exists a joint strategy that constitutes a PNE. If the independent joint strategy $(SIP_{s_1,g_1}, SIP_{s_2,g_2})$ is already a PNE, the claim holds.  Otherwise, one of the agents necessarily has an incentive to deviate from its independent shortest path in favor of cooperation with the other agent.
In this case, we show that the resulting cooperative path is stable and can serve as the basis for constructing a Pure Nash Equilibrium.

In Appendix~\ref{app:pne-existance}, we prove the following Lemma, showing that if there exists a cooperative joint strategy with a stable cooperation segment that both agents prefer over the independent strategy, then a PNE exists.
\begin{lemma}\label{lemma:PNE-existance}
    Consider a cooperative joint strategy of the form:
    \[
    \Big(SIP_{s_1,c_s} \circ \pi^S_{c_s,c_d} \circ SIP_{c_d,g_1},\ 
    SIP_{s_{2},c_s} \circ \pi^S_{c_s,c_d} \circ SIP_{c_d,g_{2}}\Big),
    \]
    where ${c_s}$ is the first cooperation node and ${c_d}$ is the last cooperation node in which the agents cooperate, and $\pi^S_{c_s,c_d} \in \Pi_{c_s,c_d}$ is a stable partial path. If the following condition holds:
    \[
    \forall i \in \{1,2\}, \quad
    T\Big(SIP_{s_i,c_s} \circ \pi^S_{c_s,c_d} \circ SIP_{c_d,g_i}|\ 
    SIP_{s_{-i},c_s} \circ \pi^S_{c_s,c_d} \circ SIP_{c_d,g_{-i}}\Big) 
    \leq T\Big(SIP_{s_i,g_i}\Big) 
    \]
    then a Pure Nash Equilibrium (PNE) exists.
\end{lemma}
Leveraging Lemma \ref{lemma:PNE-existance}, we establish the existence of a PNE in any scenario.\\
\begin{theorem}   
\label{theorem:PNE-existance}
\ictpp always admits a Pure Nash Equilibrium
\end{theorem}
\begin{proof}
Consider the \emph{Independent Joint Strategy} \( (SIP_{s_1,g_1}, SIP_{s_2,g_2}) \). If this strategy is already a PNE, the claim follows directly. Otherwise, at least one agent, without loss of generality, \( a_1 \), has an incentive to deviate and leverage cooperation with \( a_2 \).

By Theorem~\ref{cor1}, given that \( a_2 \) follows its independent shortest path \( SIP_{s_2,g_2} \), the optimal path for \( a_1 \) adheres the structure  
$
SIP_{s_1,c_s} \circ SIP_{c_s,c_d} \circ SIP_{c_d,g_1}
$,
where \( {c_s} \) is the first cooperation node along \( a_2 \)'s path that \( a_1 \) can reach in time for cooperation, and \( {c_d} \) is its optimal departure node.  
Since \( {c_d} \) lies on the shortest independent path from \( {c_s} \) to \( {g_2} \), it is also \( a_2 \)'s optimal departure node along \( SIP_{c_s,c_d} \). Thus, \( SIP_{c_s,c_d} \) forms a \emph{Stable Cooperation Partial Path}.  

Since this strategy is $a_1$'s best response to $a_2$ shortest independent path, we have:
\[
T(SIP_{s_1,c_s} \circ SIP_{c_s,c_d} \circ SIP_{c_d,g_1} \mid SIP_{s_{2},c_s} \circ SIP_{c_s,c_d} \circ SIP_{c_d,g_{2}}) \leq T(SIP_{s_1,g_1})
\]
Additionally, since it involves cooperation along $a_2$'s shortest independent path, it also improves \( a_2 \)'s arrival time at its target node:
\[
T(SIP_{s_2,c_s} \circ SIP_{c_s,c_d} \circ SIP_{c_d,g_2} \mid SIP_{s_{1},c_s} \circ SIP_{c_s,c_d} \circ SIP_{c_d,g_{1}}) \leq T(SIP_{s_2,g_2})
\]
By Lemma~\ref{lemma:PNE-existance}, this implies the existence of a PNE.
\end{proof}

\section{Mapping The Set of Relevant ECJS}
\label{section:algorithms}
The structural results presented above restrict the equilibrium search to at most one candidate per cooperation departure node. Based on this observation, Algorithm~\ref{algorithm:optimal-equilibrium-joint-strategy} computes the optimal stable joint strategy ending cooperation at a given node, while Algorithm~\ref{algorithm:map-all-equilibrium-joint-strategy} enumerates all {\em non-dominated} cooperative PNEs by iterating over all possible departure nodes.
The restriction to non-dominated strategies is important. By Corollary~\ref{cor:dominated-strategies}, if cooperation can be beneficially extended along the departure segment, then the extended cooperation joint strategy dominates the shorter one. In such cases, the shorter strategy fails to constitute a PNE, since an agent would prefer to continue cooperating rather than depart. Algorithm~\ref{algorithm:map-all-equilibrium-joint-strategy} therefore removes these dominated candidates and retains the extended strategies that dominate them, returning only the non-dominated cooperative PNEs.

\subsection{Joint Paths To Cooperation Departure Node}
Given a cooperation departure node ${c_d}$, Algorithm~\ref{algorithm:optimal-equilibrium-joint-strategy} computes the dominating stable cooperation joint strategy ending cooperation at that node. 
The algorithm first computes the shortest independent paths from each agent’s starting node to all nodes in the graph, as well as from all nodes to each agent’s respective target node [lines 1-4], It then computes the shortest stable partial paths that end cooperation at $c_d$ from all nodes in the graph using Algorithm~\ref{algorithm:shortest-stable-partial-paths} [lines 5], and the shortest non-cooperative partial paths from each agent’s starting node to all nodes in the graph using Algorithm~\ref{algorithm:shortest-non-cooperative-partial-paths} [lines 6-7].

Next, for each cooperation node $c_s \in V_C$, the algorithm evaluates the agents’ arrival times at $c_d$ under the joint strategy
\[
\big(
SIP^{NC}_{s_1,c_s} \circ SCP^S_{c_s,c_d} \circ SIP_{c_d,g_1},\;
SIP^{NC}_{s_2,c_s} \circ SCP^S_{c_s,c_d} \circ SIP_{c_d,g_2}
\big),
\]
and records the cooperation-starting node that minimizes this arrival time [lines 8-12].

Finally, if the resulting stable cooperative joint strategy ending cooperation at $c_d$ improves both agents’ performance relative to following their shortest independent paths to their respective target nodes, it is returned as the optimal stable joint strategy ending cooperation at $c_d$ [lines~13-14]. Otherwise, the algorithm returns \texttt{None} [line 15].

\begin{algorithm} 
\caption{\textsc{Optimal Stable Joint Strategy Ending Cooperation At A Given Node}($G, V_C, {c_d}, {s_1}, {s_2}, {g_1}, {g_2}$)}
\label{algorithm:optimal-equilibrium-joint-strategy}
\begin{algorithmic}[1]
    \State $SIP_{s_1}\! \gets\! \textsc{shortest paths from}(G, \tau^1, {s_1})$
    \State $SIP_{s_2}\! \gets\! \textsc{shortest paths from}(G, \tau^1, {s_2})$ 
    \State $SIP_{g_1}\! \gets\! \textsc{shortest paths to}(G, \tau^1, {g_1})$
    \State $SIP_{g_2}\! \gets\! \textsc{shortest paths to}(G, \tau^1, {g_2})$ 
    \State $SCP^S_{c_d}\! \gets\! \textsc{shortest stable paths}(G,\! V_C,\! {c_d},\! SIP_{g_1},\! SIP_{g_2})$

    \State $SIP^{NC}_{s_1}\! \gets\! \textsc{shortest nc paths}(G, V_C, {s_1}, SIP_{s_2})$
    \State $SIP^{NC}_{s_2}\! \gets\! \textsc{shortest nc paths}(G, V_C, {s_2}, SIP_{s_1})$
    \State $optNE\! \gets\! \infty,\quad {c_{s^*}} \gets None$
    \ForAll{ ${c_s} \in V_C$}      
            \If{$\max(SIP^{NC}_{s_1,c_s}, SIP^{NC}_{s_2,c_s}) + \tau^2_{c_s} + SCP^S_{c_s,c_d} \leq optNE$}
                \State $optNE \! \gets \! \max(SIP^{NC}_{s_1,c_s}, SIP^{NC}_{s_2,c_s}) + \tau^2_{c_s} + SCP^S_{c_s,c_d}$
                \State ${c_{s^*}} \gets {c_s}$
            \EndIf
    \EndFor
\If{$optNE + \tau^2_{c_d} + SIP_{c_d,g_1} \leq SIP_{s_1,g_1}$ and $optNE + \tau^2_{c_d} + SIP_{c_d,g_2} \leq SIP_{s_2,g_2}$}
    \State \Return $(SIP^{NC}_{s_1,c_{s^*}} \circ SCP^S_{c_{s^*},c_d} \circ SIP_{c_d,g_1}, SIP^{NC}_{s_2,c_{s^*}} \circ SCP^S_{c_{s^*},c_d} \circ SIP_{c_d,g_2})$
\EndIf
    \State \Return $None$
\end{algorithmic}
\end{algorithm}

\begin{lemma}\label{lemma:alg-limited-ecjs}
Algorithm~\ref{algorithm:optimal-equilibrium-joint-strategy} computes, in polynomial time, the optimal stable joint strategy ending cooperation at $c_d$.
\end{lemma}
The correctness of Lemma~\ref{lemma:alg-limited-ecjs} follows from the observation that, if there exists a stable joint strategy that ends cooperation at $c_d$, then the joint strategy
\[
(\pi^1, \pi^2) =
\big(
SIP^{NC}_{s_1,c_s} \circ SCP^S_{c_s,c_d} \circ SIP_{c_d,g_1},\;
SIP^{NC}_{s_2,c_s} \circ SCP^S_{c_s,c_d} \circ SIP_{c_d,g_2}
\big),
\]
where
\[
{c_s}
=
\arg \min_{{c} \in V_C}
\left\{
\max\!\big(
T(SIP^{NC}_{s_1,c}),
T(SIP^{NC}_{s_2,c})
\big)
+ \tau^2_{c}
+ T\!\left(SCP^S_{c,c_d} \mid SCP^S_{c,c_d}\right)
\right\},
\]
dominates all other stable joint strategies that end cooperation at $c_d$.

The algorithm iterates over all cooperation nodes, evaluating each as a potential starting node and selecting the one that minimizes the arrival time at ${c_d}$.  

\paragraph{Complexity Analysis}  
The algorithm begins with four executions of Dijkstra's algorithm to determine the shortest paths from ${s_1}$, ${s_2}$, ${g_1}$, and ${g_2}$. Each execution of Dijkstra's algorithm has a time complexity of $\mathcal{O}(|E| + |V| \log |V|)$. Consequently, the total complexity for these four runs is:
\[
\mathcal{O}(4 \cdot (|E| + |V| \log |V|)) = \mathcal{O}(|E| + |V| \log |V|)
\]
Next, the algorithm runs Algorithm \ref{algorithm:shortest-stable-partial-paths}, which operates similarly to Dijkstra's algorithm, with a complexity of:
\[
\mathcal{O}(|E| + |V| \log |V|)
\]
Subsequently, the algorithm executes Algorithm \ref{algorithm:shortest-non-cooperative-partial-paths} twice, each with the same complexity:
\[
2 \cdot \mathcal{O}(|E| + |V| \log |V|) = \mathcal{O}(|E| + |V| \log |V|)
\]
Finally, the algorithm iterates over all cooperation nodes in the graph to find the one that optimizes the Nash equilibrium overall path time. This iteration has a complexity of:
\[
\mathcal{O}(|V|)
\]
Combining all these components, the total complexity of the algorithm can be expressed as:
\[
\mathcal{O}(|E| + |V| \log |V| + |V|)
\]
Which concludes to:
\[
\mathcal{O}(|E| + |V| \log |V|)
\]


\subsection{Optimal ECJS Computation}
To map all non-dominated joint strategies in $\mathbb{ECJS}$, we evaluate each cooperation node as a potential cooperation departure node, retaining only the dominating ECJS and discarding dominated concatenations of stable paths.
\begin{algorithm}
\caption{\textsc{Optimal ECJS}($G, V_C, {s_1}, {s_2}, {g_1}, {g_2}$)}\label{algorithm:map-all-equilibrium-joint-strategy}
\begin{algorithmic}[1]
 \State $ECJSmap \gets \{(SIP_{s_1,g_1},SIP_{s_2,g_2})\}, \quad dominated \gets \emptyset$
    \ForAll{${c_d} \in V_C$}
        \If{${c_d} \notin dominated$}
            \State $ECJSmap[{c_d}] \gets \textsc{algorithm \ref{algorithm:optimal-equilibrium-joint-strategy}}(G, V_C, {c_d}, \dots)$
            \State $dominated \gets dominated \cup \textsc{shortest stable paths}(G, V_C, {c_d}, SIP_{g_1}, SIP_{g_2})$
        \EndIf
    \EndFor
    \State $S_1 \gets \textsc{Best Response Path})(G,V_C,{s_1},{g_1},{s_2},{g_2},SIP_{s_2,g_2})$
    \State $S_2 \gets \textsc{Best Response Path})(G,V_C,{s_2},{g_2},{s_1},{g_1},SIP_{s_1,g_1})$
    \If{$S_1 \neq SIP_{s_1,g_1}$ or $S_2 \neq SIP_{s_2,g_2}$}
        \State $ECJSmap \gets ECJSmap \setminus \{(SIP_{s_1,g_1},SIP_{s_2,g_2})\}$
    \EndIf
    \State \Return $ECJSmap \setminus dominated$
\end{algorithmic}
\end{algorithm}
Algorithm \ref{algorithm:map-all-equilibrium-joint-strategy} begins by initializing a dictionary that maps each cooperation node in the graph to its optimal stable joint strategy ending cooperation at that node, using Algorithm~\ref{algorithm:optimal-equilibrium-joint-strategy}.  
The dictionary is initialized with the independent shortest-paths joint strategy, and an empty set is initialized to track cooperation nodes that admit a stable cooperation partial path to another cooperation node [line~1].
The algorithm then iterates over all cooperation nodes ${c_d} \in V_C$, excluding nodes that already have a stable cooperation path to another cooperation node.  
For each node, it determines the optimal stable joint strategy ending cooperation at that node using Algorithm~\ref{algorithm:optimal-equilibrium-joint-strategy}, and updates the list of dominated nodes with those that admit a stable cooperation partial path to $c_d$ [lines 2-5].  Next, the algorithm verifies whether the independent shortest paths joint strategy constitutes a pure Nash equilibrium (PNE).  
If it does not, the strategy is removed [lines 6-9].  
Finally, the algorithm removes all joint strategies that terminate cooperation at nodes dominated by others (i.e., nodes with stable cooperation partial paths to another node)  
and returns the set of all non-dominated ECJS 
[line 10].

\begin{theorem}\label{theorem:map-correctness}
Algorithm \ref{algorithm:map-all-equilibrium-joint-strategy} returns a set of joint strategies satisfying the following properties:  
\begin{enumerate}
    \item \textbf{Nash Equilibrium Guarantee:} Every joint strategy $(\pi^1,\pi^2)$ returned by the algorithm is a PNE.
    \item \textbf{Dominance:} For every joint strategy \((\pi^1,\pi^2) \in \mathbb{ECJS}\), the algorithm returns either \((\pi^1,\pi^2)\) itself or a joint strategy \((\pi'^1,\pi'^2)\) that dominates it. Thus, the algorithm returns all non-dominated joint strategies in \(\mathbb{ECJS}\).
\end{enumerate}
\end{theorem}
\begin{proof}

    \textbf{Nash Equilibrium Guarantee:}\\
    Every joint strategy $(\pi^1,\pi^2)$ that Algorithm \ref{algorithm:map-all-equilibrium-joint-strategy} returns is either the shortest independent joint strategy or a cooperation joint strategy returned by Algorithm~\ref{algorithm:optimal-equilibrium-joint-strategy}.  
    If $(\pi^1,\pi^2) = (SIP_{s_1,g_1}, SIP_{s_2,g_2})$, the algorithm explicitly verifies that it is a PNE in lines 7-10.  
    Otherwise, if $(\pi^1,\pi^2) \neq (SIP_{s_1,g_1}, SIP_{s_2,g_2})$ and is returned by Algorithm~\ref{algorithm:optimal-equilibrium-joint-strategy}, it may fail to be a PNE only if the shortest independent path from its cooperation ending node ${c_d}$ to one of the agents' target nodes ${g_i}$ contains a stable cooperation partial path, violating Condition~\ref{item:cond-departure} for all cooperation joint strategies that end cooperation at ${c_d}$. 
    However, since the algorithm returns only joint strategies that end cooperation at nodes without stable cooperation partial paths to other nodes in the graph, this scenario cannot occur. Consequently, all conditions of Definition \ref{def:equilibrium_cooperation_joint_strategy} are satisfied, implying that $(\pi^1,\pi^2) \in ECJS = PNE$.\\
    \textbf{Dominance:}\\
      If $(\pi^1,\pi^2) = (SIP_{s_1,g_1},SIP_{s_2,g_2})$, then the algorithm implicitly adds this joint strategy to its result set. 
      Otherwise, if $(\pi^1,\pi^2) \neq (SIP_{s_1,g_1},SIP_{s_2,g_2})$, then by Theorem \ref{theorem:pne_equivalance}, $(\pi^1,\pi^2) \in \mathbb{ECJS}$ and can be represented as:
\[
\pi^1 = \pi^{1(NC)}_{s_1, c_s} \circ \pi^S_{c_s, c_d} \circ SIP_{c_d, g_1}, \quad 
\pi^2 = \pi^{2(NC)}_{s_2, c_s} \circ \pi^S_{c_s, c_d} \circ SIP_{c_d, g_2}
\]
where ${c_s}$ is the cooperation starting node and ${c_d}$ is the cooperation ending node.
This joint strategy is dominated by the optimal joint strategy in $\mathbb{ECJS}^{c_d}$:
\[
(\pi'^1, \pi'^2) = \left(SIP^{NC}_{s_1,c_s} \circ SCP^S_{c_s,c_d} \circ SIP_{c_d,g_1},\ SIP^{NC}_{s_2,c_s} \circ SCP^S_{c_s,c_d} \circ SIP_{c_d,g_2}\right)
\]
If there is no stable cooperation partial path from ${c_d}$ to any other cooperation node in the graph, the algorithm uses Algorithm \ref{algorithm:optimal-equilibrium-joint-strategy} to identify $(\pi'^1, \pi'^2)$, ensuring that $(\pi'^1, \pi'^2)$ is returned.
If, however, there exists a cooperation node ${c_{d'}}$ such that a stable cooperation partial path $\pi^S_{c_d,c_{d'}}$ exists  
(among multiple such nodes, we select ${c_{d'}}$ as one that has no stable cooperation partial path leading to any other cooperation node in the graph),  
then by Lemma~\ref{lemma:strategy-pruning-claim}, the following joint strategy:
\[
(\pi''^1, \pi''^2) = \left(SIP^{NC}_{s_1,c_s} \circ SCP^S_{c_s,c_d} \circ \pi^S_{c_d,c_{d'}} \circ SIP_{c_{d'},g_1},\  
SIP^{NC}_{s_2,c_s} \circ SCP^S_{c_s,c_d} \circ \pi^S_{c_d,c_{d'}} \circ SIP_{c_{d'},g_2}\right)
\]
dominates $(\pi'^1, \pi'^2)$.  
Since $(\pi''^1, \pi''^2) \in \mathbb{ECJS}^{c_{d'}}$ and the algorithm returns the optimal joint strategy in $\mathbb{ECJS}^{c_{d'}}$,  
which dominates all other joint strategies in $\mathbb{ECJS}^{c_{d'}}$, it follows that the returned strategy also dominates $(\pi''^1, \pi''^2)$,  
which in turn dominates $(\pi^1, \pi^2)$.\\
Consequently, for any given joint strategy $(\pi^1,\pi^2) \in PNE$, Algorithm~\ref{algorithm:map-all-equilibrium-joint-strategy} either returns $(\pi^1,\pi^2)$ or an equilibrium cooperation joint strategy $(\pi'^1,\pi'^2)$ that dominates it.
\end{proof}

\paragraph{Complexity Analysis.}
Algorithm \ref{algorithm:map-all-equilibrium-joint-strategy} iterates over all $m$ cooperation nodes, invoking
\textsc{shortest stable paths} and
Algorithm~\ref{algorithm:optimal-equilibrium-joint-strategy} for each, yielding 
complexity of $\mathcal{O}(m \cdot(|E| + |V| \log |V|))$.

\section{ECJS Selection} \label{section:strategy selection}
When multiple non-dominated ECJSs exist, trade-offs arise: 
one agent’s travel time improves only if the other’s worsens (see example in Figure~\ref{fig:correlation}).
However, since every ECJS improves upon independent shortest paths, agents share an interest in agreeing on a strategy.

\begin{figure}[ht]
    \centering
    \includegraphics[width=0.6\linewidth]{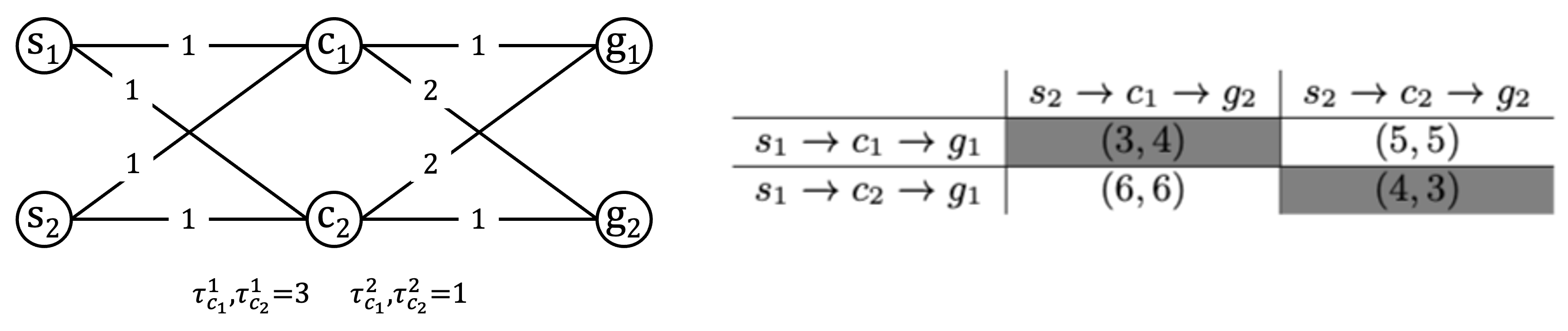}
    \caption{Two non-dominated ECJSs: cooperation at $c_1$ benefits $a_1$, while cooperation at $c_2$ benefits $a_2$; both outperform no cooperation.}
    \label{fig:correlation}
\end{figure}


This coordination problem can be modeled as a \emph{correlation game}
\cite{aumann1974subjectivity}, where agreement on a joint strategy is essential to
avoid suboptimal outcomes. Coordination can be achieved either through
\emph{conventions} \cite{lewis2008convention,vanderschraaf1995convention} (e.g.,
minimizing the maximum path time, maximizing minimum utility relative to the
$SIP$, or maximizing social welfare) or through the concept of a correlated equilibrium \cite{aumann1974subjectivity}.
In this setting, the agents agree on a shared probability distribution over the set of non-dominated ECJSs 
and use it to jointly select a strategy (e.g. via an external correlation device). This mechanism ensures consistent execution and prevents mismatched decisions, enabling the agents to maximize the expected value of a shared objective function, such as fairness or social welfare.

Since the joint strategies in $\mathbb{ECJS}$ inherently reflect conflicting preferences, reaching a mutually acceptable agreement can be challenging. To address this, we model the problem as a bargaining game \cite{nash1950bargaining}, aiming to identify a correlated equilibrium guided by standard bargaining solution concepts.
Specifically, we seek a solution that satisfies the following properties:\\
\begin{enumerate}
    \item \textbf{Pareto Optimality}: No agent's outcome can be improved without worsening the outcome of the other.
    \item \textbf{Symmetry}: Identical agents with symmetric options should receive identical outcomes.
    \item \textbf{Fairness}: The outcome should be impartial and just, avoiding favoritism or discrimination between agents.
\end{enumerate}
While \emph{Pareto optimality} and \emph{symmetry} are well defined, \emph{fairness} is more nuanced and context dependent. We therefore operationalize fairness by comparing four common bargaining solutions: Nash solution~\shortcite{nash1950bargaining}, Kalai-Smorodinsky solution~\shortcite{Other-Solutions}, and the egalitarian and utilitarian solutions.

Formally, we denote the set of non-dominated joint strategies in $\mathbb{ECJS}$ by $\mathbb{ECJS}^{ND}$, and represent the selection of a probability distribution over the joint strategies in $\mathbb{ECJS}^{ND}$ as a bargaining game $(S, d)$. The utility of each agent is defined by
the time saved compared to its $SIP$. That is, for a given joint strategy $(\pi^1,\pi^2)$ the utility of agent $a_i$ is defined as $u_i(\pi^1,\pi^2)=T(SIP_{s_i,g_i})-T(\pi^{i}|\pi^{-i})$.
We define the disagreement point \( d \) as:
\[
d = \left( u_1(SIP_{s_1,g_1}, SIP_{s_2,g_2}),\; u_2(SIP_{s_2,g_2}, SIP_{s_1,g_1}) \right)
\]
The set \( S \) of all possible outcomes in the bargaining game is defined as the set of expected utility pairs achievable by randomizing over joint strategies in \(\mathbb{ECJS}^{ND}\). 
Let $\mathbb{ECJS}^{ND}=\{(\pi_1^1,\pi_1^2),\ldots,(\pi_n^1,\pi_n^2)\}$. Then
\[
S=\Bigl\{\sum_{i=1}^n \alpha_i \cdot
\bigl(u_1(\pi_i^1,\pi_i^2),\,u_2(\pi_i^2,\pi_i^1)\bigr)
\;\Big|\;
\alpha_i\ge0,\ \sum_{i=1}^n\alpha_i=1
\Bigr\}
\]

where each $\alpha_i$ denotes the probability assigned to the $i$-th joint strategy in  $\mathbb{ECJS}^{ND}$, and the resulting pair represents the expected utilities of the two agents.
We determine the parameters $\alpha_1,\ldots,\alpha_n$ for each solution concept as follows. 
\begin{enumerate}
    \item \textbf{Nash Solution}, maximizes the product of the agents’ utilities relative to the disagreement point, capturing mutual benefit under rational cooperation.
    \item \textbf{Kalai-Smorodinsky Solution}, ensures proportional fairness by preserving each agent’s utility ratio relative to its maximum attainable utility.
    \item \textbf{Egalitarian Solution}, seeks to equalize outcomes by maximizing the minimum utility across agents.
    \item \textbf{Utilitarian Solution}, maximizes the overall social welfare.
\end{enumerate}

Since we seek a Pareto-optimal solution, we restrict our attention to the Pareto frontier of the set~\(S\), 
where no agent's expected utility can be improved without worsening the other's. As a result, the solution corresponds to a probability distribution supported on at most {\em two} joint strategies. The specific strategies involved, however, may vary depending on the chosen bargaining solution concept.

\section{Implementation and Results}
To provide a comprehensive perspective of \ictpp and understand how different selection methods affect individual and social outcomes, we developed a simulation tool and fully implemented all algorithms\footnote{The repositories for the simulation tool, algorithm implementation, experimental data, and results are publicly available at \url{https://iscmpp.info/}.}. We then evaluated \ictpp algorithms on randomly generated graphs with diverse topologies, travel times, and node delays, examining the factors that influence cooperation incentives across these varied settings.
To broaden the experimental evaluation, we also utilize the MAPF Benchmark Set~\cite{stern2019multi}, which provides a variety of grid-based maps. To adapt these maps to \ictpp, each grid cell is modeled as a node with two delay values, $\tau^1$ for single-agent and $\tau^2$ for cooperative execution, while edges between adjacent cells have a travel time of $1$ time unit.

To capture the factors that influence cooperation, we varied the following attributes of an \ictpp instance:
    \begin{itemize}
        \item \textbf{Cooperation Magnitude} --- The average ratio between the task execution time of a single agent and that of two agents cooperating at a cooperation node.
         \item \textbf{Cooperation Density} --- The ratio of cooperation nodes to all nodes in the graph.
         \item \textbf{Path Lengths} --- The minimum, over the two agents, of the shortest independent path time from the agent's start node to its target node.
        \item \textbf{Shortest Paths Divergence (SPD)} --- A path-alignment measure inspired by the Fréchet-distance view of comparing paths according to their ordered progression. SPD measures the minimum shortest independent path time between the two agents while they traverse their respective shortest independent paths.
        Lower SPD values indicate stronger alignment between the agents' independent paths, while higher values indicate that the paths remain more separated.
    \end{itemize}

   
    

We compare methods for joint strategy selection as follows. First, we compute the optimal social welfare for each scenario using a polynomial time algorithm (explained in details in Appendix~\ref{supp:social-welfare}). Then, using Algorithm~\ref{algorithm:map-all-equilibrium-joint-strategy}, we identify all non-dominated ECJS and evaluate the \emph{Price of Anarchy} (PoA) \cite{koutsoupias1999worst} and \emph{Price of Stability} (PoS) \cite{anshelevich2008price}, defined respectively as the ratios of the worst and best ECJS to the optimal social welfare. Finally, we compare individual path times and social welfare across ECJS selection methods under varying \emph{Cooperation Magnitude, Density, Path Length}, and \emph{SPD}.
We conducted four separate sets of experiments, each evaluating the effect of a different cooperation factor (Magnitude, Density, Path Length, and SPD) across a range of values.

Cooperation Density and Cooperation Magnitude are controlled directly by assigning cooperation nodes and their corresponding delay values. In contrast, Path Length and Shortest Paths Divergence (SPD) depend on the induced shortest paths between the agents’ start and target nodes, and therefore cannot be fixed directly. We control these factors indirectly using proxy measures based on the grid topology: the Manhattan distance between each agent’s start and target nodes for Path Length, and the Manhattan distance between the agents’ corresponding start and target nodes for SPD. The generated instances are then grouped according to their actual Path Length and SPD values, and the results are aggregated within each group.
For each configuration of Cooperation Density and Cooperation Magnitude, the algorithms were evaluated on 100 randomly sampled scenarios from the MAPF Benchmark Set. For SPD and Path Length, we evaluated approximately 2,500 and 3,500 scenarios, respectively, using varying configurations of the relevant Manhattan-distance proxy measures. Since these factors are controlled only indirectly, the resulting number of instances differs across intervals.
When evaluating the effect of a given attribute, the remaining attributes were fixed as follows: Cooperation Density was set to $0.7$, Cooperation Magnitude to $10$, the Manhattan distance between each agent's start and target nodes to $20$, corresponding to Path Length values of approximately $200$, and the Manhattan distance between the agents' corresponding start and target nodes to $3$, corresponding to SPD values of approximately $30$. These values were chosen because they produced a diverse set of PNEs, enabling a meaningful comparison among the different selection methods.
Appendix~\ref{appendix:tabular-data} reports the aggregated numerical values underlying all charts presented in this section.

Figure~\ref{fig:avg} reports the average number of PNEs across the four cooperation factors. The number of PNEs increases with the potential for cooperation, as reflected by higher Cooperation Density, Cooperation Magnitude, and Path Length, as well as by lower SPD values. For high SPD values, which correspond to highly separated paths, the average number of PNEs converges to one. For short paths, the number of potential PNEs is limited. As Path Length increases, more scenarios admit multiple PNEs. However, for very long paths, the number of PNEs decreases slightly again, possibly because longer paths introduce more opportunities for profitable deviations, making cooperative equilibria harder to sustain.
\begin{figure}[ht]
    \centering
    \begin{subfigure}[b]{0.23\linewidth}
        \centering
        \includegraphics[width=\linewidth]{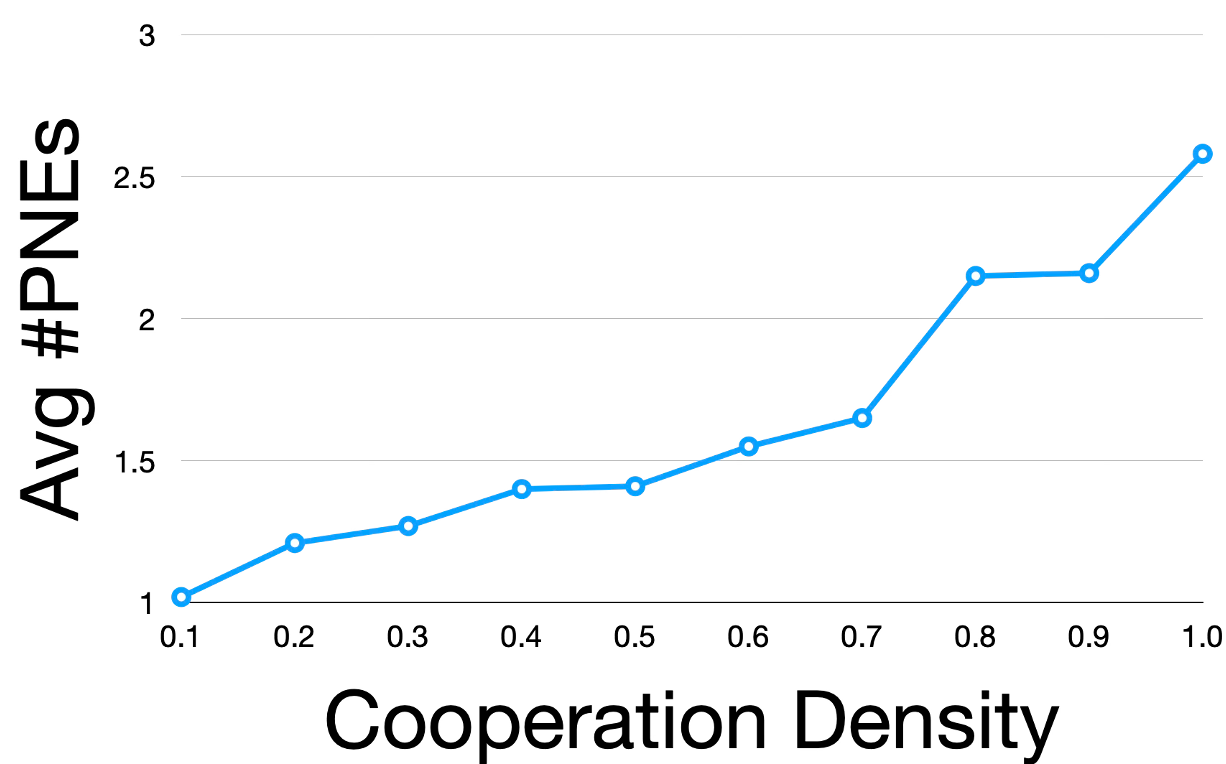}
        \label{fig:avg-den}
    \end{subfigure}
    \hfill
    \begin{subfigure}[b]{0.23\linewidth}
        \centering
        \includegraphics[width=\linewidth]{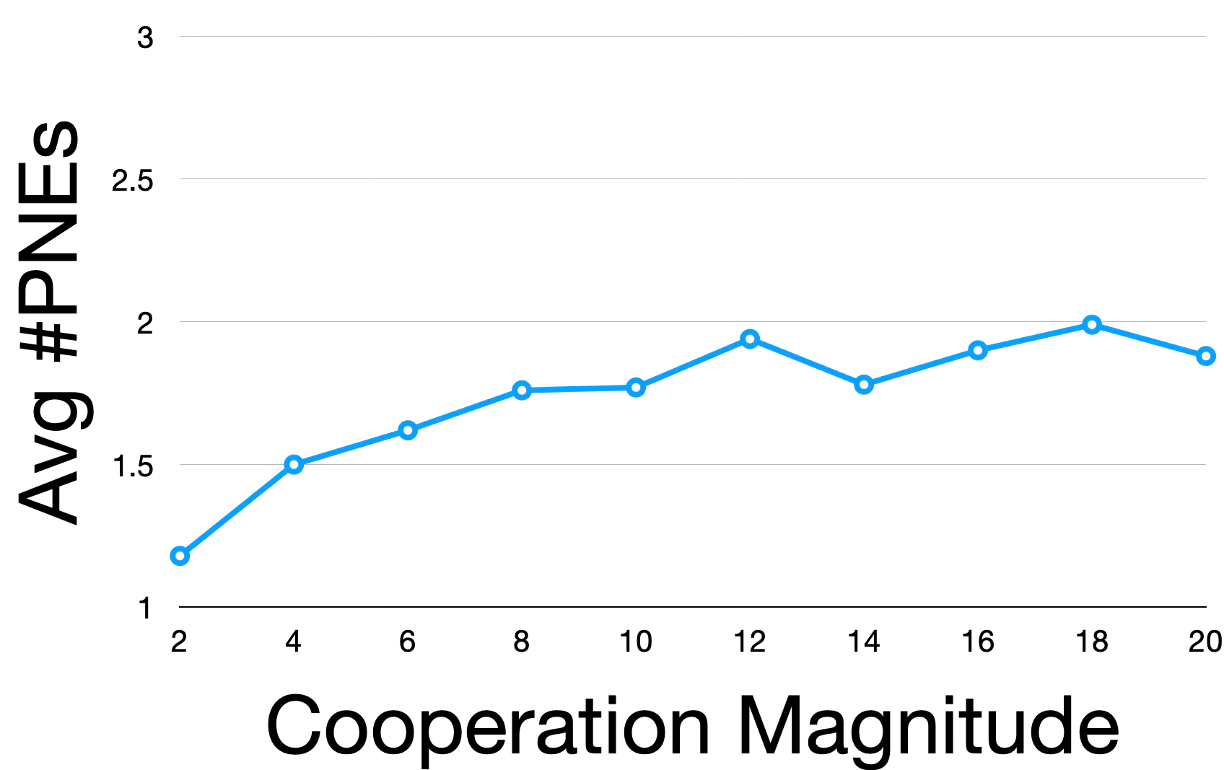}
        \label{fig:avg-magnitude}
    \end{subfigure}
    \hfill
     \begin{subfigure}[b]{0.23\linewidth}
        \centering
        \includegraphics[width=\linewidth]{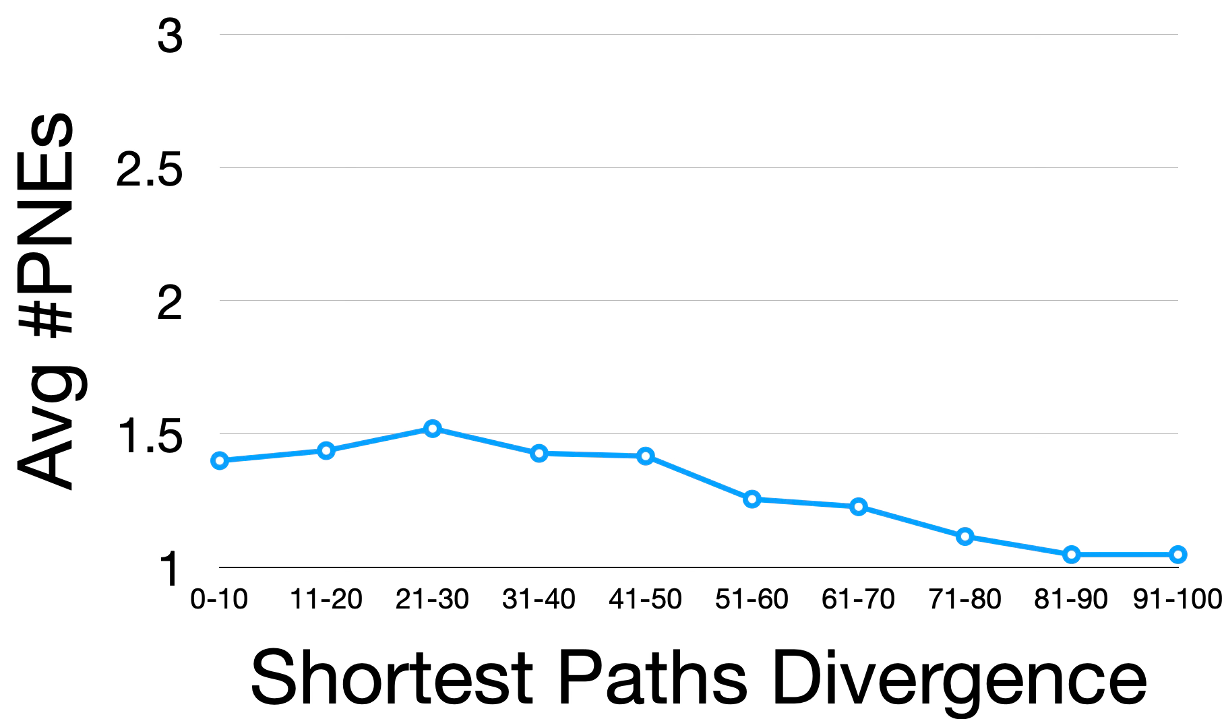}
        \label{fig:avg-divergence}
    \end{subfigure}
    \hfill
    \begin{subfigure}[b]{0.23\linewidth}
        \centering
        \includegraphics[width=\linewidth]{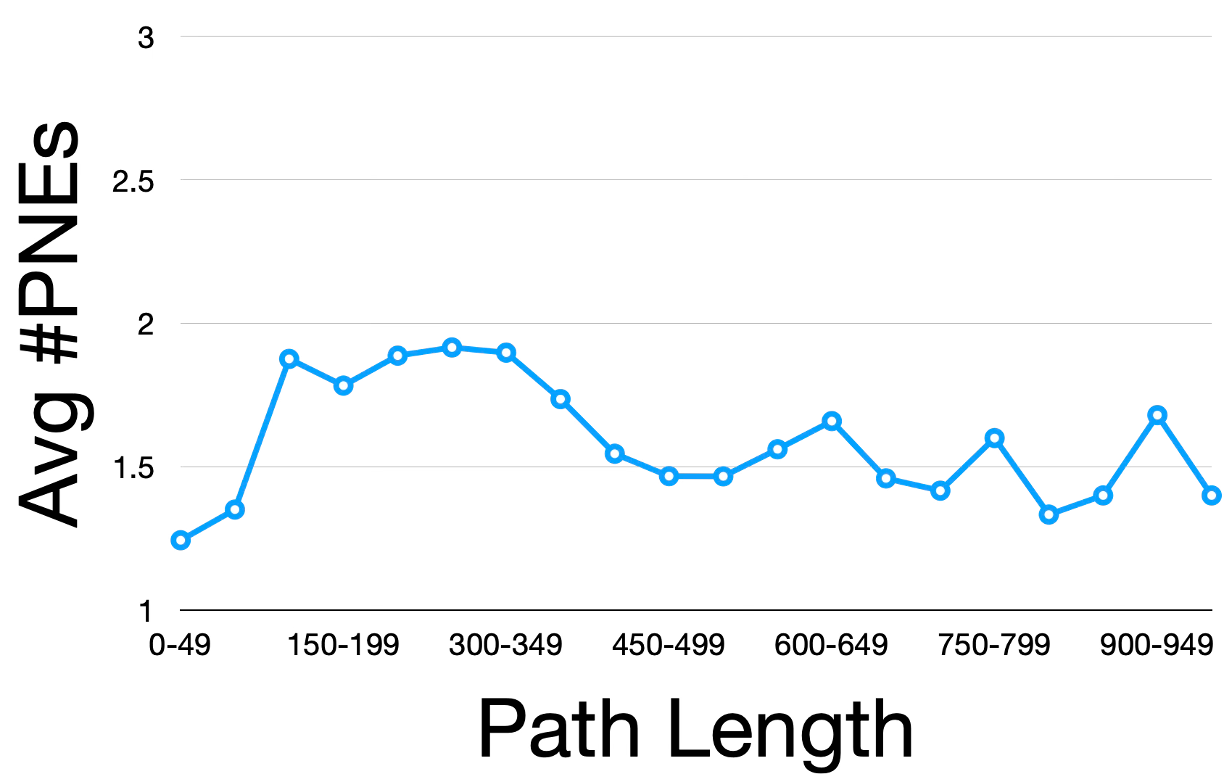}
        \label{fig:avg-length}
    \end{subfigure}
    \hfill
    \caption{Effects of cooperation factors on average amount of PNEs.}
      \label{fig:avg}
\end{figure}

\begin{figure}[ht]
    \centering
    \begin{subfigure}[b]{0.49\linewidth}
        \centering
        \includegraphics[width=\linewidth]{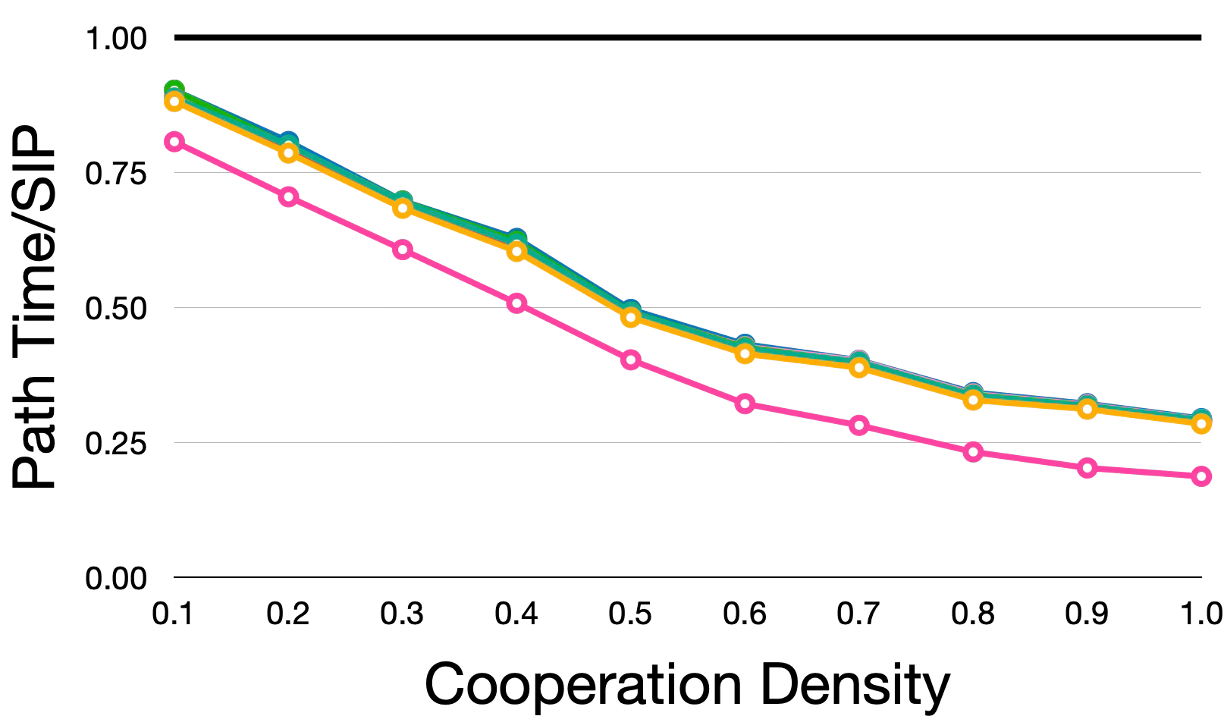}
        \label{fig:res-density-individual}
    \end{subfigure}
    \hfill
    \begin{subfigure}[b]{0.49\linewidth}
        \centering
        \includegraphics[width=\linewidth]{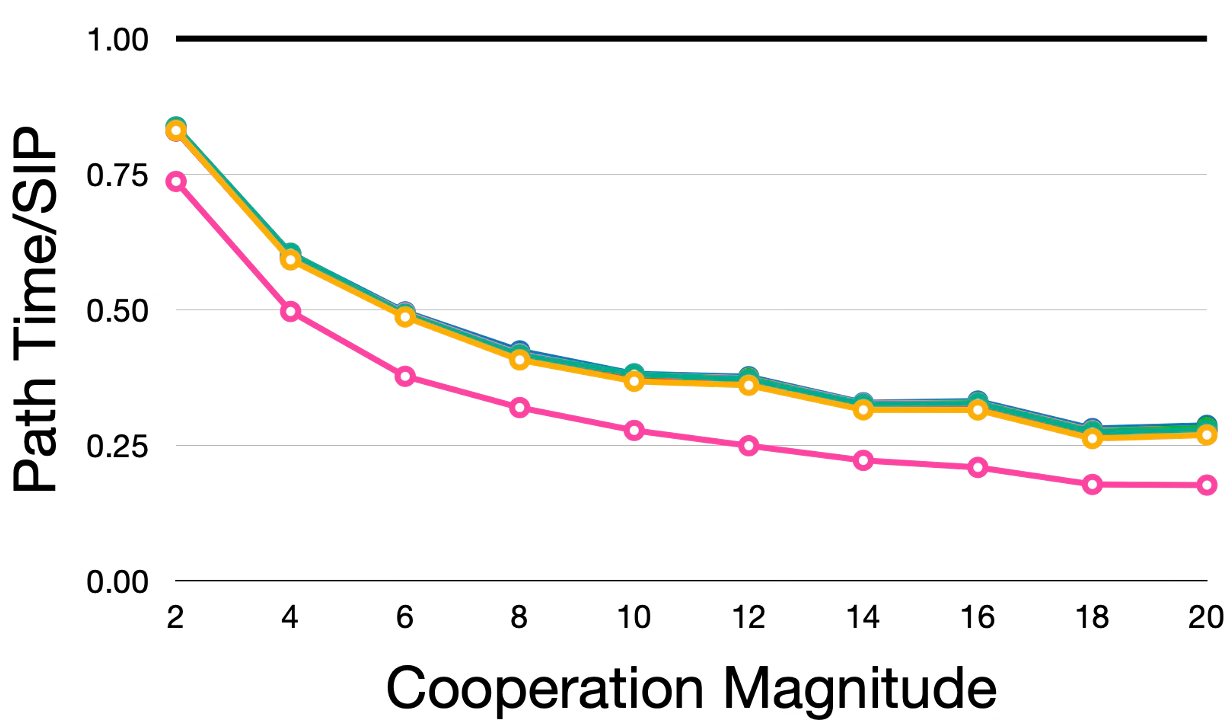}
        \label{fig:res-magnitude-individual}
    \end{subfigure}
    \hfill
     \begin{subfigure}[b]{0.49\linewidth}
        \centering
        \includegraphics[width=\linewidth]{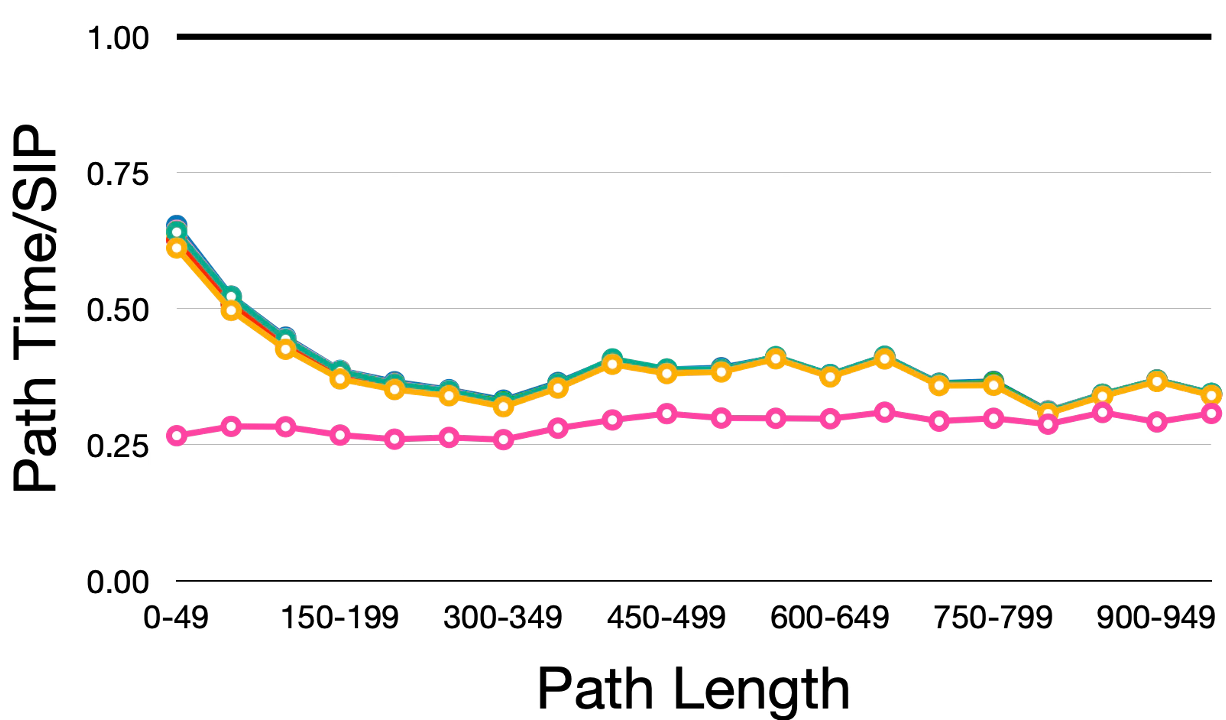}
        \label{fig:res-divergence-individual}
    \end{subfigure}
    \hfill
    \begin{subfigure}[b]{0.49\linewidth}
        \centering
        \includegraphics[width=\linewidth]{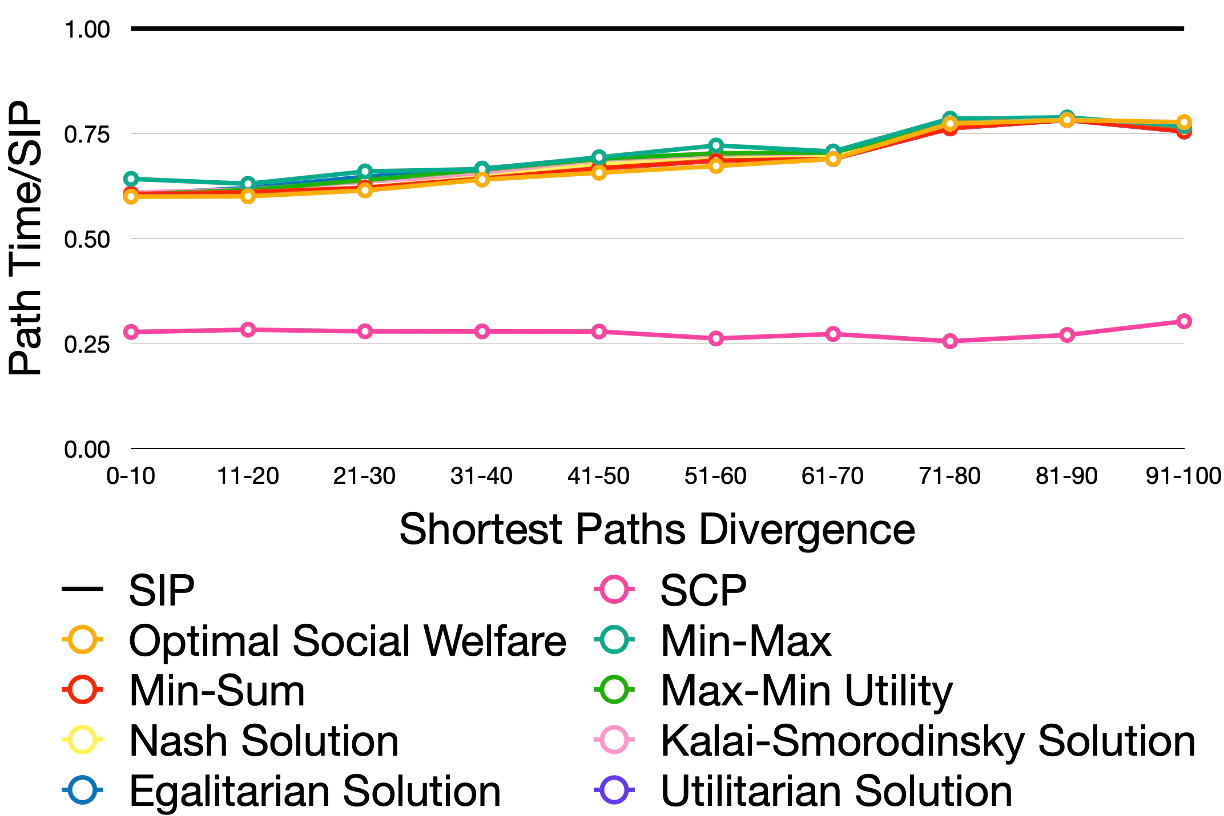}
        \label{fig:res-length-individual}
    \end{subfigure}
    \hfill
    \caption{Effects of cooperation factors on individual path times.}
      \label{fig:results}
\end{figure}
\begin{figure}[ht]
    \begin{subfigure}[b]{0.49\linewidth}
        \centering
        \includegraphics[width=\linewidth]{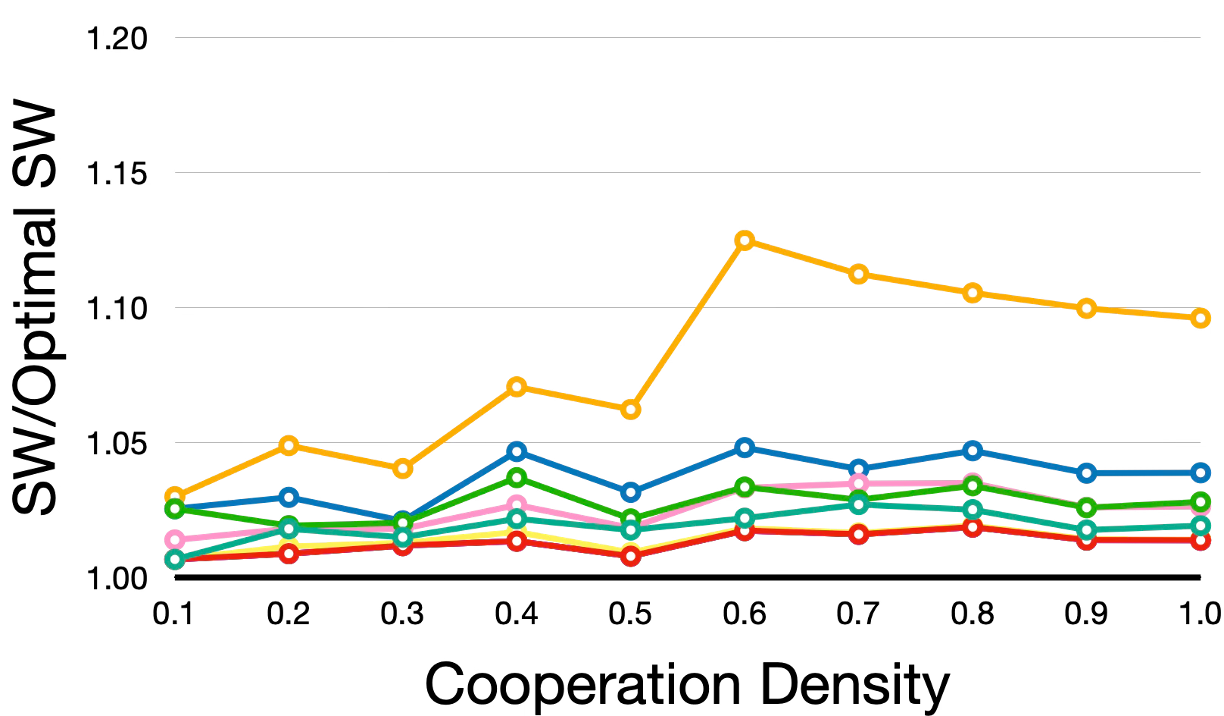}
    \end{subfigure}
    \hfill
    \begin{subfigure}[b]{0.49\linewidth}
        \centering
        \includegraphics[width=\linewidth]{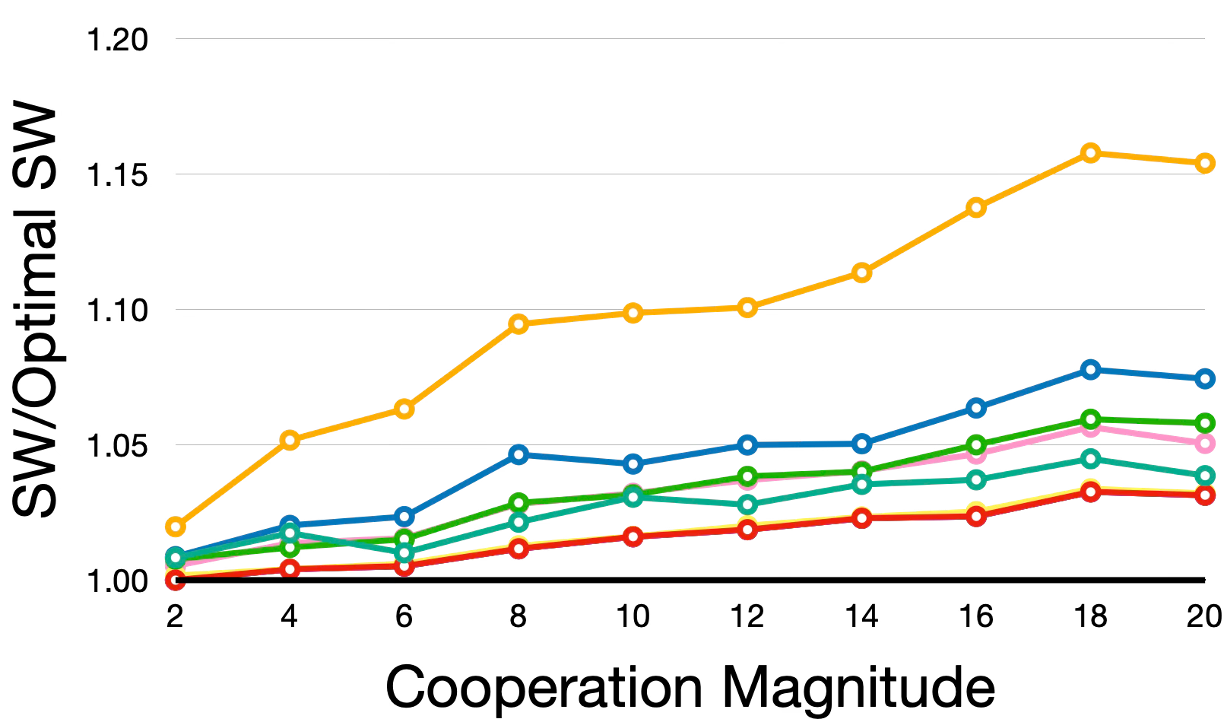}
    \end{subfigure}
    \hfill
    \begin{subfigure}[b]{0.49\linewidth}
        \centering
        \includegraphics[width=\linewidth]{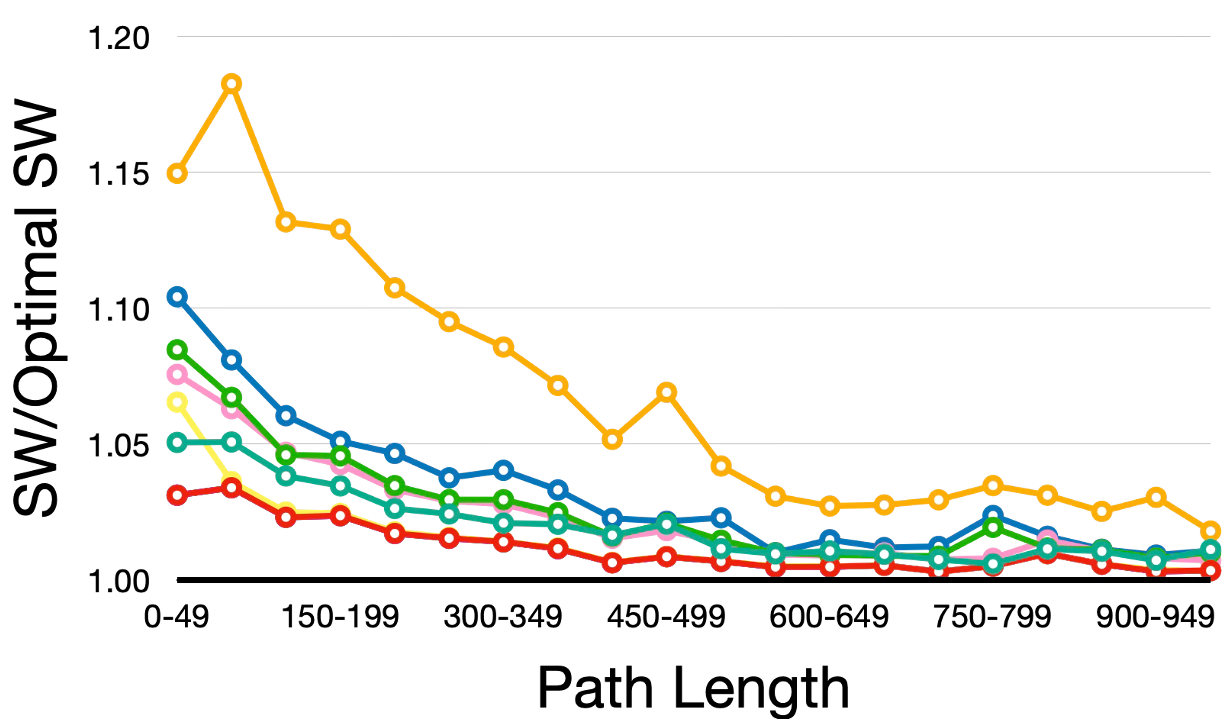}
    \end{subfigure}
    \hfill
    \begin{subfigure}[b]{0.49\linewidth}
        \centering
        \includegraphics[width=\linewidth]{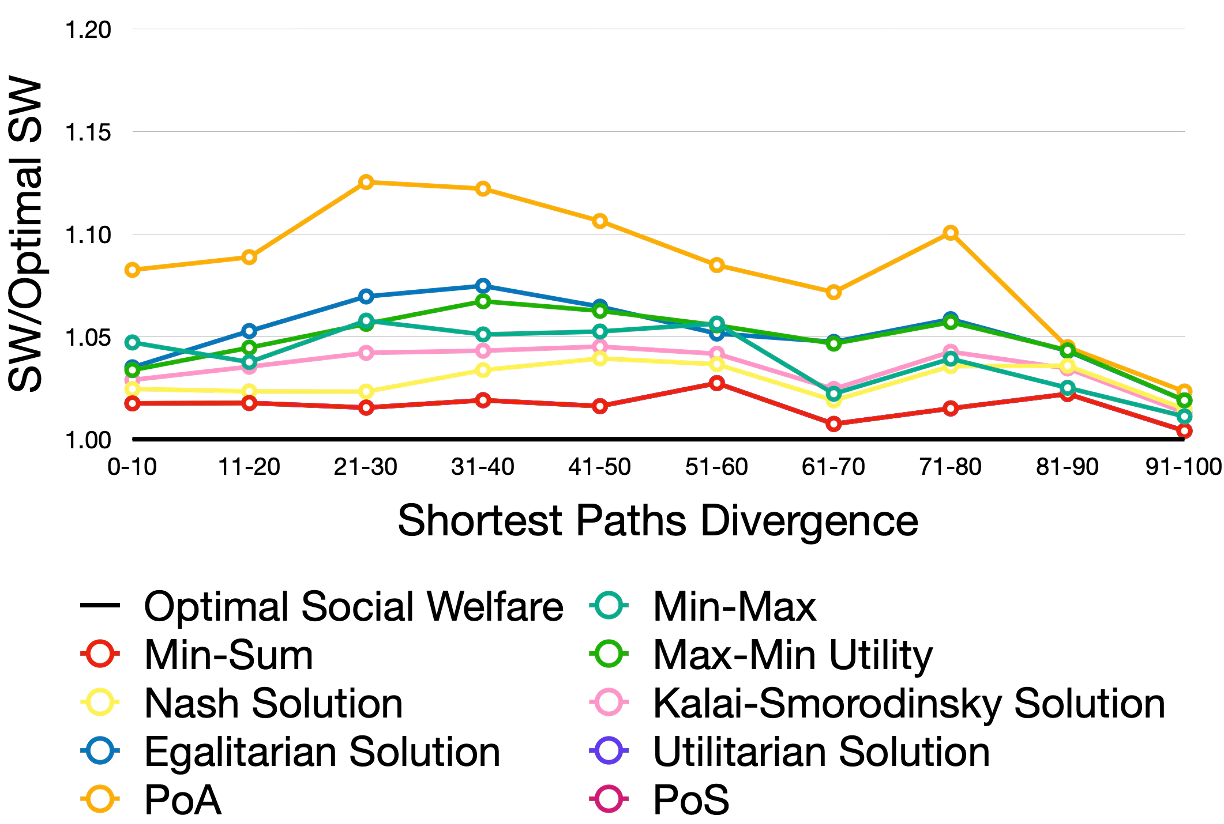}
    \end{subfigure}
    \caption{Effects of cooperation factors on social welfare.}
      \label{fig:res-length-social}
\end{figure}

To enable a more meaningful comparison among the different selection methods, we focus on the subset of instances that admit more than one PNE. Appendix~\ref{appendix:non-filtered-charts} reports the corresponding results over all tested scenarios.
Figure~\ref{fig:results} summarizes the impact of the cooperation parameters on individual path times, measured by improvement ratios relative to the shortest independent paths. Increasing Cooperation Density and Cooperation Magnitude improves the agents' individual path times under all selection methods, following trends similar to those of the shortest cooperation path ($SCP$)\footnote{Recall that the shortest cooperation path ($SCP$) is defined as the shortest path obtained under the idealized assumption that the agents cooperate at all nodes without waiting. This path is generally not attainable and applies only in specific cases, such as when both agents follow the same path.}.
Longer paths also lead to greater individual improvements, with the outcomes approaching the performance of the shortest cooperation path.
The effect of SPD is more gradual: as paths become more separated, the improvement decreases slowly. However, this trend should be interpreted with caution, since the analysis is restricted to scenarios with more than one PNE, and such scenarios constitute only a small fraction of highly separated instances.
Across all four factors, the different selection methods yield nearly identical individual path times. Moreover, these path times are close to those achieved by the optimal social-welfare joint strategy.
Figure~\ref{fig:res-length-social} summarizes the impact of the cooperation parameters on social welfare, measured by ratios relative to the optimal social welfare. Cooperation Density has little effect on the social welfare achieved by the different selection methods. However, when the density of cooperation nodes exceeds 0.5, the importance of coordination increases substantially, as reflected by the sharp increase in the PoA.
As Cooperation Magnitude increases, the gap between PNE outcomes and the optimal social-welfare solution widens, highlighting the social-welfare cost of stability.
SPD does not exhibit a consistent trend. Nevertheless, for highly separated paths, stable joint strategies appear to approach the optimal social-welfare solution. In contrast, Path Length exhibits a strong relation to the social welfare of PNE outcomes: as paths become longer, the outcomes approach the optimal social welfare.
Across all scenarios, the Min-Sum convention and the Utilitarian solution consistently attain the lowest social-welfare ratios, as both explicitly optimize total path time. In contrast, the Min-Max convention and the Egalitarian solution yield the highest ratios, reflecting their emphasis on fairness and equality rather than purely minimizing total path time.

Overall, since the different selection methods produce comparable individual outcomes, agents can adopt shared social objectives, such as makespan minimization or social-welfare maximization, without substantially compromising individual performance.

\section{Conclusions}
In this paper, we introduced the \ictpp framework for analyzing strategic cooperation among two self-interested agents whose incentives to cooperate depend on temporal and spatial context. 
We characterized the structure of Pure Nash Equilibria (PNE) and developed an
efficient algorithm to enumerate all non-dominated PNE, proving existence in
all instances.
We compared coordination mechanisms for two agents and examined how different factors affect cooperation outcomes in general.

These findings establish a foundation for understanding and guiding rational cooperation in autonomous multi-agent systems operating in real-world settings, where cooperation incentives are context-dependent and shaped by temporal and spatial synchronization. 
Therefore, this work opens the door to a broad line of future research, including extensions to $k>2$ agents, online planning, and scenarios with limited information about other agents.

\appendix
\bibliographystyle{plainnat}
\bibliography{bibli}

\appendix
\section{Best Response Algorithm}\label{app:br-algo}
Algorithm \ref{algorithm:BRP} identifies the best response strategy for agent $a_1$, given $a_2$'s path $\pi$.  
The algorithm begins by computing $SIP_{s_1}$ (the shortest independent paths from $a_1$'s starting node to all other nodes)  
and $SIP_{g_1}$ (the shortest independent paths from all nodes to $a_1$'s target node) and initializes the optimal cooperation  
starting and departure nodes to \textsc{None} [lines 1-3]. It then iterates over all nodes in $\pi$ to find the first cooperation  
node $a_1$ can reach at a cooperation-relevant time, denoted as ${c^*_1}$ [lines 4-7].  
If no such node exists, since no cooperation can be established, the algorithm returns the independent shortest path as the best response to $\pi$ [lines 8-9]. Otherwise, it tracks the cooperation path time along $\pi_{c^*_1,g_2}$  
and determines the departure node $d$ that optimizes $a_1$'s arrival at its target [lines 10-18].  
Finally, the algorithm compares the cooperation-assisted path time with the independent shortest path time  
and returns the strategy that minimizes travel time [lines 19-22].
\begin{algorithm}[ht]
\caption{\textsc{Best Response Path ($G, V_C, {s_1}, {g_1}, {s_2}, {g_2}, \pi$)}}\label{algorithm:best-response}
\label{algorithm:BRP}
\begin{algorithmic}[1]

    \State $SIP_{s_1} \gets \textsc{all shortest paths from}(G, \tau^1, {s_1})$ \Comment{Find shortest paths from starting nodes}
    \State $SIP_{g_1} \gets \textsc{all shortest paths to}(G, \tau^1, {g_1})$ \Comment{Find shortest paths to target nodes}
    \State ${c^*_1}, d \gets None$
    \ForAll {${c_i} \in \pi$ by its order of occurrence along $\pi$}
        \If{$SIP_{s_1,c_i} \leq \pi_{s_2,c_i} +\tau^1_{c_1}- \tau^2_{c_i}$}  \Comment{Find the first cooperation node in $\pi$ that $a_1$ can reach}
            \State ${c^*_1} \gets {c_i}$ \Comment{At a cooperation-relevant time}
            \State \textbf{break}
        \EndIf
    \EndFor
    \If{${c^*_1} = None$} \Comment{If cooperation cannot be established, return the independent shortest path}
        \State \Return $SIP_{s_1,g_1}$
    \EndIf
    \State $optimalPathTime \gets \infty$
    \State $cooperationPathTime \gets 0$
    \ForAll {$v \in \pi_{c^*_1,g_2}$ by its order of occurrence along $\pi$} \Comment{Find the optimal cooperation departure node}
        \State $cooperationPathTime \gets cooperationPathTime + \tau^2_{v}$
        \If{$cooperationPathTime + SIP_{v,g_1} \leq optimalPathTime $} \footnote{}
            \State $optimalPathTime \gets cooperationPathTime + SIP_{v,g_1}$
            \State ${d} \gets v$
        \EndIf
        \State $v_{next} \gets \text{next node in } \pi$
        \State $cooperationPathTime \gets cooperationPathTime + \tau_v,v_{next}$
    \EndFor
    \If{$SIP_{s_1,c^*_1} + optimalPathTime \leq SIP_{s_1,g_1}$} \Comment{Return the shorter path between the independent}
        \State \Return $SIP_{s_1,c^*_1} \circ \pi_{c^*_1,d} \circ SIP_{d,g_1}$ \Comment{shortest path and the cooperation-assisted path}
    \Else
        \State \Return $SIP_{s_1,g_1}$
    \EndIf
\end{algorithmic}
\end{algorithm}

\footnotetext{In the context of the pseudocode, $SP_{v,u}$ denotes both the shortest path between nodes $v$ and $u$ and the associated path time, depending on the context.}

\paragraph{Complexity Analysis}  
The algorithm begins with two executions of Dijkstra's algorithm to compute the shortest paths from ${s_1}$ and ${g_1}$ to all nodes in the graph.  
Each run of Dijkstra's algorithm has a time complexity of \( \mathcal{O}(|E| + |V| \log |V|) \).  
Thus, the total complexity for these two runs is:  
\[
\mathcal{O}(2 \cdot (|E| + |V| \log |V|)) = \mathcal{O}(|E| + |V| \log |V|)
\]  
The algorithm then iterates over all nodes in $\pi$ twice, first to find the optimal cooperation starting node  
and then to determine the optimal cooperation departure node. Since $|\pi| \leq |V|$, these two iterations contribute:  
\[
\mathcal{O}(2 \cdot |V|) = \mathcal{O}(|V|)
\]  
Combining these, the overall complexity of the algorithm remains:  
\[
\mathcal{O}(|E| + |V| \log |V|)
\]

\section{Proof of Theorem \ref{theorem:pne_equivalance}}\label{app:theorem:pne_equivalance}
\paragraph{Theorem \ref{theorem:pne_equivalance}}
A cooperation joint strategy constitutes a PNE if and only if it is an ECJS.
\begin{proof}
     \textbf{Direction 1: } $\boldsymbol{(\pi^1, \pi^2)\neq (SIP_{s_1,g_1},SIP_{s_2,g_2})}$ \textbf{is an ECJS} $\boldsymbol{\implies(\pi^1, \pi^2)}$ \textbf{is a PNE:}\\
Since $(\pi^1,\pi^2)$ is an ECJS, it can be expressed as:
\[
\pi^1 = \pi^{1(NC)}_{s_1, c_s} \circ \pi^S_{c_s, c_d} \circ SIP_{c_d, g_1}, \quad 
\pi^2 = \pi^{2(NC)}_{s_2, c_s} \circ \pi^S_{c_s, c_d} \circ SIP_{c_d, g_2}.
\]
Where ${c_s}, {c_d} \in V_C$, $\pi^{1(NC)}_{s_1, c_s} \in \Pi^{NC}_{s_1,c_s}$, $\pi^{2(NC)}_{s_2, c_s} \in \Pi^{NC}_{s_2,c_s}$, and $\pi^S_{c_s,c_d} \in \Pi^S_{c_s,c_d}$.\\    

    By Theorem \ref{cor1}, the optimal path for $a_i$, assuming $\pi_{-i}$'s is fixed, known in advance, and rational, is one of the following:
    \begin{enumerate}
        \item The shortest independent path from ${s_i}$ to ${g_i}$.
        \item The path $SIP_{s_i,c^*_1} \circ \pi^{-i}_{c^*_1,d} \circ SIP_{d,g_i}$, where ${c^*_1}$ is the first cooperation node along $\pi^{-i}$ that $a_i$ can reach at a cooperation-relevant time, and $d$ is the optimal departure node for $a_i$ along $\pi^{-i}_{c^*_1,g_i}$
    \end{enumerate}
    Since $(\pi^1, \pi^2)$ is an ECJS, following Condition \ref{item:cond-cooperation-over-independent}, it holds that:
    \[
    T(\pi^i \mid \pi^{-i}) \leq T(SIP_{s_i, g_i})
    \]
    We therefore examine the cooperation-assisted path $SIP_{s_i,c^*_1} \circ \pi^{-i}_{c^*_1,d} \circ SIP_{d,g_i}$.  
    Since $a_{-i}$'s path to ${c_s}$ is non-cooperative (Condition \ref{item:cond-non-coop}, it follows that ${c_s}$ is the first cooperation node along $\pi^{-i}$ that $a_i$ can reach at a cooperation-relevant time. Furthermore, by Condition~\ref{item:cond-departure}, \(c_d\) is the optimal departure node for $a_i$ along $\pi^{-i}_{c_s,g_{-i}}$. Therefore, $a_i$'s best response to $\pi^{-i}$ is the path 
    \[
    \pi^{i^*} = SIP_{s_i,c_s} \circ \pi^{-i}_{c_s, c_d} \circ SIP_{c_d, g_i} =  SIP_{s_i,c_s} \circ \pi^S_{c_s, c_d} \circ SIP_{c_d, g_i}
    \]
    The arrival time of $a_i$ at its target node using this path can be expressed as:
    \[
    T(\pi^{i^*} | \pi^{-i}) = \max(T(SIP_{s_i, c_s}), T(\pi^{{-i}(NC)}_{s_{-i}, c_s})) + \tau^2_{c_s} + D(\pi^S_{c_s, c_d} | \pi^S_{c_s, c_d}) + T(SIP_{c_d, g_i})
    \]  
    We now analyze $a_i$'s path $\pi^i = \pi^{i(NC)}_{s_i, c_s} \circ \pi^S_{c_s, c_d} \circ SIP_{c_d, g_i}$. The arrival time of $a_i$ using $\pi^i$ is:
    \[
    T(\pi^i | \pi^{-i}) = \max(T(\pi^{i(NC)}_{s_i, c_s}), T(\pi^{{-i}(NC)}_{s_{-i}, c_s})) + \tau^2_{c_s} + D(\pi^S_{c_s, c_d} | \pi^S_{c_s, c_d}) + T(SIP_{c_d, g_i})
    \]
    Since $\pi^{1(NC)}_{s_1, c_s}$ and $\pi^{2(NC)}_{s_2, c_s}$ are \emph{Mutually-Robust Non-Cooperative Partial Paths}, there are two possible cases for $T(\pi^{i(NC)}_{s_i, c_s})$:
    \begin{enumerate}
        \item If $T(\pi^{i(NC)}_{s_i, c_s}) \leq T(\pi^{{-i}(NC)}_{s_{-i}, c_s})$, then $T(SIP_{s_i, c_s}) \leq T(\pi^{i(NC)}_{s_i, c_s}) \leq T(\pi^{{-i}(NC)}_{s_{-i}, c_s})$ Thus:
        \[
        \max(T(\pi^{i(NC)}_{s_i, c_s}), T(\pi^{{-i}(NC)}_{s_{-i}, c_s})) = \max(T(SIP_{s_i, c_s}), T(\pi^{{-i}(NC)}_{s_{-i}, c_s})) = T(\pi^{{-i}(NC)}_{s_{-i}, c_s})
        \]
        Hence:
        \[
        T(\pi^i | \pi^{-i}) = T(\pi^{i^*} | \pi^{-i})
        \]
        \item If $T(\pi^{i(NC)}_{s_i, c_s}) > T(\pi^{{-i}(NC)}_{s_{-i}, c_s})$, then, since $\pi^{1(NC)}_{s_1, c_s}$ and $\pi^{2(NC)}_{s_2, c_s}$ are Mutually-Robust Non-Cooperative Partial Paths, it follows that $\pi^{i(NC)}_{s_i, c_s} = SIP_{s_i, c_s}$. Thus:
        \[
        \max(T(\pi^{i(NC)}_{s_i, c_s}), T(\pi^{{-i}(NC)}_{s_{-i}, c_s})) = \max(T(SIP_{s_i, c_s}), T(\pi^{{-i}(NC)}_{s_{-i}, c_s}))
        \]
        Hence:
        \[
        T(\pi^i | \pi^{-i}) = T(\pi^{i^*} | \pi^{-i})
        \]
    \end{enumerate} 
    In both cases, $a_i$'s path $\pi^i$ serves as the best response to $\pi^{-i}$.  
    Therefore, the joint strategy $(\pi^1, \pi^2)$ constitutes a PNE.
    
 \textbf{Direction 2: } $\boldsymbol{(\pi^1, \pi^2)\neq (SIP_{s_1,g_1},SIP_{s_2,g_2})}$ \textbf{is a PNE}  $\boldsymbol{\implies (\pi^1, \pi^2)}$ \textbf{is an ECJS:}\\
Given a Pure Nash Equilibrium (PNE) joint strategy $(\pi^1, \pi^2)$, we assume, by contradiction, that one or more of the five conditions for an \emph{ECJS} is violated. Each condition is analyzed individually:
\begin{enumerate}
    \item \underline{Condition~\ref{item:cond-stable}:} 
    We consider both parts of Condition \ref{item:cond-stable}: (a) the agents cooperate along exactly one continuous cooperation segment, and (b) this cooperation segment is stable.
    \begin{enumerate}
        \item
        We assume, by contradiction, that the joint strategy $(\pi^1, \pi^2)$ contains more than one cooperation segment. In that case, it can be expressed as:
        \[
        \pi^1 = \pi^1_{s_1, c_{s_1}} \circ \pi^C_{c_{s_1}, c_{d_1}} \circ \pi^1_{c_{d_1}, c_{s_2}} \circ \pi^C_{c_{s_2}, c_{d_2}} \circ \pi^1_{c_{d_2}, c_{g_1}}
        \]
        \[
        \pi^2 = \pi^2_{s_2, c_{s_1}} \circ \pi^C_{c_{s_1}, c_{d_1}} \circ \pi^2_{c_{d_1}, c_{s_2}} \circ \pi^C_{c_{s_2}, c_{d_2}} \circ \pi^2_{c_{d_2}, c_{g_2}}
        \]
        However, according to Lemma \ref{cont} (Cooperation Continuity), we have:
        \[
        T(\pi^1_{s_1, c_{d_1}} \circ \pi^2_{c_{d_1}, c_{s_2}} | \pi^2_{s_2, c_{s_2}}) \leq T(\pi^1_{s_1, c_{s_2}}| \pi^2_{s_2, c_{s_2}})
        \]
        Therefore, it follows that:
        \[
        T(\pi^1_{s_1, c_{d_1}} \circ \pi^2_{c_{d_1}, c_{s_2}} \circ \pi^1_{c_{s_2}, g_1}|\pi^2) \leq T(\pi^1| \pi^2)
        \]
If equality holds, since the joint path \((\pi^1_{s_1, c_{d_1}} \circ \pi^2_{c_{d_1}, c_{s_2}} \circ \pi^1_{c_{s_2}, g_1},\pi^2)\) involves a single continuous cooperation segment, \(a_1\) would prefer the strategy \(\pi^1_{s_1, c_{d_1}} \circ \pi^2_{c_{d_1}, c_{s_2}} \circ \pi^1_{c_{s_2}, g_1}\) over $\pi^1$, contradicting the assumption that \((\pi^1, \pi^2)\) is a PNE.
    
        \item
            We assume, by contradiction, that the cooperation partial path $\pi_{c_s, c_d}$ is not stable. In that case, for one of the agents, without loss of generality $a_1$, there exists a better departure node ${c_{d'}} = v_{d^*}^1(\pi_{c_s, c_d})$, where ${c_{d'}} \neq {c_d}$, along the cooperation path: 
        \[
        T(\pi_{c_s, c_{d'}} \circ SIP_{c_{d'}, g_1} | \pi_{c_s, c_d}) < T(\pi_{c_s, c_{d}} \circ SIP_{c_{d}, g_1} | \pi_{c_s, c_d})
        \]
         Therefore, it follows that:
        \[
        T(\pi^1_{s_1, c_s} \circ \pi_{c_s, c_{d'}} \circ SIP_{c_{d'}, g_1} | \pi^2) < T(\pi^1, \pi^2)
        \]
        which contradicts the assumption that $(\pi^1, \pi^2)$ is a PNE.
    \end{enumerate}

    \item \underline{Condition \ref{item:cond-non-coop}:}
    We assume, by contradiction, that the paths of the two agents from their starting nodes to ${c_s}$ are not \emph{Mutually-Robust Non-Cooperative Partial Paths}.
    We consider two options:
    \begin{enumerate}
        \item The path of one of the agents, without loss of generality $a_2$, to ${c_s}$, $\pi^2_{s_2,c_s}$, is not non-cooperative. In this case, there exists a cooperation node along this path, $c \in \pi^2_{s_2,c_s}$, $c \neq {c_s}$, that $a_1$ can reach at a cooperation-relevant time (if there is more than one such node, we take $c$ to be the first one along $\pi^2_{s_2,c_s}$).  
    From Lemma \ref{early} (Early Cooperation), it follows that:
    \[
        T(SIP_{s_1, c} \circ \pi^2_{c, c_s} \circ \pi^1_{c_s, g_1}, \pi^2) \leq T(SIP_{s_1, c_s} \circ \pi^1_{c_s, g_1}, \pi^2)
    \]
    Additionally, since $SIP_{s_1,c_s}$ is the shortest path from ${s_1}$ to ${c_s}$ without involving cooperation, it follows that:
    \[
        T(SIP_{s_1, c_s} \circ \pi^1_{c_s, g_1}, \pi^2) \leq T(\pi^1_{s_1, c_s} \circ \pi^1_{c_s, g_1}, \pi^2) = T(\pi^{1}_{s_{1},g_{1}} | \pi^2_{s_2,g_2})
    \]
    If equality holds, since the joint path $(SIP_{s_1, c} \circ \pi^2_{c, c_s} \circ \pi^1_{c_s, g_1}, \pi^2)$ starts cooperation earlier than
    $(\pi^1,\pi^2)$, \(a_1\) would prefer the strategy $SIP_{s_1, c} \circ \pi^2_{c, c_s} \circ \pi^1_{c_s, g_1}$ over $\pi^1$, contradicting the assumption that \((\pi^1, \pi^2)\) is a PNE.
    \item One of the agents, without loss of generality $a_1$, can modify its path toward ${c_s}$ to a new path $\pi^{1'}_{s_1,c_s}$, allowing cooperation to start earlier at time 
    \[
        t' = \max(T(\pi^{1'}_{s_1,c_s}), T(\pi^{2}_{s_{2},c_s})) < \max(T(\pi^{1}_{s_1,c_s}), T(\pi^{2}_{s_{2},c_s}))
    \]
    The agent can then use the modified full path $\pi^{1'}_{s_1,g_1} = \pi^{1'}_{s_1,c_s} \circ \pi^1_{c_s,g_1}$ and reach its target node at an improved time:
    \[
        T(\pi^{1'}_{s_1,g_1} | \pi^{2}_{s_{2},g_{2}}) = t' + \tau^2_{c_s} + T(\pi^1_{c_s,g_1} | \pi^{2}_{c_s,g_{2}}) <
    \]
    \[
        \max(T(\pi^{1}_{s_1,c_s}), T(\pi^{2}_{s_{2},c_s})) + \tau^2_{c_s} + T(\pi^1_{c_s,g_1} | \pi^{2}_{c_s,g_{2}}) = T(\pi^1, \pi^2)
    \]
    contradicting the assumption that $(\pi^1, \pi^2)$ is a PNE.
    \end{enumerate}
    \item \underline{Condition \ref{item:cond-departure}:} 
We assume, by contradiction, that the optimal departure node for one of the agents, without loss of generality $a_1$, along the other agent's path $\pi^2_{c_s,g_2}$ is ${c_{d'}} \neq {c_d}$. This implies that:
\[
    T(\pi^1_{s_1,c_s} \circ \pi^{2}_{c_s,c_{d'}} \circ SIP_{c_{d'},g_1} | \pi^{2}) < T(SIP_{s_1,c_s} \circ \pi^{2}_{c_s,c_d} \circ SIP_{c_d,g_1} | \pi^{2}) = T(\pi^1, \pi^2)
\]
which contradicts the assumption that $(\pi^1, \pi^2)$ is a PNE.

    \item \underline{Condition \ref{item:cond-shortest-departure}:} 
We assume, by contradiction, that for one of the agents, without loss of generality $a_1$, the partial path $\pi^1_{c_d,g_1}$ from the departure node ${c_d}$ to its target node ${g_1}$ differs from the shortest independent path $SIP_{c_d,g_1}$.  
Since $(\pi^1, \pi^2)$ involves a single continuous cooperation segment (Condition \ref{item:cond-stable}), and ${c_d}$ is the cooperation departure node, the joint path $(\pi^1_{c_d,g_1}, \pi^2_{c_d,g_2})$ contains no cooperation. Therefore:  
\[
T(\pi^1_{c_d,g_1} | \pi^2_{c_d,g_2}) = T(\pi^1_{c_d,g_1}) > T(SIP_{c_d,g_1}),
\]  
which implies:  
\[
T(\pi^1 | \pi^2) > T(\pi^1_{s_1, c_d} \circ SIP_{c_d, g_1} | \pi^2),
\]  
contradicting the assumption that $(\pi^1, \pi^2)$ is a PNE.

\item \underline{Condition \ref{item:cond-cooperation-over-independent}:} 
We assume, by contradiction, that one of the agents, without loss of generality $a_1$, prefers to take its independent path directly from its starting node to its target node, $SIP_{s_1,g_1}$. In this case, the independent shortest path for $a_1$ would be a better response to the $\pi^2$ than $\pi^1$:
\[
T(SIP_{s_1,g_1}, \pi^2) < T(\pi^1, \pi^2).
\]
This contradicts the assumption that the joint strategy $(\pi^1, \pi^2)$ is a PNE.
\end{enumerate}
Therefore, it holds that a cooperation joint strategy $(\pi^1, \pi^2)\neq (SIP_{s_1,g_1},SIP_{s_2,g_2})$ constitutes a \emph{Pure Nash Equilibrium (PNE)} if and only if $(\pi^1, \pi^2)$ is an \emph{ECJS}.
\end{proof}

\section{Proof of Lemma \ref{lemma:strategy-pruning-claim}}\label{app:lem:pruning}
\paragraph{Lemma \ref{lemma:strategy-pruning-claim}}
Let \(c_{d_1},c_{d_2} \in V_C\) be two distinct cooperation nodes. 
Suppose there exists a stable cooperation partial path from \(c_{d_1}\) to \(c_{d_2}\), denoted 
\(\pi^S_{c_{d_1},c_{d_2}}\). 
Then, for any stable cooperation partial path \(\pi_{c_s,c_{d_1}}\) from \(c_s\) to \(c_{d_1}\), the concatenated path
\[
\pi'_{c_s,c_{d_2}}
=
\pi_{c_s,c_{d_1}} \circ \pi^S_{c_{d_1},c_{d_2}}
\]
is also stable. Moreover, the optimal ECJS whose cooperation segment is \(\pi'_{c_s,c_{d_2}}\) dominates the optimal ECJS whose cooperation segment is \(\pi_{c_s,c_{d_1}}\). That is,
\begin{align*}
\forall i \in \{1,2\}, \quad
&
T\!\left(
    \pi_{c_s,c_{d_1}} \circ SIP_{c_{d_1},g_i}
    \mid
    \pi_{c_s,c_{d_1}} \circ SIP_{c_{d_1},g_{-i}}
\right)
\\
&\geq
T\!\left(
    \pi'_{c_s,c_{d_2}} \circ SIP_{c_{d_2},g_i}
    \mid
    \pi'_{c_s,c_{d_2}} \circ SIP_{c_{d_2},g_{-i}}
\right).
\end{align*}
\begin{proof}
    Given a stable path \( \pi_{c_s, c_{d_1}} \), since $\pi^S_{c_{d_1}, c_{d_2}}$ is stable,  the path $ \pi'_{c_s, c_{d_2}} = \pi_{c_s, c_{d_1}} \circ \pi^S_{c_{d_1}, c_{d_2}} $ is also stable.\\
    For an agent \( a_i \), the path time to its target node via a cooperation along $\pi_{c_s, c_{d_1}}$ can be described as:
    \[
    T(\pi_{c_s, c_{d_1}} \circ SIP_{c_{d_1}, g_i} | \pi_{c_s, c_{d_1}} \circ SIP_{c_{d_1}, g_{-i}}) = T
    (\pi_{c_s, c_{d_1}} | \pi_{c_s, c_{d_1}}) + \tau^2_{c_{d_1}} + T(SIP_{c_{d_1}, g_i} | SIP_{c_{d_1}, g_{-i}})
    \]
    Similarly, the path time to its target node via a cooperation along $\pi'_{c_s, c_{d_2}}$  is:
    \[
    T(\pi'_{c_s, c_{d_2}}\!\! \circ SIP_{c_{d_2}, g_i} | \pi'_{c_s, c_{d_2}}\!\! \circ SIP_{c_{d_2}, g_{-i}}) = T(\pi_{c_s, c_{d_1}} | \pi_{c_s, c_{d_1}}) + \tau^2_{c_{d_1}}\!\! + T(\pi^S_{c_{d_1}, c_{d_2}}\!\! \circ SIP_{c_{d_2}, g_i} | \pi^S_{c_{d_1}, c_{d_2}}\!\! \circ SIP_{c_{d_2}, g_{-i}})
    \]
    Since \( \pi^S_{c_{d_1}, c_{d_2}} \) is stable, it holds that:
    \[
    T(SIP_{c_{d_1}, g_i} | \pi^S_{c_{d_1}, c_{d_2}} \circ SIP_{c_{d_2}, g_{-i}}) \geq T(\pi^S_{c_{d_1}, c_{d_2}} \circ SIP_{c_{d_2}, g_i} | \pi^S_{c_{d_1}, c_{d_2}} \circ SIP_{c_{d_2}, g_{-i}})
    \]
    Additionally, since $c_{d_1}$ is the last cooperation node along $\pi_{c_s,c_{d_1}}$, \( (SIP_{c_{d_1}, g_i}, SIP_{c_{d_1}, g_{-i}}) \) involves no cooperation, and it follows that:
    \[
    T(SIP_{c_{d_1}, g_i} | SIP_{c_{d_1}, g_{-i}}) \geq T(SIP_{c_{d_1}, g_i} | \pi^S_{c_{d_1}, c_{d_2}} \circ SIP_{c_{d_2}, g_{-i}})
    \]
    Hence,
    \[
    T(\pi_{c_s, c_{d_1}} | \pi_{c_s, c_{d_1}}) + \tau^2_{c_{d_1}}\!\! + T(SIP_{c_{d_1}, g_i} | SIP_{c_{d_1}, g_{-i}})\!\! \geq \!\!\]\[ T(\pi_{c_s, c_{d_1}} | \pi_{c_s, c_{d_1}}) + \tau^2_{c_{d_1}} + T(\pi^S_{c_{d_1}, c_{d_2}} \circ SIP_{c_{d_2}, g_i} | \pi^S_{c_{d_1}, c_{d_2}} \circ SIP_{c_{d_2}, g_{-i}})
    \]
    which implies that:
    \[
    T(\pi_{c_s, c_{d_1}} \circ SIP_{c_{d_1}, g_i} | \pi_{c_s, c_{d_1}} \circ SIP_{c_{d_1}, g_{-i}}) \geq 
    T(\pi'_{c_s, c_{d_2}} \circ SIP_{c_{d_2}, g_i} | \pi'_{c_s, c_{d_2}} \circ SIP_{c_{d_2}, g_{-i}})
    \] 
   and the lemma holds.
\end{proof}

\section{Proof of Lemma \ref{lemma:PNE-existance}}\label{app:pne-existance}
\paragraph{Lemma~\ref{lemma:PNE-existance}}
    Consider a cooperative joint strategy of the form:
    \[
    \Big(SIP_{s_1,c_s} \circ \pi^S_{c_s,c_d} \circ SIP_{c_d,g_1},\ 
    SIP_{s_{2},c_s} \circ \pi^S_{c_s,c_d} \circ SIP_{c_d,g_{2}}\Big),
    \]
    where ${c_s}$ is the first cooperation node and ${c_d}$ is the last cooperation node in which the agents cooperate, and $\pi^S_{c_s,c_d} \in \Pi_{c_s,c_d}$ is a stable partial path. If the following condition holds:
    \[
    \forall i \in \{1,2\}, \quad
    T\Big(SIP_{s_i,c_s} \circ \pi^S_{c_s,c_d} \circ SIP_{c_d,g_i}|\ 
    SIP_{s_{-i},c_s} \circ \pi^S_{c_s,c_d} \circ SIP_{c_d,g_{-i}}\Big) 
    \leq T\Big(SIP_{s_i,g_i}\Big) 
    \]
    then a Pure Nash Equilibrium (PNE) exists.
    \begin{proof}
Following Theorem~\ref{theorem:pne_equivalance}, to show that a cooperative joint strategy constitutes a PNE, it is sufficient to show that it is an ECJS (by verifying the five conditions of Definition~\ref{def:equilibrium_cooperation_joint_strategy}).

Considering the joint strategy:
\[
\left(SIP_{s_1,c_s} \circ \pi^S_{c_s,c_d} \circ SIP_{c_d,g_1},\ 
SIP_{s_2,c_s} \circ \pi^S_{c_s,c_d} \circ SIP_{c_d,g_2}\right)
\]
We show that this joint strategy is either already an ECJS or can serve as the foundation for constructing one.
    The cooperation segment $\pi^S_{c_s,c_d}$ is stable by the Lemma's assumption, satisfying \textbf{Condition~\ref{item:cond-stable}}.  
    Additionally, since both agents follow their respective shortest independent paths from ${c_d}$ to their target nodes, \textbf{Condition~\ref{item:cond-shortest-departure}} is satisfied.  
    Furthermore, because both agents prefer this cooperative joint strategy over their independent shortest-path strategy, \textbf{Condition~\ref{item:cond-cooperation-over-independent}} also holds.

    We show that if Conditions~\ref{item:cond-non-coop} and~\ref{item:cond-departure} do not hold, the cooperation segment $\pi^S_{c_s,c_d}$ can be extended to satisfy them without violating the other conditions.

\textbf{Condition~\ref{item:cond-non-coop}:} If the paths $SIP_{s_1,c_s}$ and $SIP_{s_2,c_s}$ toward the cooperation starting node are non-cooperative, then, since neither agent can shorten its path to ${c_s}$, these paths are mutually robust, satisfying the condition.  
However, if one agent (without loss of generality, \(a_1\)) can reach a cooperation node \(c_e \in SIP_{s_2,c_s}\) at a cooperation-relevant time, we show that the partial path \(\pi_{c_e,c_d} = SIP_{c_e,c_s} \circ \pi^S_{c_s,c_d}\) is also stable, and the resulting joint strategy  
\[
\left(SIP_{s_1,c_e} \circ \pi_{c_e,c_d} \circ SIP_{c_d,g_1},\ 
SIP_{s_2,c_e} \circ \pi_{c_e,c_d} \circ SIP_{c_d,g_2}\right)
\]  
continues to satisfy Conditions~\ref{item:cond-stable}, \ref{item:cond-shortest-departure}, and~\ref{item:cond-cooperation-over-independent}.

     We assume, by contradiction, that for one of the agents, without loss of generality \(a_1\),  \(v_{d^*}^1(\pi_{c_e,c_d}) \neq {c_d}\), and denote the optimal departure node as \(c_{e'}\). Then:
\[
    T(SIP_{s_1,c_e} \circ \pi_{c_e,c_{e'}} \circ SIP_{c_{e'},g_1} \mid SIP_{s_2,c_e} \circ \pi_{c_e,c_d} \circ SIP_{c_d,g_2}) < 
    T(SIP_{s_1,c_e} \circ \pi_{c_e,c_d} \circ SIP_{c_d,g_1} \mid SIP_{s_2,c_e} \circ \pi_{c_e,c_d} \circ SIP_{c_d,g_2})
\]
Since both agents start cooperating at $c_e$, the above inequality can be rewritten as:
\[
    \max(T(SIP_{s_1,c_e} \mid SIP_{s_2,c_e}), T(SIP_{s_2,c_e} \mid SIP_{s_1,c_e})) + \tau^2_{c_e} + D(\pi_{c_e,c_{e'}} \mid \pi_{c_e,c_{e'}}) + T(SIP_{c_{e'},g_1} \mid \pi_{c_{e'},c_d} \circ SIP_{c_d,g_2})
\]
\[
    < \max(T(SIP_{s_1,c_e} \mid SIP_{s_2,c_e}), T(SIP_{s_2,c_e} \mid SIP_{s_1,c_e})) + \tau^2_{c_e} + D(\pi_{c_e,c_{e'}} \mid \pi_{c_e,c_{e'}}) + T(\pi_{c_{e'},c_d} \circ SIP_{c_d,g_1} \mid \pi_{c_{e'},c_d} \circ SIP_{c_d,g_2})
\]
Simplifying:
\[
    T(SIP_{c_{e'},g_1} \mid \pi_{c_{e'},c_d} \circ SIP_{c_d,g_2}) < T(\pi_{c_{e'},c_d} \circ SIP_{c_d,g_1} \mid \pi_{c_{e'},c_d} \circ SIP_{c_d,g_2})
\]
Adding the arrival time of \(a_1\) at \(c_{e'}\) (without cooperation) to both sides of the in-equation:
\[
    D(SIP_{s_1,c_{e'}}) + T(SIP_{c_{e'},g_1} \mid \pi_{c_{e'},c_d} \circ SIP_{c_d,g_2}) < D(SIP_{s_1,c_{e'}}) + T(\pi_{c_{e'},c_d} \circ SIP_{c_d,g_1} \mid \pi_{c_{e'},c_d} \circ SIP_{c_d,g_2})
\]
Since \(c_{e'}\) is the optimal departure node for \(a_1\), \((SIP_{c_{e'},g_1}, \pi_{c_{e'},c_d} \circ SIP_{c_d,g_2})\) does not involve cooperation. Moreover, as \(SIP_{s_1,g_1}\) denotes the shortest path to \(g_1\) without cooperation:
\[
    D(SIP_{s_1,c_{e'}}) + T(SIP_{c_{e'},g_1} \mid \pi_{c_{e'},c_d} \circ SIP_{c_d,g_2}) = D(SIP_{s_1,c_{e'}}) + T(SIP_{c_{e'},g_1}) \geq T(SIP_{s_1,g_1})
\]
Thus:
\[
D(SIP_{s_1,c_{e'}}) + T(\pi_{c_{e'},c_d} \circ SIP_{c_d,g_1} \mid \pi_{c_{e'},c_d} \circ SIP_{c_d,g_2}) > T(SIP_{s_1,g_1})
\]
Since the path time $T(\pi_{c_{e'},c_d} \circ SIP_{c_d,g_1} \mid \pi_{c_{e'},c_d} \circ SIP_{c_d,g_2})$ assumes cooperation along $\pi_{c_{e'},c_s}$, it follows that:
\[
T(SIP_{s_1,c_s} \circ \pi^S_{c_s,c_d} \circ SIP_{c_d,g_1} \mid SIP_{s_2,c_s} \circ \pi^S_{c_s,c_d} \circ SIP_{c_d,g_2}) 
\geq\]\[ D(SIP_{s_1,c_{e'}}) + T(\pi_{c_{e'},c_d} \circ SIP_{c_d,g_1} \mid \pi_{c_{e'},c_d} \circ SIP_{c_d,g_2})
> T(SIP_{s_1,g_1})
\]
contradicting our assumption.
Thus, we conclude that $\pi_{c_e,c_d} = SIP_{c_e,c_s} \circ \pi^S_{c_s,c_d}$ constitutes a \emph{Stable Cooperation Partial Path}. Since the joint strategy
$
\left(SIP_{s_1,c_e} \circ\pi_{c_e,c_d} \circ SIP_{c_d,g_1},\ 
SIP_{s_2,c_e} \circ \pi_{c_e,c_d} \circ SIP_{c_d,g_2}\right)
$ initiates cooperation at node $c_{e}$, which precedes ${c_s}$, it dominates the original joint strategy and therefore preserves the conditions of Definition~\ref{def:equilibrium_cooperation_joint_strategy} that were satisfied by the original joint strategy.

Since the set of cooperation nodes is finite, the cooperation segment can be repeatedly extended until reaching a cooperation starting node \( {c_s} \in V_C \) for which the paths \( SIP_{s_1,c_s} \) and \( SIP_{s_2,c_s} \) are mutually robust and non-cooperative, satisfying Condition~\ref{item:cond-non-coop}.

  \textbf{Condition \ref{item:cond-departure}:} By Corollary \ref{cor:dominated-strategies}, if there exists a cooperation joint strategy that satisfies all conditions of Definition \ref{def:equilibrium_cooperation_joint_strategy} except Condition \ref{item:cond-departure} (i.e., one of the agents, $a_i$, prefers to continue cooperating with $a_{-i}$ beyond ${c_d}$ along the path to $a_{-i}$'s target node), then a cooperation joint strategy extending this cooperation satisfies the same conditions of Definition \ref{def:equilibrium_cooperation_joint_strategy} and dominates the original one.  
    
Since the set of cooperation nodes is finite, the cooperation segment can be repeatedly extended until reaching a cooperation ending node \( {c_{d'}} \in V_C \) that satisfies Condition~\ref{item:cond-departure}.

Thus, the joint strategy  
\[
\left(SIP_{s_1,c_{s'}} \circ SCP^S_{c_{s'},c_{d'}} \circ SIP_{c_{d'},g_1}, \quad SIP_{s_2,c_{s'}} \circ SCP^S_{c_{s'},c_{d'},g_2} \circ SCP_{c_{d'}}\right)
\]
where \({c_{s'}}\) is the cooperation starting node that satisfies Condition~\ref{item:cond-non-coop}, found by repeatedly extending the cooperation segment from the beginning, and \({c_{d'}}\) is the cooperation ending node that satisfies Condition~\ref{item:cond-departure}, found by repeatedly extending the segment from the end, is an ECJS and, by Theorem~\ref{theorem:pne_equivalance}, constitutes a PNE.
\end{proof}

\section{Social Welfare Optimization}\label{supp:social-welfare}
We aim to determine the joint strategy \((\pi^{1^*}, \pi^{2^*})\) for agents \(a_1\) and \(a_2\) that minimizes their total path time:
\[
(\pi^{1^*}, \pi^{2^*})=\arg \min_{\pi^1, \pi^2} \left[ T_{g_1}(\pi^1|\pi^2) + T_{g_2}(\pi^2|\pi^1) \right]
\]
Considering $V_C = \{{c_1}, \ldots, {c_m}\}$ for some $m > 0$,  
a naive approach might involve evaluating all combinations of possible cooperation paths.  
However, Lemma~\ref{cont} provides key insights into the structure of the optimal social welfare joint strategy,  
significantly streamlining the search process. According to Lemma \ref{cont}, the most efficient cooperation path, minimizing the total path time between some initial cooperation node (denoted with ${c_s}$) and a final cooperation node (denoted with ${c_d}$) for both agents, involves continuous cooperation along the same path $SCP_{c_s,c_d}$.
The joint strategy that optimizes social welfare is therefore:
\[
(\pi^{1^*}, \pi^{2^*})=(SIP_{s_1, c_s} \circ SCP_{c_s, c_d} \circ SIP_{c_d, g_1}, SIP_{s_2, c_s} \circ SCP_{c_s, c_d} \circ SIP_{c_d, g_2})
\]
Thus, our objective is to identify the cooperation starting node ${c_s}$ and the cooperation departure node ${c_d}$ that optimize social welfare (see illustration in Figure \ref{fig:phases}):
\[
\begin{aligned}
   {c_s}, {c_d}=\arg \min_{{c_s}, {c_d}} \big[\\ & T_{g_1}(SIP_{s_1, c_s} \circ SCP_{c_s, c_d} \circ SIP_{c_d, g_1}| SIP_{s_2, c_s} \circ SCP_{c_s, c_d} \circ SIP_{c_d, g_2}) +\\
   &T_{g_2}(SIP_{s_2, c_s} \circ SCP_{c_s, c_d} \circ SIP_{c_d, g_2}|SIP_{s_1, c_s} \circ SCP_{c_s, c_d} \circ SIP_{c_d, g_1})\\ \big]
\end{aligned}
\]
\begin{figure}[ht]
    \centering
    \includegraphics[width=0.25\linewidth]{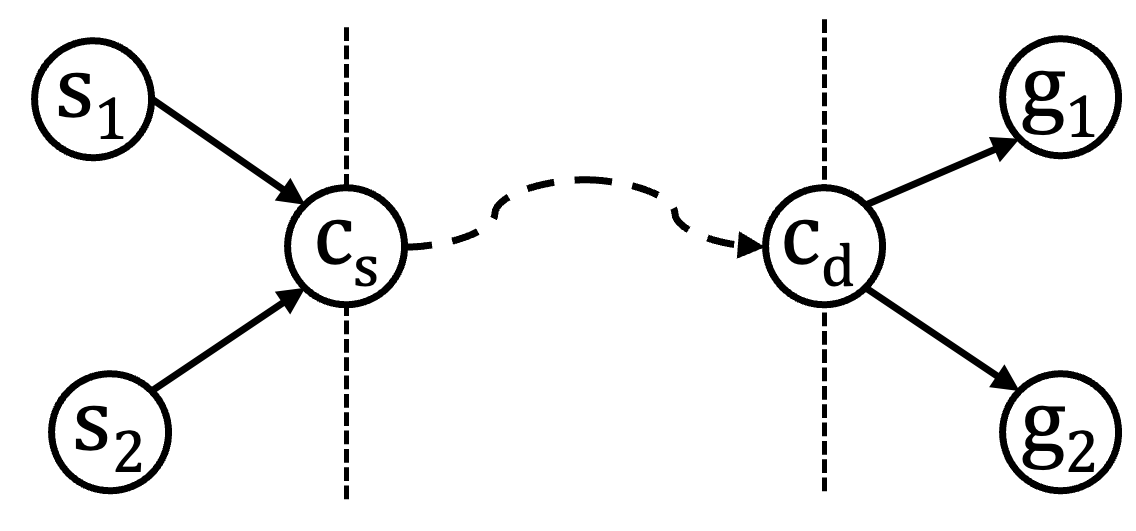}
    \caption{Social welfare optimization occurs when both agents follow the shortest path from their initial nodes to a cooperation node ${c_s}$, then continuously cooperate along a shared path to the final cooperation node ${c_d}$. After reaching ${c_d}$, the agents depart and take the shortest paths to their respective target nodes.}
    \label{fig:phases}
\end{figure}
\noindent Using a shortest path algorithm from ${g_1}$ and ${g_2}$, we can determine the shortest path time from each cooperation node ${c_i} \in V_C$ to ${g_1}$, denoted as $SIP_{c_i, g_1}$, and to ${g_2}$, denoted as $SIP_{c_i, g_2}$. The total path time of these paths, $T_{c_i,{g_1}}(SIP_{c_i, g_1}) + T_{c_i,{g_2}}(SIP_{c_i, g_2})$, is referred to as the \emph{departure value} of node ${c_i}$ and is denoted by $d^*({c_i})$:
\[
d^*({c_i}) = T_{c_i,{g_1}}(SIP_{c_i g_1}) + T_{c_i,{g_2}}(SIP_{c_i, g_2})
\]
To find the \emph{optimal departure node} for a cooperation that starts at a given node ${c_s}$, we identify the departure node that minimizes the combined path time of reaching it cooperatively from $c_s$ and subsequently reaching the target nodes from it:
\[
v_d({c_s}) = \arg\min_{c \in V_C} \left( 2 \cdot \left( T_{{c_s},c}(SCP_{{c_s},c}|SCP_{{c_s},c})+\tau^2_c\right) + d^*(c) \right)
\]
However, determining the optimal node to initiate cooperation, ${c_s}$, may still require evaluating all potential cooperation nodes. While Lemma \ref{early} demonstrates that starting cooperation earlier along a specific path generally results in a shorter overall path time, it cannot be applied to globally compare cooperation nodes, as different nodes may result in distinct subsequent paths. Consequently, a cooperation that begins later but follows a different path may achieve a better overall path time than one that begins earlier (see example in Figure \ref{fig:early-label}).
\begin{figure}[ht]
    \centering
    \includegraphics[width=0.25\linewidth]{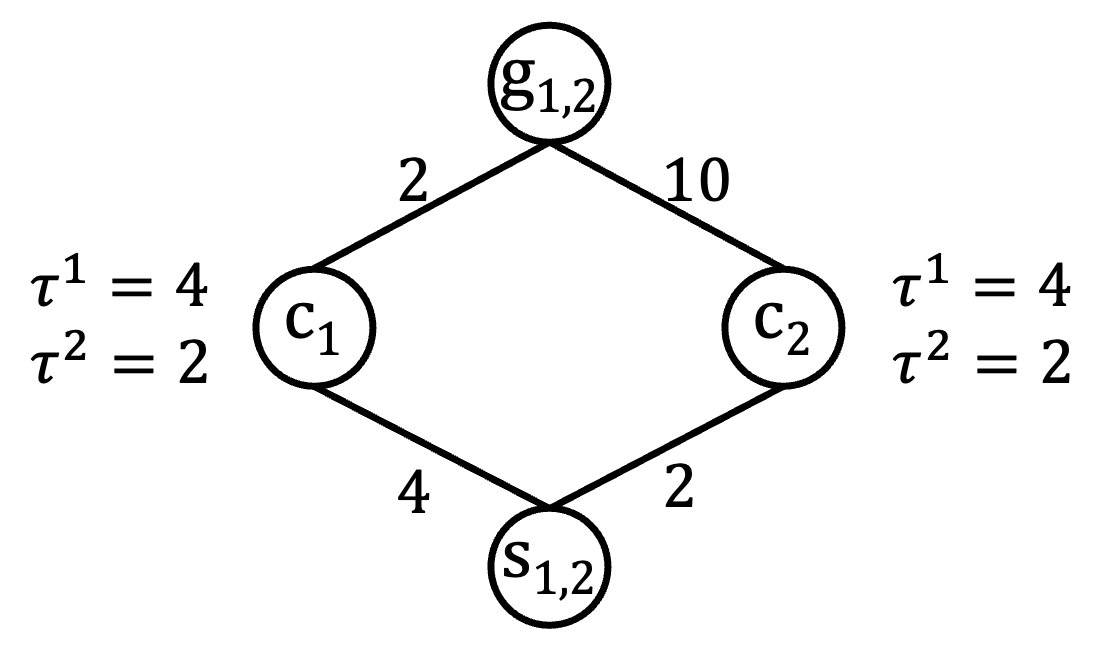}
    \caption{Although $a_1$ and $a_2$ can begin cooperation at ${c_2}$ earlier than at ${c_1}$, the optimal social welfare is achieved when the cooperation starts at ${c_1}$.}
    \label{fig:early-label}
\end{figure}
To improve the performance of finding the optimal cooperation starting node ${c_s}$ we propose a pruning approach that leverages Lemma \ref{early} to reduce the number of evaluated cooperation nodes while ensuring all possible cooperation paths are considered.

Using a shortest path algorithm from ${s_1}$ and ${s_2}$, we can determine the shortest path time to each cooperation node ${c_i} \in V_C$, denoted as $SIP_{s_1,c_i}$ and $SIP_{s_2,c_i}$, respectively.
We denote the \emph{earliest cooperation time} of node ${c_i}$ as $t^*({c_i})$, which represents the earliest shared departure time for two agents starting to cooperate at this node:
\[
t^*({c_i}) = \max \left( T_{s_1, c_i}(SIP_{s_1, c_i}), T_{s_2, c_i}(SIP_{s_2, c_i}) \right) + \tau^2_{c_i}
\]
We sort the cooperation nodes in ascending order based on their \emph{earliest cooperation times} . For each potential cooperation starting node ${c_i}$, we determine its associated \emph{optimal departure node} $v_{d_i}=v_d({c_i})$ and calculate the social welfare of the path dictated by the two nodes:
\[
2 \cdot \left(t^*({c_i}) + T_{c_i,d_i}(SCP_{c_i,d_i}| SCP_{c_i,d_i})+\tau^2_{d_i}\right) + d^*(v_{d_i})
\]
With the goal of finding the optimal cooperation starting node ${c_s}$ that minimizes social welfare:
\[
{c_s}=\arg\min_{{c_i} \in V_C} \left( 2 \cdot \left(t^*({c_i}) + T_{c_i,d_i}(SCP_{c_i,d_i}| SCP_{c_i,d_i})+\tau^2_{v_{d_i}}\right) + d^*(v_{d_i}) \right)
\]
Since evaluating a cooperation starting node ${c_i}$ as a starting node involves finding the shortest cooperative path to every other cooperation node ${c_j} \in V_C$, we compare the arrival time at ${c_j}$ through ${c_i}$ with $t^*({c_j})$. If 
\[
t^*({c_j}) \geq t^*({c_i})  + D_{c_i,c_j}(SCP_{c_i,{c_j}}|SCP_{c_i,{c_j}}) 
\]
then we can prune ${c_j}$ and exclude it as a potential cooperation starting node, thereby reducing the number of nodes to be evaluated (see example in Figure \ref{fig:pruning}).
\begin{figure}[ht]
    \centering
    \includegraphics[width=0.15\linewidth]{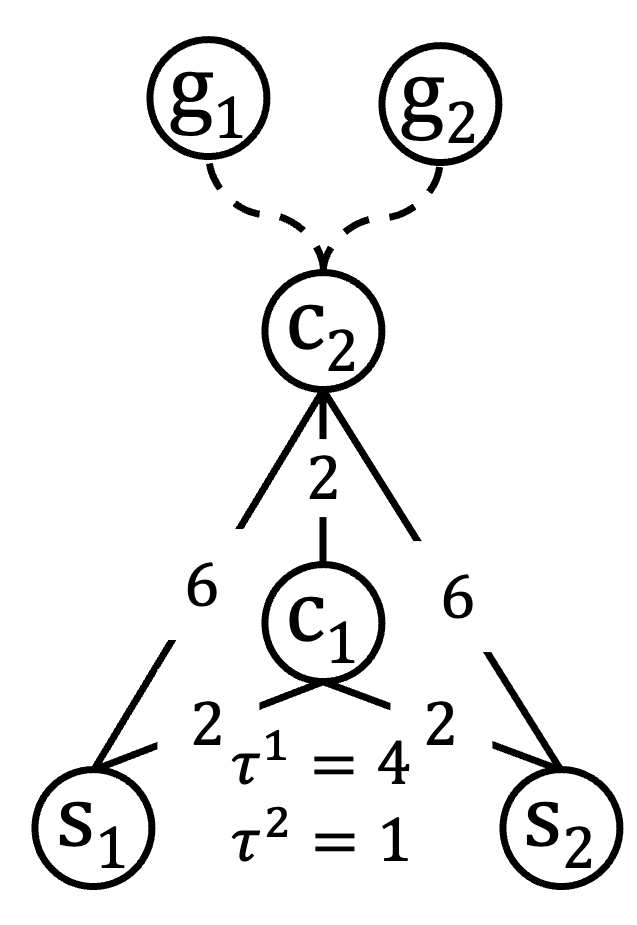}
    \caption{Although $a_1$ and $a_2$ can individually reach ${c_2}$ by $t=6$, starting cooperation earlier at ${c_1}$ allows them to reach ${c_2}$ by $t=5$, improving path efficiency. As a result, ${c_2}$ can be excluded from the set of potential cooperation starting nodes, as initiating cooperation earlier at ${c_1}$ yields a more efficient outcome.}
    \label{fig:pruning}
\end{figure}

Algorithm \ref{algorithm:so-wel-opt} is designed to identify the optimal joint strategy ($\pi^*_1, \pi^*_2$) that minimizes social welfare for two agents navigating a graph with $m$ cooperation nodes.
The algorithm starts by finding $SIP_{s_1}$ and $SIP_{s_2}$ (all shortest independent paths from the initial nodes to other graph nodes), and $SIP_{g_1}$ and $SIP_{g_2}$ (all shortest independent paths to the target nodes from other graph nodes) [lines 1-4]. These paths are then used to compute the \emph{departure value} $d^*({c})$ and the \emph{earliest cooperation time} $t^*({c})$ for each cooperation node ${c} \in V_C$ [lines 5-8].
The algorithm iterates over all cooperation nodes $c_s \in V_C$, in ascending order of their $t^*$ values, evaluating each as a potential cooperation starting node within a continuous cooperation joint strategy $(SIP_{s_1, c_s} \circ SCP_{c_s, c_d} \circ SIP_{c_d, g_1}, SIP_{s_2, c_s} \circ SCP_{c_s, c_d} \circ SIP_{c_d, g_2})$. It first computes the shortest paths from the selected node to all other nodes in the graph, assuming cooperation at all nodes. Then, it evaluates the remaining cooperation nodes as potential cooperation departure nodes by calculating the total path time for both agents when these nodes are designated as the cooperation endpoints [lines 9-26]. If a cooperation node can be reached earlier through cooperation than individually, it is pruned from consideration as a cooperation starting node [lines 16-17]. Finally, the algorithm returns the social welfare optimal joint strategy [line 27].

\begin{algorithm}[ht]
\caption{\textsc{Social Welfare Optimal Path($G, V_C, {s_1}, {s_2}, {g_1}, {g_2}$)}}\label{algorithm:so-wel-opt}
\begin{algorithmic}[1]

    \State $SIP_{s_1} \gets \textsc{all shortest paths from}(G, \tau^1, {s_1})$ \Comment{Find shortest paths from starting nodes}
    \State $SIP_{s_2} \gets \textsc{all shortest paths from}(G, \tau^1, {s_2})$ 
    \State $SIP_{g_1} \gets \textsc{all shortest paths to}(G, \tau^1, {g_1})$ \Comment{Find shortest paths to target nodes}
    \State $SIP_{g_2} \gets \textsc{all shortest paths to}(G, \tau^1, {g_2})$ 
    
    \ForAll {${c} \in V_C$} \Comment{Determine $t^*$ and $d^*$ for all node}
        \State $t^*_{c} \gets \max(SIP_{s_1,c}, SIP_{s_2,c}) + \tau^2_{c}$
        \State $d^*_{c} \gets SIP_{g_1,c} + SIP_{g_2,c}$
    \EndFor
    
    \State $coopStartNodes \gets \textsc{Sort}(V_C, \text{by ascending } t^*)$ \Comment{Sort cooperation nodes by $t^*$}
    \State $\pi^*_1 \gets SIP_{s_1,g_1}, \pi^*_2 \gets SIP_{s_2,g_2}$
    \State $minimalJointPathTime \gets SIP_{s_1,g_1} + SIP_{s_2,g_2}$
    \While{$coopStartNodes \neq \emptyset$}
        \State ${c_s} \gets coopStartNodes.\textsc{pop}()$ \Comment{Select a starting node}
        \State $SCP_{c_s} \gets \textsc{all shortest paths from}(G, \tau^2, {c_s})$ \Comment{Find shortest paths from ${c_s}$}
        
        \ForAll {${c} \in V_C$}
            \If{$t^*_{c_s} + SCP_{c_s,c} + \tau^2_{c} \leq t^*_{c}$}   
                 
                \State $coopStartNodes.\textsc{remove}({c})$ \Comment{Prune non-optimal cooperation starting nodes}           
                \State $totalPathTime \gets 2 \cdot (t^*_{c_s} + SCP_{c_s,c} + \tau^2_{c}) + d^*_{c}$
                \If{$totalPathTime < minimalJointPathTime$} \Comment{Update optimal joint path time}
                    \State $minimalJointPathTime \gets totalPathTime$
                    \State $\pi^*_1 \gets SIP_{s_1,c_s} \circ SCP_{c_s,c} \circ SIP_{c,g_1}$
                    \State $\pi^*_2 \gets SIP_{s_2,c_s} \circ SCP_{c_s,c} \circ SIP_{c,g_2}$
                \EndIf
            \EndIf
        \EndFor
    \EndWhile
    \State \Return $(\pi^*_1, \pi^*_2)$  \Comment{Return social welfare optimal path}
\end{algorithmic}
\end{algorithm}

\noindent To establish the optimality of the algorithm, we verify two key lemmas, analogous to Lemmas \ref{early} and \ref{cont} presented in the Fixed Strategy Assumption scenario:
\begin{enumerate}
\item  \textbf{Cooperation Continuity:} The joint strategy that optimizes social welfare through cooperation is one in which $a_1$ and $a_2$ independently follow the shortest paths to a cooperation starting node ${c_s}$ ($SIP_{s_1,c_s}, SIP_{s_2,c_s}$), cooperate along the joint shortest path to a cooperation departure node ${c_d}$ ($SCP_{c_s,c_d}$), and then each independently follows the shortest path to their respective target nodes ($SIP_{c_d,g_1}, SIP_{c_d,g_2}$).
    \item \textbf{Early Cooperation:} If agents $a_1$ and $a_2$ can reach a cooperation node $c \in V_C$ earlier through cooperation rather than individually, then $c$ is not the first cooperation node in the joint strategy that optimizes social welfare.
\end{enumerate}
\begin{lemma}[Cooperation Continuity]\label{lem_social_1}
Let $\pi^1 \in \Pi_{s_1,g_1}$ and $\pi^2 \in \Pi_{s_2,g_2}$ be two paths that begin cooperation at node ${c_s}$ and finish the cooperation at node ${c_d}$. Then, it holds that:
\[
    T_{g_1}(\pi'^1|\pi'^2) + T_{g_2}(\pi'^2|\pi'^1) \leq T_{g_1}(\pi^1|\pi^2) + T_{g_2}(\pi^2|\pi^1)
\]
Where:
\[
    \pi'^1=SIP_{s_1, c_s} \circ SCP_{c_s, c_d} \circ SIP_{c_d, g_1}
\]
\[
    \pi'^2=SIP_{s_2, c_s} \circ SCP_{c_s, c_d} \circ SIP_{c_d, g_2}
\]
\end{lemma}
\begin{proof}
Since $a_1$ and $a_2$ cooperate at nodes ${c_s}$ and ${c_d}$, the social welfare induced by $\pi'^1$ and $\pi'^2$ is given by:
\[
 T_{g_1}(\pi'^1|\pi'^2) + T_{g_2}(\pi'^2|\pi'^1)=
\]
\[
2 \cdot \left(\max \left( T_{s_1, c_s}(SIP_{s_1, c_s}), T_{s_2, c_s}(SIP_{s_2, c_s}) \right) 
+ \tau^2_{c_s} + T_{c_s,c_d}(SCP_{c_s,c_d}|SCP_{c_s,c_d}) + \tau^2_{c_d}\right) 
\]
\[
+ T_{c_d, g_1}(SIP_{c_d, g_1}) + T_{c_d, g_2}(SIP_{c_d, g_2})
\]
The social welfare induced by $\pi^1$ and $\pi^2$ is:
\[
T_{g_1}(\pi^1|\pi^2) + T_{g_2}(\pi^2|\pi^1) =
\]
\[
2 \cdot \left(\max \left( T_{s_1, c_s}(\pi^1), T_{s_2, c_s}(\pi^2) \right) + \tau^2_{c_s} + \max(T_{c_s, c_d}(\pi^1|\pi^2), T_{c_s, c_d}(\pi^2|\pi^1)) +  \tau^2_{c_d}\right)
\]
\[
 + T_{c_d, g_1}(\pi^1) + T_{c_d, g_2}(\pi^2)
\]

Since the execution time for $a_1$ at each node $v$ visited by $\pi_1$ before ${c_s}$ or after ${c_d}$ is $\tau^1_i$, Then by the definitions of $SIP_{s_1, c_s}$ and $SIP_{c_d, g_1}$ we have:
\[
T_{s_1, c_s}(SIP_{s_1, c_s}) \leq T_{s_1, c_s}(\pi^1)
\]
\[
T_{c_d, g_1}(SIP_{c_d, g_1}) \leq T_{c_d, g_1}(\pi^1)
\]
Similarly:
\[
T_{s_2, c_s}(SIP_{s_2, c_s}) \leq T_{s_2, c_s}(\pi^2)
\]
\[
T_{c_d, g_2}(SIP_{c_d, g_2}) \leq T_{c_d, g_2}(\pi^2)
\]
Therefore:
\[
\max \left( T_{s_1, c_s}(SIP_{s_1, c_s}), T_{s_2, c_s}(SIP_{s_2, c_s}) \right) \leq  \max \left( T_{s_1, c_s}(\pi^1|\pi^2), T_{s_2, c_s}(\pi^2|\pi^1) \right)
\]
Additionally, since $a_1$ and $a_2$ depart from node ${c_s}$ simultaneously, traveling together ensures they arrive at subsequent cooperation nodes at the same time, enabling them to cooperate at those nodes. Therefore, by the definition of $SCP_{c_s, c_d}$:
\[
T_{c_s, c_d}(SCP_{c_s, c_d}|SCP_{c_s, c_d}) \leq T_{c_s, c_d}(\pi^1|\pi^2)
\]
\[
T_{c_s, c_d}(SCP_{c_s, c_d}|SCP_{c_s, c_d}) \leq T_{c_s, c_d}(\pi^2|\pi^1)
\]
Thus:
\[
T_{c_s, c_d}(SCP_{c_s, c_d}|SCP_{c_s, c_d}) \leq \max(T_{c_s, c_d}(\pi^1|\pi^2),T_{c_s, c_d}(\pi^2|\pi^1))
\]
Combining these results:
\[
T_{g_1}(\pi'^1|\pi'^2) + T_{g_2}(\pi'^2|\pi'^1) \leq T_{g_1}(\pi^1|\pi^2) + T_{g_2}(\pi^2|\pi^1)
\]
\end{proof}

\begin{lemma}[Early Cooperation]\label{lem_social_2}
Assume two cooperation nodes ${c_i}, {c_j} \in V_C$ such that $a_1$ and $a_2$ can both reach ${c_i}$ by cooperating at ${c_j}$ no later than they would by traveling directly, i.e.,
\[
t^*({c_i}) \geq t^*({c_j}) + T_{c_j,c_i}(SCP_{c_j,c_i}| SCP_{c_j,c_i}) + \tau^2_{c_i}.
\]
Then, the minimal joint path time possible by a path starting cooperation at ${c_j}$ (achieved by departing from $v_{d_j} = v_d({c_j})$) is at least as low as the minimal joint path time possible by a path starting cooperation at ${c_i}$ (achieved by departing from $v_{d_i} = v_d({c_i})$), i.e.,
\[
2 \cdot \left(t^*({c_j}) + T_{c_j,v_{d_j}}(SCP_{c_j,v_{d_j}}|SCP_{c_j,v_{d_j}}) + \tau^2_{v_{d_j}}\right) + d^*(v_{d_j}) \leq
\]
\[
2 \cdot \left(t^*({c_i}) + T_{c_i,v_{d_i}}(SCP_{c_i,v_{d_i}}|SCP_{c_i,v_{d_i}}) + \tau^2_{v_{d_i}}\right) + d^*(v_{d_i}).
\]
\end{lemma}

\begin{theorem}
Given an \emph{IC2PP} instance with two agents $a_1$ and $a_2$, Algorithm \ref{algorithm:so-wel-opt} finds the joint strategy $(S^{1^*}, S^{2^*})$ that optimizes the social welfare of $a_1$ and $a_2$.
\end{theorem}

\emph{Outline.} This result follows directly from Lemmas \ref{lem_social_1} and \ref{lem_social_2}. Lemma \ref{lem_social_1} establishes that the joint strategy that optimizes social welfare is of the following form:
\[
(\pi^{1^*}, \pi^{2^*}) = (SIP_{s_1, c_s} \circ SCP_{c_s, c_d} \circ SIP_{c_d, g_1}, SIP_{s_2, c_s} \circ SCP_{c_s, c_d} \circ SIP_{c_d, g_2}),
\]
which implies that finding the joint strategy that optimizes social welfare is equivalent to finding the nodes ${c_s}$ and ${c_d}$ that minimize the total path time.

Lemma \ref{lem_social_2} demonstrates that if a cooperation node can be reached earlier through cooperation rather than individually, then this node is not the starting cooperation node in the optimal social welfare joint strategy and can be disregarded.

By combining these results, we conclude that Algorithm \ref{algorithm:so-wel-opt}, which evaluates all possible cooperation nodes ${c} \in V_C$ as potential starting cooperation nodes (excluding those that can be pruned by Lemma \ref{lem_social_2}), identifies their corresponding optimal departure nodes $v_d({c})$, and returns the paths $\pi^{1^*}, \pi^{2^*}$
\[
(\pi^{1^*}, \pi^{2^*}) = (SIP_{s_1, c_s} \circ SCP_{c_s, c_d} \circ SIP_{c_d, g_1}, SIP_{s_2, c_s} \circ SCP_{c_s, c_d} \circ SIP_{c_d, g_2}),
\]
corresponding to:
\[
{c_s}=\arg\min_{{c} \in V_C} \left( 2 \cdot \left(t^*({c}) + T_{c,c_d}(SCP_{c,c_d}| SCP_{c,c_d})+\tau^2_{c_d}\right) + d^*(v_d({c}) \right)
\]
\[
{c_d}=v_d({c_s})
\]
successfully finds the joint strategy $(\pi^{1^*}, \pi^{2^*})$ that optimizes the social welfare of $a_1$ and $a_2$.

\paragraph{Complexity Analysis}
Every run of the algorithm starts with four executions of Dijkstra's algorithm to determine the shortest paths from ${s_1},{s_2}, {g_1}$, and ${g_2}$.
Each run of Dijkstra's algorithm has a time complexity of \( \mathcal{O}(|E| + |V| \log |V|) \). Consequently, the total complexity for these four runs is:
\[
\mathcal{O}(4 \cdot (|E| + |V| \log |V|)) = \mathcal{O}(|E| + |V| \log |V|)
\]
The algorithm then evaluates at most \( m \) potential starting cooperation nodes. Thus, the overall complexity for evaluating these nodes is:
\[
\mathcal{O}(m \cdot (|E| + |V| \log |V|))
\]
Combining these, the total complexity of the algorithm is:
\[
\mathcal{O}(m \cdot (|E| + |V| \log |V|))
\]

\section{Experimental Results}
\subsection{Results Across All Scenarios}\label{appendix:non-filtered-charts}
\begin{figure}[H]
    \centering
    \begin{subfigure}[b]{0.37\linewidth}
        \centering
        \includegraphics[width=\linewidth]{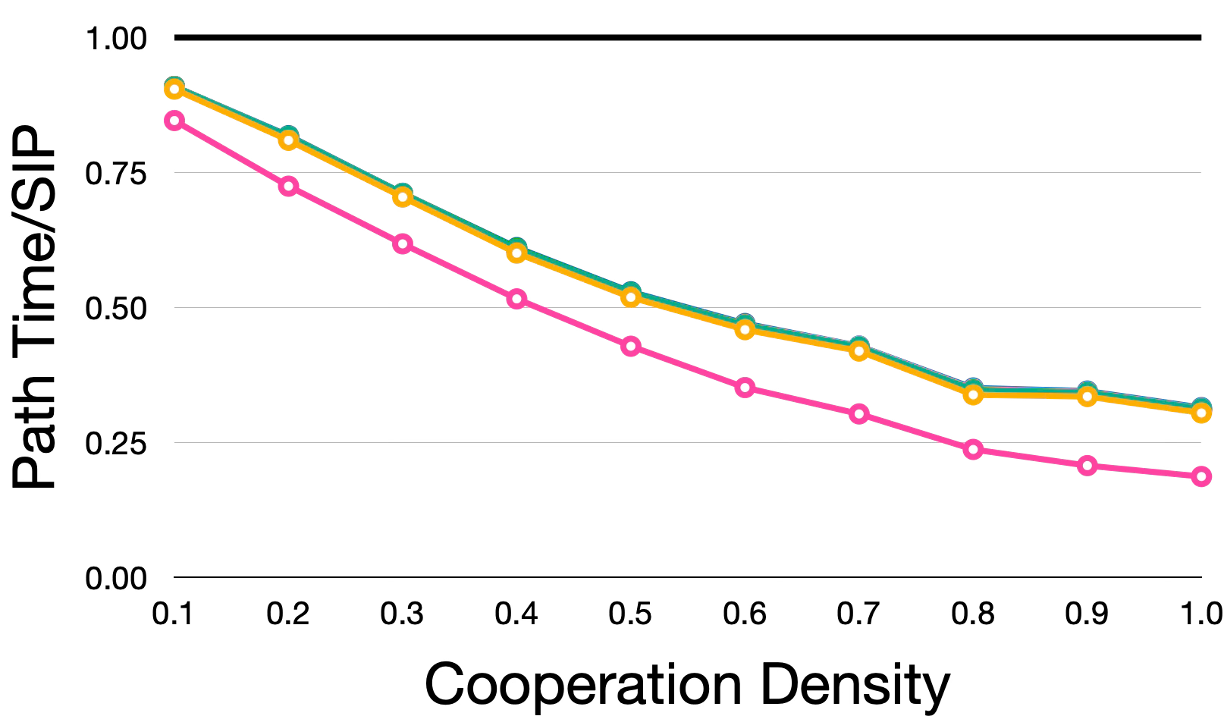}
        \label{fig:res-density-individual-nf}
    \end{subfigure}
    \hfill
    \begin{subfigure}[b]{0.37\linewidth}
        \centering
        \includegraphics[width=\linewidth]{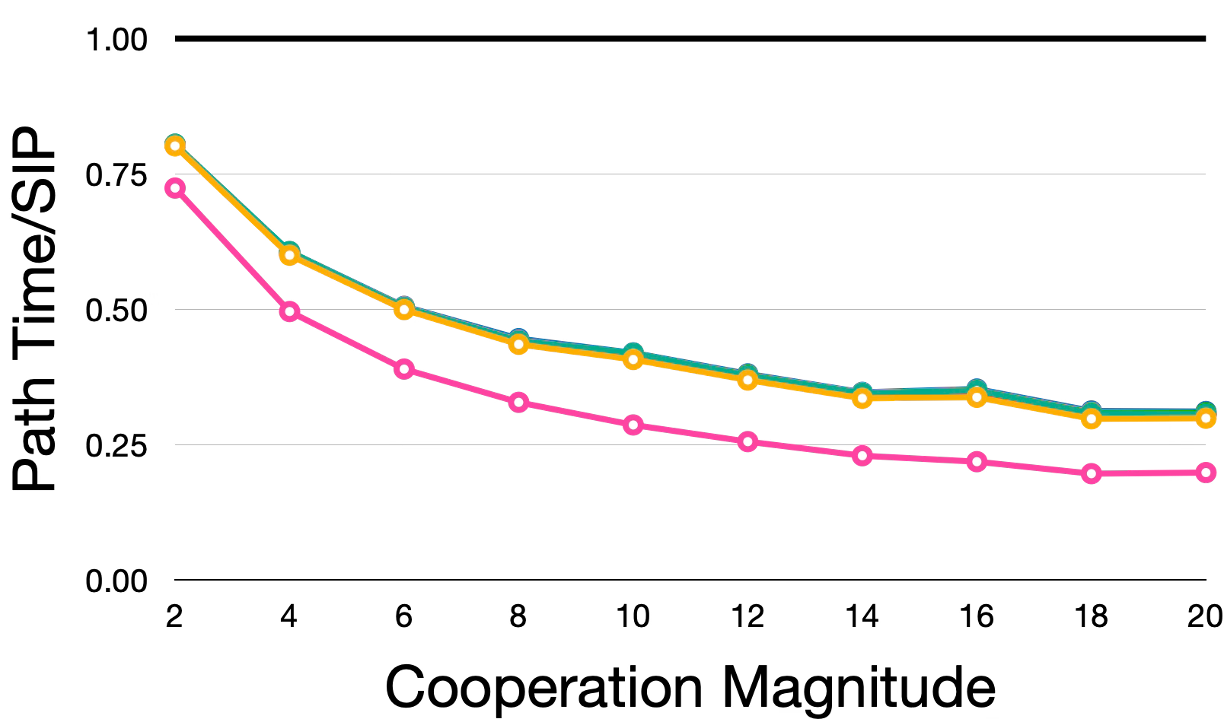}
        \label{fig:res-magnitude-individual-nf}
    \end{subfigure}
    \hfill
     \begin{subfigure}[b]{0.37\linewidth}
        \centering
        \includegraphics[width=\linewidth]{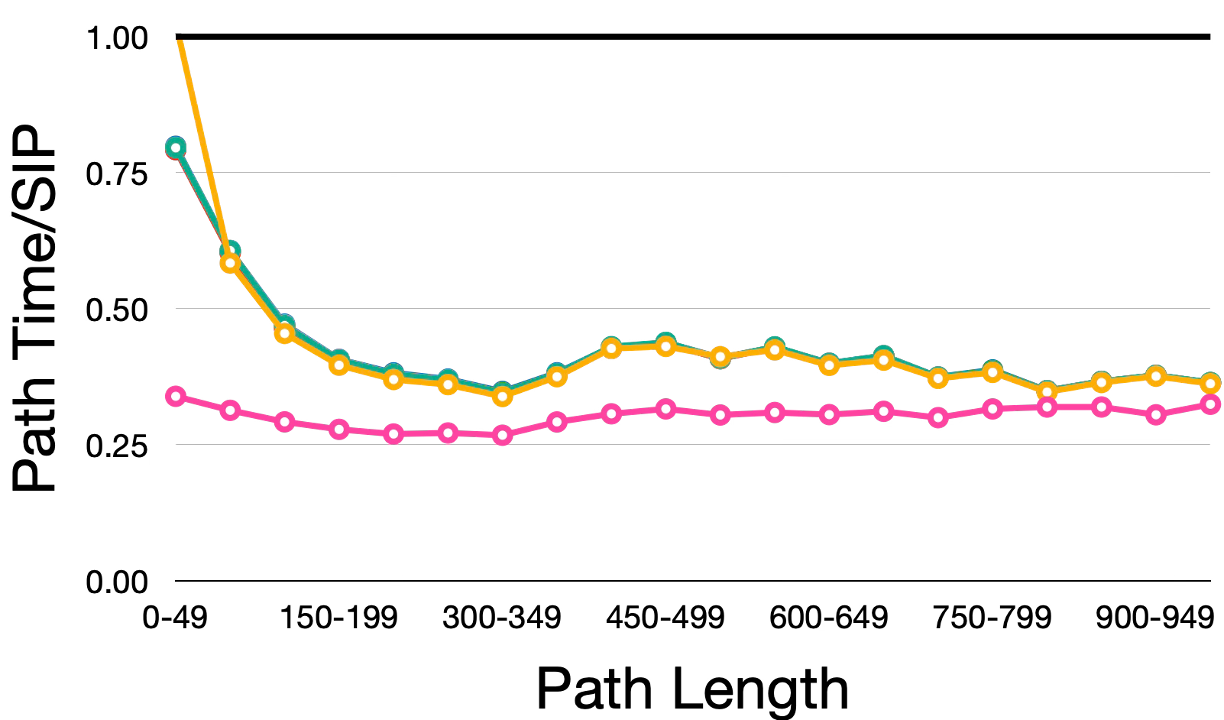}
        \label{fig:res-divergence-individual-nf}
    \end{subfigure}
    \hfill
    \begin{subfigure}[b]{0.37\linewidth}
        \centering
        \includegraphics[width=\linewidth]{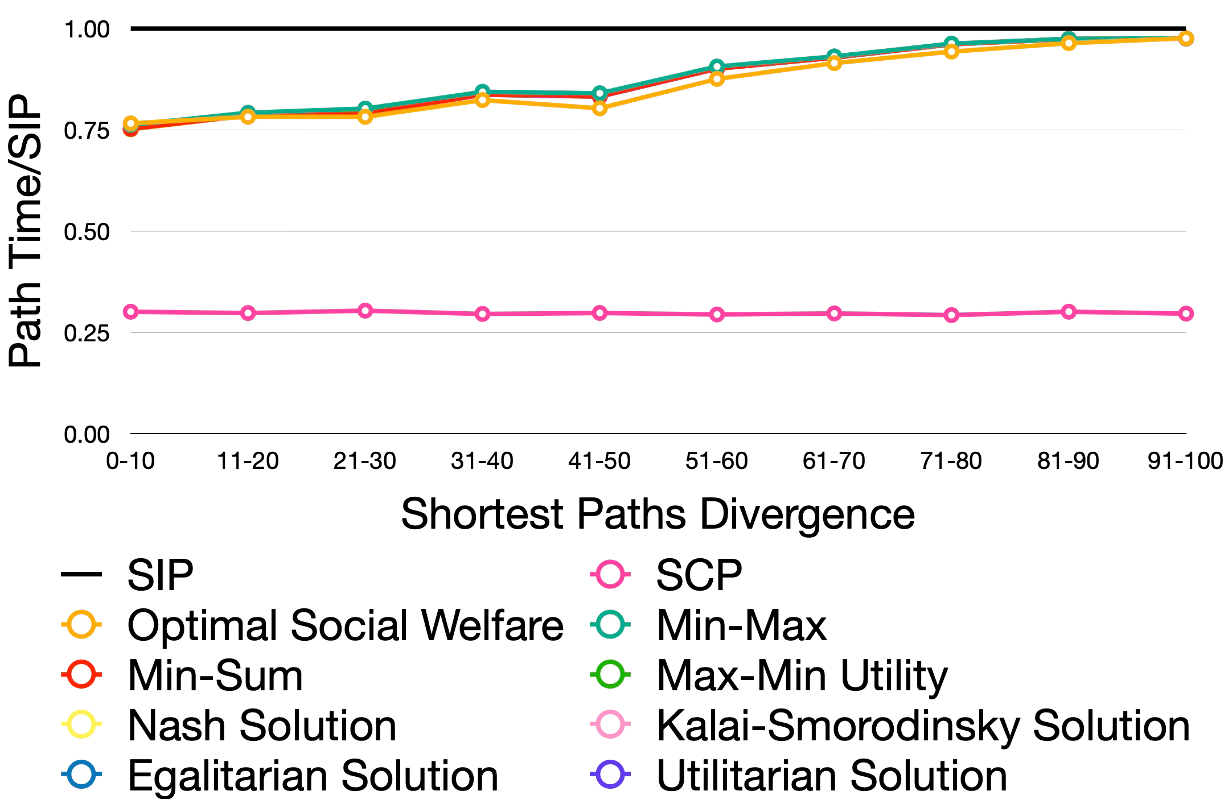}
        \label{fig:res-length-individual-nf}
    \end{subfigure}
    \hfill
    \caption{Effects of cooperation factors on individual path times across all tested scenarios, including those with a single PNE.}
      \label{fig:results-nf}
\end{figure}
\begin{figure}[H]
    \begin{subfigure}[b]{0.37\linewidth}
        \centering
        \includegraphics[width=\linewidth]{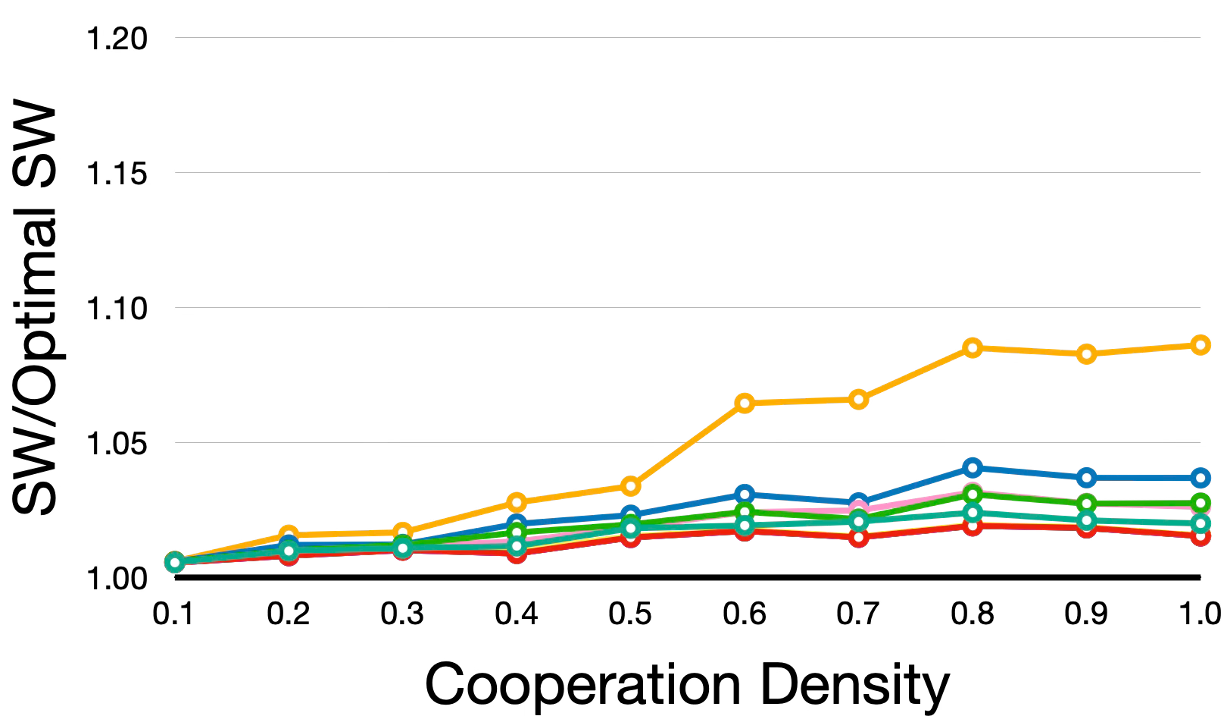}
    \end{subfigure}
    \hfill
    \begin{subfigure}[b]{0.37\linewidth}
        \centering
        \includegraphics[width=\linewidth]{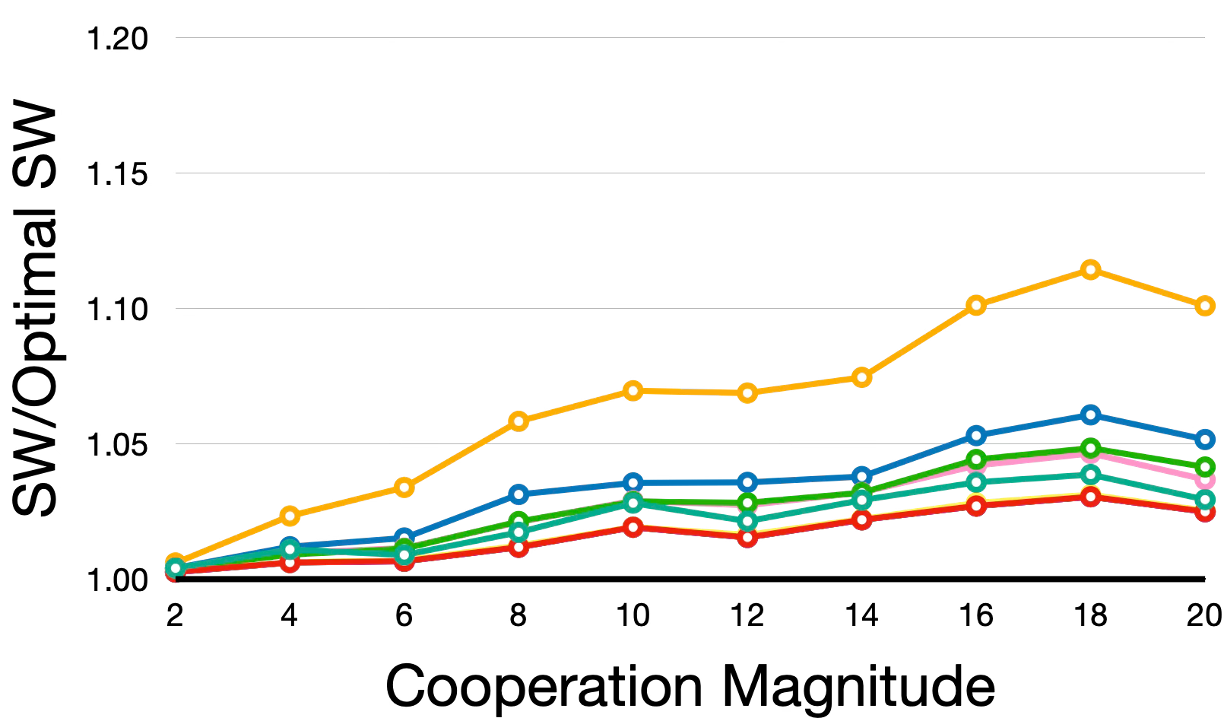}
    \end{subfigure}
    \hfill
    \begin{subfigure}[b]{0.37\linewidth}
        \centering
        \includegraphics[width=\linewidth]{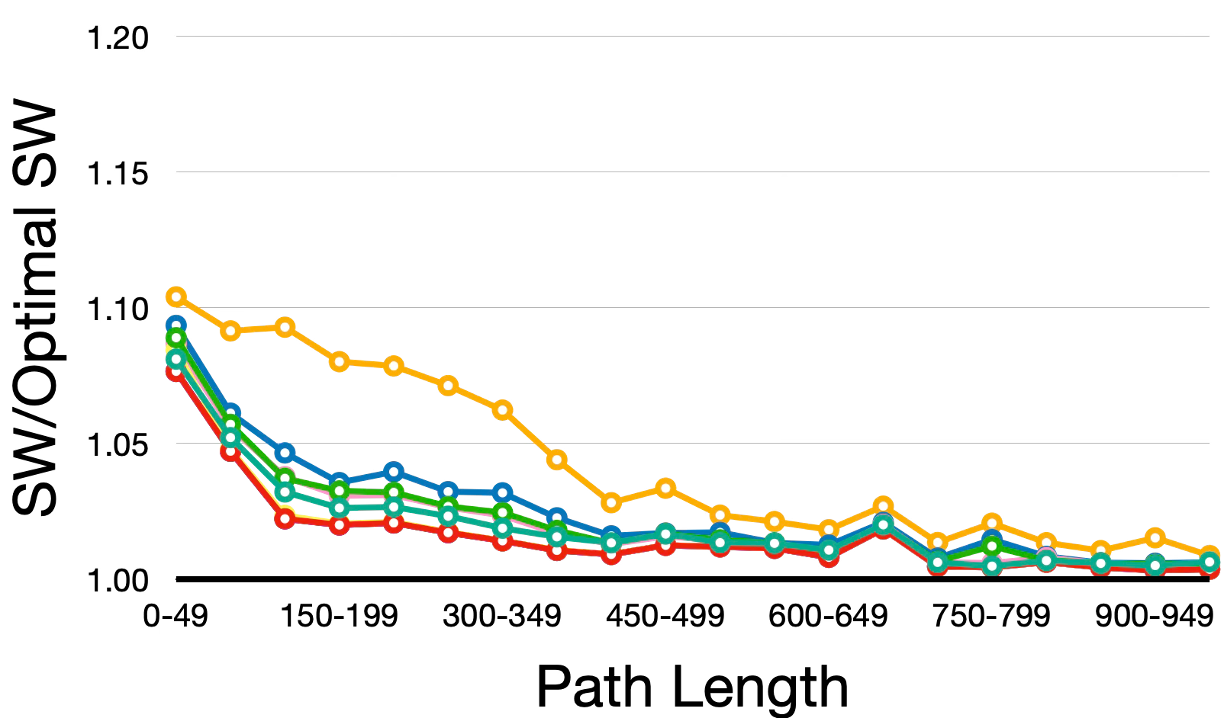}
    \end{subfigure}
    \hfill
    \begin{subfigure}[b]{0.37\linewidth}
        \centering
        \includegraphics[width=\linewidth]{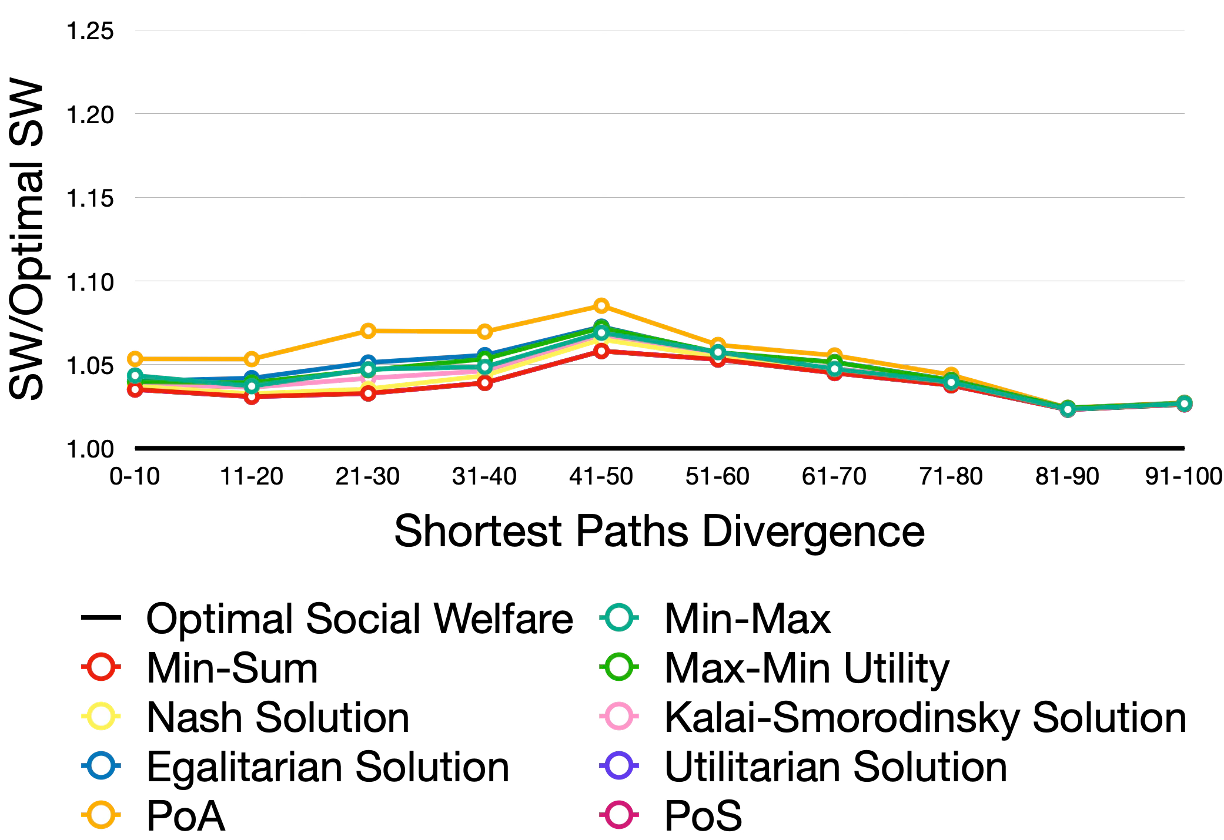}
    \end{subfigure}
    \caption{Effects of cooperation factors on social welfare across all tested scenarios, including those with a single PNE.}
      \label{fig:res-length-social-nf}
\end{figure}

\subsection{Summary Tables for the Experimental Results}\label{appendix:tabular-data}
All values in the table below are reported together with their corresponding standard deviations.\\
\subsubsection{Average Number of PNEs}
\phantom{a}
\begin{table}[H]
\centering
\scriptsize
\caption{Effects of Cooperation Density on average amount of PNEs.}
\label{tab:density-non-filtered-potential-pairs}
\begin{tabular}{@{}lrr@{}}
\toprule
Cooperation Density & Average Number of PNEs \\
\midrule
0.1 & $1.02 \pm 0.14$ \\
0.2 & $1.21 \pm 0.4538$ \\
0.3 & $1.27 \pm 0.5264$ \\
0.4 & $1.4 \pm 0.6164$ \\
0.5 & $1.41 \pm 0.6016$ \\
0.6 & $1.55 \pm 0.6837$ \\
0.7 & $1.65 \pm 0.6837$ \\
0.8 & $2.15 \pm 0.8047$ \\
0.9 & $2.16 \pm 0.9562$ \\
1.0 & $2.58 \pm 1.2343$ \\
\bottomrule
\end{tabular}
\end{table}
\begin{table}[H]
\centering
\scriptsize
\caption{Effects of Cooperation Magnitude on average amount of PNEs.}
\label{tab:magnitude-non-filtered-potential-pairs}
\begin{tabular}{@{}lrr@{}}
\toprule
Cooperation Magnitude & Average Number of PNEs \\
\midrule
2 & $1.18 \pm 0.4094$ \\
4 & $1.5 \pm 0.7416$ \\
6 & $1.62 \pm 0.7454$ \\
8 & $1.76 \pm 0.9069$ \\
10 & $1.77 \pm 0.7329$ \\
12 & $1.94 \pm 0.8924$ \\
14 & $1.78 \pm 0.8436$ \\
16 & $1.9 \pm 0.8888$ \\
18 & $1.99 \pm 0.8185$ \\
20 & $1.88 \pm 0.828$ \\
\bottomrule
\end{tabular}
\end{table}
\begin{table}[H]
\centering
\scriptsize
\caption{Effects of Path Length on average amount of PNEs.}
\label{tab:extent-non-filtered-potential-pairs}
\begin{tabular}{@{}lrr@{}}
\toprule
Path Length & Average Number of PNEs \\
\midrule
0-49 & $1.2439 \pm 0.457$ \\
50-99 & $1.3509 \pm 0.5771$ \\
100-149 & $1.8759 \pm 0.7964$ \\
150-199 & $1.7828 \pm 0.8154$ \\
200-249 & $1.8873 \pm 0.8145$ \\
250-299 & $1.9157 \pm 0.7841$ \\
300-349 & $1.8977 \pm 0.7689$ \\
350-399 & $1.7359 \pm 0.8348$ \\
400-449 & $1.5455 \pm 0.7264$ \\
450-499 & $1.4674 \pm 0.7725$ \\
500-549 & $1.4667 \pm 0.7846$ \\
550-599 & $1.561 \pm 0.964$ \\
600-649 & $1.6591 \pm 0.9521$ \\
650-699 & $1.4595 \pm 0.6407$ \\
700-749 & $1.4167 \pm 0.6821$ \\
750-799 & $1.6 \pm 0.5831$ \\
800-849 & $1.3333 \pm 0.4714$ \\
850-899 & $1.4 \pm 0.6325$ \\
900-949 & $1.68 \pm 1.2238$ \\
950-999 & $1.4 \pm 0.5831$ \\
\bottomrule
\end{tabular}
\end{table}
\begin{table}[H]
\centering
\scriptsize
\caption{Effects of SPD on average amount of PNEs.}
\label{tab:correlation-non-filtered-potential-pairs}
\begin{tabular}{@{}lrr@{}}
\toprule
SPD & Average Number of PNEs \\
\midrule
0-10 & $1.3996 \pm 0.7276$ \\
11-20 & $1.4372 \pm 0.7351$ \\
21-30 & $1.52 \pm 0.8364$ \\
31-40 & $1.4271 \pm 0.7394$ \\
41-50 & $1.4167 \pm 0.7426$ \\
51-60 & $1.2552 \pm 0.6865$ \\
61-70 & $1.2267 \pm 0.5611$ \\
71-80 & $1.1152 \pm 0.4459$ \\
81-90 & $1.0471 \pm 0.2118$ \\
91-100 & $1.0472 \pm 0.212$ \\
\bottomrule
\end{tabular}
\end{table}

\subsubsection{Individual path times in multiple-PNE scenarios}
\phantom{a}
\begin{table}[H]
\centering
\scriptsize
\caption{Effect of Cooperation Density on individual path times among scenarios with more than one PNE.}
\label{tab:density-filtered-individual}
\resizebox{\textwidth}{!}{
\begin{tabular}{@{}lrrrrrrrrrr@{}}
\toprule
Cooperation Density & SIP & SCP & Optimal Social Welfare & Min-Max & Min-Sum & Max-Min-Surplus & Nash & Kalai\_Smorodinsky & Egalitarian & Utilitarian \\
\midrule
0.1 & $1 \pm 0$ & $0.8072 \pm 0.0562$ & $0.8816 \pm 0.0674$ & $0.8878 \pm 0.0644$ & $0.8878 \pm 0.0644$ & $0.9024 \pm 0.0332$ & $0.8878 \pm 0.0644$ & $0.8935 \pm 0.0544$ & $0.9024 \pm 0.0332$ & $0.8878 \pm 0.0644$ \\
0.2 & $1 \pm 0$ & $0.705 \pm 0.0635$ & $0.786 \pm 0.0797$ & $0.8007 \pm 0.0808$ & $0.7924 \pm 0.0761$ & $0.7996 \pm 0.067$ & $0.794 \pm 0.0724$ & $0.7994 \pm 0.0684$ & $0.8074 \pm 0.0608$ & $0.7924 \pm 0.0761$ \\
0.3 & $1 \pm 0$ & $0.6072 \pm 0.073$ & $0.6838 \pm 0.0983$ & $0.6952 \pm 0.0947$ & $0.692 \pm 0.0971$ & $0.6966 \pm 0.0942$ & $0.692 \pm 0.0929$ & $0.6962 \pm 0.0912$ & $0.6975 \pm 0.0914$ & $0.6914 \pm 0.0959$ \\
0.4 & $1 \pm 0$ & $0.5079 \pm 0.083$ & $0.6038 \pm 0.1185$ & $0.6174 \pm 0.1221$ & $0.6116 \pm 0.1204$ & $0.6229 \pm 0.105$ & $0.6131 \pm 0.1128$ & $0.6185 \pm 0.1073$ & $0.6279 \pm 0.0993$ & $0.6115 \pm 0.1203$ \\
0.5 & $1 \pm 0$ & $0.4031 \pm 0.0946$ & $0.4816 \pm 0.1126$ & $0.4911 \pm 0.1117$ & $0.4854 \pm 0.1137$ & $0.4918 \pm 0.1087$ & $0.4857 \pm 0.1082$ & $0.4902 \pm 0.1078$ & $0.4964 \pm 0.1081$ & $0.4851 \pm 0.1135$ \\
0.6 & $1 \pm 0$ & $0.322 \pm 0.0651$ & $0.4144 \pm 0.087$ & $0.4235 \pm 0.0908$ & $0.4211 \pm 0.0901$ & $0.4261 \pm 0.0856$ & $0.4212 \pm 0.0876$ & $0.4275 \pm 0.0852$ & $0.4319 \pm 0.0836$ & $0.4211 \pm 0.0895$ \\
0.7 & $1 \pm 0$ & $0.2818 \pm 0.0593$ & $0.3886 \pm 0.1201$ & $0.3996 \pm 0.1248$ & $0.3938 \pm 0.1227$ & $0.3985 \pm 0.1247$ & $0.394 \pm 0.1217$ & $0.4015 \pm 0.1181$ & $0.4024 \pm 0.12$ & $0.3938 \pm 0.1228$ \\
\bottomrule
\end{tabular}
}
\end{table}
\begin{table}[H]
\centering
\scriptsize
\caption{Effect of Cooperation Magnitude on individual path times among scenarios with more than one PNE.}\label{tab:magnitude-filtered-individual}
\resizebox{\textwidth}{!}{%
\begin{tabular}{@{}lrrrrrrrrrr@{}}
\toprule
Cooperation Magnitude & SIP & SCP & Optimal Social Welfare & Min-Max & Min-Sum & Max-Min-Surplus & Nash & Kalai\_Smorodinsky & Egalitarian & Utilitarian \\
\midrule
2 & $1 \pm 0$ & $0.7368 \pm 0.0298$ & $0.8307 \pm 0.0746$ & $0.8375 \pm 0.0576$ & $0.8299 \pm 0.0639$ & $0.8357 \pm 0.0521$ & $0.8313 \pm 0.0594$ & $0.8344 \pm 0.0583$ & $0.8366 \pm 0.0497$ & $0.8301 \pm 0.0665$ \\
4 & $1 \pm 0$ & $0.497 \pm 0.0456$ & $0.592 \pm 0.0803$ & $0.6044 \pm 0.0892$ & $0.5943 \pm 0.086$ & $0.5987 \pm 0.0754$ & $0.5941 \pm 0.0785$ & $0.6006 \pm 0.078$ & $0.6033 \pm 0.0712$ & $0.5942 \pm 0.0814$ \\
6 & $1 \pm 0$ & $0.3774 \pm 0.0566$ & $0.4869 \pm 0.1155$ & $0.492 \pm 0.1174$ & $0.4887 \pm 0.1169$ & $0.4924 \pm 0.1162$ & $0.489 \pm 0.116$ & $0.494 \pm 0.116$ & $0.4962 \pm 0.1136$ & $0.4888 \pm 0.1164$ \\
8 & $1 \pm 0$ & $0.3197 \pm 0.0643$ & $0.4076 \pm 0.0826$ & $0.4169 \pm 0.0852$ & $0.412 \pm 0.0856$ & $0.4172 \pm 0.0798$ & $0.4121 \pm 0.0838$ & $0.4186 \pm 0.0788$ & $0.4243 \pm 0.0793$ & $0.412 \pm 0.0856$ \\
10 & $1 \pm 0$ & $0.2776 \pm 0.059$ & $0.3682 \pm 0.0892$ & $0.3813 \pm 0.0968$ & $0.3728 \pm 0.0905$ & $0.3779 \pm 0.0872$ & $0.3728 \pm 0.0901$ & $0.3798 \pm 0.0876$ & $0.3819 \pm 0.0845$ & $0.373 \pm 0.0914$ \\
12 & $1 \pm 0$ & $0.2493 \pm 0.0638$ & $0.3608 \pm 0.1131$ & $0.3713 \pm 0.1147$ & $0.3667 \pm 0.1157$ & $0.3731 \pm 0.1123$ & $0.367 \pm 0.1134$ & $0.3737 \pm 0.1106$ & $0.3767 \pm 0.11$ & $0.3667 \pm 0.1156$ \\
14 & $1 \pm 0$ & $0.2223 \pm 0.0555$ & $0.3156 \pm 0.0942$ & $0.3265 \pm 0.0952$ & $0.3213 \pm 0.0918$ & $0.3252 \pm 0.0881$ & $0.3214 \pm 0.091$ & $0.3271 \pm 0.0882$ & $0.3286 \pm 0.0866$ & $0.3213 \pm 0.0918$ \\
\bottomrule
\end{tabular}
}
\end{table}
\begin{table}[H]
\centering
\scriptsize
\caption{Effect of Path Length on individual path times among scenarios with more than one PNE.}\label{tab:extent-filtered-individual}
\resizebox{\textwidth}{!}{%
\begin{tabular}{@{}lrrrrrrrrrr@{}}
\toprule
Path Length & SIP & SCP & Optimal Social Welfare & Min-Max & Min-Sum & Max-Min-Surplus & Nash & Kalai\_Smorodinsky & Egalitarian & Utilitarian \\
\midrule
0-49 & $1 \pm 0$ & $0.2667 \pm 0.0814$ & $0.6116 \pm 0.1568$ & $0.6408 \pm 0.1606$ & $0.6267 \pm 0.1654$ & $0.642 \pm 0.1355$ & $0.6328 \pm 0.1312$ & $0.6441 \pm 0.1233$ & $0.653 \pm 0.1174$ & $0.6267 \pm 0.1654$ \\
50-99 & $1 \pm 0$ & $0.2838 \pm 0.0834$ & $0.4968 \pm 0.1369$ & $0.5224 \pm 0.1429$ & $0.5086 \pm 0.1329$ & $0.517 \pm 0.1199$ & $0.5091 \pm 0.129$ & $0.5218 \pm 0.1177$ & $0.5232 \pm 0.1143$ & $0.5106 \pm 0.1383$ \\
100-149 & $1 \pm 0$ & $0.2831 \pm 0.0715$ & $0.4255 \pm 0.1059$ & $0.4439 \pm 0.118$ & $0.4346 \pm 0.111$ & $0.4419 \pm 0.1095$ & $0.4348 \pm 0.1078$ & $0.4452 \pm 0.103$ & $0.4478 \pm 0.1061$ & $0.4346 \pm 0.1105$ \\
150-199 & $1 \pm 0$ & $0.2678 \pm 0.0559$ & $0.3709 \pm 0.0822$ & $0.385 \pm 0.0883$ & $0.3786 \pm 0.0808$ & $0.3849 \pm 0.0785$ & $0.3785 \pm 0.0786$ & $0.386 \pm 0.0748$ & $0.3867 \pm 0.0753$ & $0.3785 \pm 0.0811$ \\
200-249 & $1 \pm 0$ & $0.2603 \pm 0.0559$ & $0.3515 \pm 0.0921$ & $0.3611 \pm 0.0966$ & $0.3568 \pm 0.0941$ & $0.3617 \pm 0.0909$ & $0.3568 \pm 0.0923$ & $0.3625 \pm 0.0897$ & $0.3658 \pm 0.0898$ & $0.3569 \pm 0.094$ \\
250-299 & $1 \pm 0$ & $0.2635 \pm 0.0556$ & $0.34 \pm 0.0907$ & $0.3488 \pm 0.0938$ & $0.3448 \pm 0.0922$ & $0.3485 \pm 0.0894$ & $0.3446 \pm 0.0907$ & $0.3494 \pm 0.0888$ & $0.3513 \pm 0.0884$ & $0.3445 \pm 0.0916$ \\
300-349 & $1 \pm 0$ & $0.2594 \pm 0.0534$ & $0.3199 \pm 0.0706$ & $0.3267 \pm 0.0724$ & $0.3242 \pm 0.073$ & $0.3287 \pm 0.0716$ & $0.3242 \pm 0.0723$ & $0.3286 \pm 0.0694$ & $0.332 \pm 0.0697$ & $0.3242 \pm 0.0729$ \\
\bottomrule
\end{tabular}
}
\end{table}
\begin{table}[H]
\centering
\scriptsize
\caption{Effect of SPD on individual path times among scenarios with more than one PNE.}\label{tab:correlation-filtered-individual}
\resizebox{\textwidth}{!}{%
\begin{tabular}{@{}lrrrrrrrrrr@{}}
\toprule
SPD & SIP & SCP & Optimal Social Welfare & Min-Max & Min-Sum & Max-Min-Surplus & Nash & Kalai\_Smorodinsky & Egalitarian & Utilitarian \\
\midrule
0-10 & $1 \pm 0$ & $0.2776 \pm 0.0718$ & $0.5997 \pm 0.172$ & $0.6421 \pm 0.1947$ & $0.6058 \pm 0.1686$ & $0.6047 \pm 0.1583$ & $0.6013 \pm 0.1571$ & $0.6094 \pm 0.1566$ & $0.6053 \pm 0.1553$ & $0.6056 \pm 0.1687$ \\
11-20 & $1 \pm 0$ & $0.283 \pm 0.0734$ & $0.6005 \pm 0.1805$ & $0.6304 \pm 0.19$ & $0.6102 \pm 0.1798$ & $0.6157 \pm 0.1663$ & $0.6091 \pm 0.1719$ & $0.6168 \pm 0.1667$ & $0.6204 \pm 0.1638$ & $0.6088 \pm 0.1774$ \\
21-30 & $1 \pm 0$ & $0.2791 \pm 0.0729$ & $0.6145 \pm 0.176$ & $0.6598 \pm 0.1755$ & $0.621 \pm 0.1741$ & $0.639 \pm 0.1452$ & $0.6215 \pm 0.1527$ & $0.6356 \pm 0.1396$ & $0.6472 \pm 0.1322$ & $0.621 \pm 0.1741$ \\
31-40 & $1 \pm 0$ & $0.2789 \pm 0.0637$ & $0.6402 \pm 0.2109$ & $0.6654 \pm 0.1651$ & $0.6428 \pm 0.1787$ & $0.6639 \pm 0.1489$ & $0.6463 \pm 0.1531$ & $0.6529 \pm 0.1498$ & $0.6665 \pm 0.1402$ & $0.6428 \pm 0.1787$ \\
41-50 & $1 \pm 0$ & $0.2787 \pm 0.0772$ & $0.6567 \pm 0.1943$ & $0.6936 \pm 0.1709$ & $0.6673 \pm 0.2022$ & $0.6898 \pm 0.1589$ & $0.6764 \pm 0.1636$ & $0.6821 \pm 0.1629$ & $0.6919 \pm 0.1544$ & $0.6662 \pm 0.2017$ \\
51-60 & $1 \pm 0$ & $0.2621 \pm 0.0529$ & $0.6726 \pm 0.1764$ & $0.7217 \pm 0.1641$ & $0.6855 \pm 0.1641$ & $0.7029 \pm 0.1533$ & $0.6892 \pm 0.1485$ & $0.6981 \pm 0.1466$ & $0.7001 \pm 0.1476$ & $0.6856 \pm 0.1656$ \\
61-70 & $1 \pm 0$ & $0.2728 \pm 0.0728$ & $0.6895 \pm 0.1652$ & $0.7073 \pm 0.1563$ & $0.6894 \pm 0.1567$ & $0.7049 \pm 0.1364$ & $0.6934 \pm 0.1385$ & $0.6982 \pm 0.1393$ & $0.706 \pm 0.1357$ & $0.6894 \pm 0.1567$ \\
\bottomrule
\end{tabular}
}
\end{table}
\subsubsection{Social welfare in multiple-PNE scenarios}
\phantom{a}
\begin{table}[H]
\centering
\scriptsize
\caption{Effects of Cooperation Density on social welfare among scenarios with more than one PNE.}
\label{tab:density-filtered-social}
\resizebox{\textwidth}{!}{%
\begin{tabular}{@{}lrrrrrrrrrrr@{}}
\toprule
Cooperation Density & Optimal Social Welfare & Min-Max & Min-Sum & Max-Min Improvement & Nash & Kalai\_Smorodinsky & Egalitarian & Utilitarian & SIP|SIP & PoA & PoS \\
\midrule
0.1 & $1 \pm 0$ & $1.0067 \pm 0.0067$ & $1.0067 \pm 0.0067$ & $1.0255 \pm 0.0121$ & $1.0067 \pm 0.0067$ & $1.0139 \pm 0.0011$ & $1.0255 \pm 0.0121$ & $1.0067 \pm 0.0067$ & $1.1384 \pm 0.0466$ & $1.03 \pm 0.0076$ & $1.0067 \pm 0.0067$ \\
0.2 & $1 \pm 0$ & $1.018 \pm 0.0211$ & $1.0088 \pm 0.0136$ & $1.0192 \pm 0.0217$ & $1.0115 \pm 0.0136$ & $1.0181 \pm 0.0151$ & $1.0297 \pm 0.0244$ & $1.0088 \pm 0.0136$ & $1.2836 \pm 0.0998$ & $1.0489 \pm 0.0343$ & $1.0088 \pm 0.0136$ \\
0.3 & $1 \pm 0$ & $1.0149 \pm 0.0171$ & $1.0118 \pm 0.017$ & $1.0202 \pm 0.0226$ & $1.0128 \pm 0.0178$ & $1.018 \pm 0.0172$ & $1.021 \pm 0.0226$ & $1.0118 \pm 0.017$ & $1.4857 \pm 0.1778$ & $1.0404 \pm 0.0258$ & $1.0118 \pm 0.017$ \\
0.4 & $1 \pm 0$ & $1.0218 \pm 0.0257$ & $1.0135 \pm 0.0195$ & $1.037 \pm 0.0432$ & $1.0169 \pm 0.0216$ & $1.0269 \pm 0.0187$ & $1.0466 \pm 0.039$ & $1.0135 \pm 0.0195$ & $1.7153 \pm 0.3258$ & $1.0706 \pm 0.0468$ & $1.0135 \pm 0.0195$ \\
0.5 & $1 \pm 0$ & $1.0176 \pm 0.0297$ & $1.0079 \pm 0.0138$ & $1.0217 \pm 0.0304$ & $1.0093 \pm 0.0138$ & $1.0184 \pm 0.0143$ & $1.0316 \pm 0.0327$ & $1.0079 \pm 0.0138$ & $2.1641 \pm 0.4262$ & $1.0622 \pm 0.0518$ & $1.0079 \pm 0.0138$ \\
0.6 & $1 \pm 0$ & $1.022 \pm 0.0281$ & $1.0174 \pm 0.0222$ & $1.0335 \pm 0.0344$ & $1.0182 \pm 0.0234$ & $1.0332 \pm 0.0262$ & $1.0481 \pm 0.0477$ & $1.0174 \pm 0.0222$ & $2.5041 \pm 0.4661$ & $1.1249 \pm 0.1005$ & $1.0174 \pm 0.0222$ \\
0.7 & $1 \pm 0$ & $1.027 \pm 0.041$ & $1.016 \pm 0.0197$ & $1.0288 \pm 0.0357$ & $1.0165 \pm 0.02$ & $1.0348 \pm 0.0251$ & $1.0401 \pm 0.0358$ & $1.016 \pm 0.0197$ & $2.7322 \pm 0.5806$ & $1.1124 \pm 0.0848$ & $1.016 \pm 0.0197$ \\
\bottomrule
\end{tabular}
}
\end{table}
\begin{table}[H]
\centering
\scriptsize
\caption{Effects of Magnitude on social welfare among scenarios with more than one PNE.}
\label{tab:magnitude-filtered-social}
\resizebox{\textwidth}{!}{%
\begin{tabular}{@{}lrrrrrrrrrrr@{}}
\toprule
Cooperation Magnitude & Optimal Social Welfare & Min-Max & Min-Sum & Max-Min Improvement & Nash & Kalai\_Smorodinsky & Egalitarian & Utilitarian & SIP|SIP & PoA & PoS \\
\midrule
2 & $1 \pm 0$ & $1.0083 \pm 0.0124$ & $1 \pm 0$ & $1.0077 \pm 0.0095$ & $1.0017 \pm 0.0049$ & $1.0053 \pm 0.0045$ & $1.0087 \pm 0.0097$ & $1 \pm 0$ & $1.2081 \pm 0.0573$ & $1.0197 \pm 0.016$ & $1 \pm 0$ \\
4 & $1 \pm 0$ & $1.0174 \pm 0.0281$ & $1.004 \pm 0.0166$ & $1.0121 \pm 0.0258$ & $1.0041 \pm 0.0165$ & $1.0137 \pm 0.0165$ & $1.0202 \pm 0.0281$ & $1.004 \pm 0.0166$ & $1.7107 \pm 0.1889$ & $1.0517 \pm 0.038$ & $1.004 \pm 0.0166$ \\
6 & $1 \pm 0$ & $1.0101 \pm 0.0249$ & $1.0052 \pm 0.0075$ & $1.0151 \pm 0.0237$ & $1.0061 \pm 0.0083$ & $1.0155 \pm 0.0075$ & $1.0235 \pm 0.0284$ & $1.0052 \pm 0.0075$ & $2.1341 \pm 0.3834$ & $1.0632 \pm 0.056$ & $1.0052 \pm 0.0075$ \\
8 & $1 \pm 0$ & $1.0214 \pm 0.0255$ & $1.0116 \pm 0.0153$ & $1.0285 \pm 0.03$ & $1.0127 \pm 0.0155$ & $1.0281 \pm 0.0203$ & $1.0464 \pm 0.0444$ & $1.0116 \pm 0.0153$ & $2.5248 \pm 0.3835$ & $1.0946 \pm 0.0734$ & $1.0116 \pm 0.0153$ \\
10 & $1 \pm 0$ & $1.0306 \pm 0.0462$ & $1.016 \pm 0.02$ & $1.0316 \pm 0.0416$ & $1.0161 \pm 0.0201$ & $1.0321 \pm 0.0255$ & $1.0429 \pm 0.0448$ & $1.016 \pm 0.02$ & $2.8271 \pm 0.5381$ & $1.0987 \pm 0.0719$ & $1.016 \pm 0.02$ \\
12 & $1 \pm 0$ & $1.0279 \pm 0.0357$ & $1.0186 \pm 0.0233$ & $1.0383 \pm 0.039$ & $1.0201 \pm 0.0233$ & $1.037 \pm 0.0278$ & $1.0499 \pm 0.0449$ & $1.0186 \pm 0.0233$ & $2.963 \pm 0.6845$ & $1.1007 \pm 0.0768$ & $1.0186 \pm 0.0233$ \\
14 & $1 \pm 0$ & $1.0354 \pm 0.0487$ & $1.0229 \pm 0.0314$ & $1.0401 \pm 0.0439$ & $1.0233 \pm 0.0317$ & $1.0403 \pm 0.034$ & $1.0504 \pm 0.0492$ & $1.0229 \pm 0.0314$ & $3.3515 \pm 0.688$ & $1.1136 \pm 0.0931$ & $1.0229 \pm 0.0314$ \\
\bottomrule
\end{tabular}
}
\end{table}
\begin{table}[H]
\centering
\scriptsize
\caption{Effects of Path Length on social welfare among scenarios with more than one PNE.}
\label{tab:extent-filtered-social}
\resizebox{\textwidth}{!}{%
\begin{tabular}{@{}lrrrrrrrrrrr@{}}
\toprule
Path Length & Optimal Social Welfare & Min-Max & Min-Sum & Max-Min Improvement & Nash & Kalai\_Smorodinsky & Egalitarian & Utilitarian & SIP|SIP & PoA & PoS \\
\midrule
0-49 & $1 \pm 0$ & $1.0506 \pm 0.084$ & $1.0312 \pm 0.0408$ & $1.0847 \pm 0.083$ & $1.0655 \pm 0.0627$ & $1.0757 \pm 0.0505$ & $1.1043 \pm 0.0742$ & $1.0312 \pm 0.0408$ & $1.737 \pm 0.2708$ & $1.1497 \pm 0.082$ & $1.0312 \pm 0.0408$ \\
50-99 & $1 \pm 0$ & $1.0508 \pm 0.0427$ & $1.034 \pm 0.0293$ & $1.0672 \pm 0.0782$ & $1.0363 \pm 0.0313$ & $1.063 \pm 0.0466$ & $1.081 \pm 0.082$ & $1.034 \pm 0.0293$ & $2.106 \pm 0.3886$ & $1.1827 \pm 0.1206$ & $1.034 \pm 0.0293$ \\
100-149 & $1 \pm 0$ & $1.0383 \pm 0.0498$ & $1.0231 \pm 0.027$ & $1.046 \pm 0.0426$ & $1.025 \pm 0.0268$ & $1.047 \pm 0.0302$ & $1.0605 \pm 0.0522$ & $1.0231 \pm 0.027$ & $2.4494 \pm 0.4404$ & $1.1318 \pm 0.0779$ & $1.0231 \pm 0.027$ \\
150-199 & $1 \pm 0$ & $1.0347 \pm 0.0413$ & $1.0237 \pm 0.0234$ & $1.0457 \pm 0.0446$ & $1.0245 \pm 0.024$ & $1.0425 \pm 0.0269$ & $1.0511 \pm 0.0461$ & $1.0237 \pm 0.0234$ & $2.7851 \pm 0.4783$ & $1.1291 \pm 0.089$ & $1.0237 \pm 0.0234$ \\
200-249 & $1 \pm 0$ & $1.0263 \pm 0.0348$ & $1.0171 \pm 0.0204$ & $1.0347 \pm 0.0412$ & $1.0178 \pm 0.021$ & $1.0333 \pm 0.0245$ & $1.0466 \pm 0.0482$ & $1.0171 \pm 0.0204$ & $2.9738 \pm 0.5577$ & $1.1076 \pm 0.0842$ & $1.0171 \pm 0.0204$ \\
250-299 & $1 \pm 0$ & $1.0243 \pm 0.0292$ & $1.0154 \pm 0.0186$ & $1.0296 \pm 0.0335$ & $1.0157 \pm 0.0188$ & $1.0291 \pm 0.0206$ & $1.0376 \pm 0.0361$ & $1.0154 \pm 0.0186$ & $3.0781 \pm 0.5858$ & $1.0951 \pm 0.0754$ & $1.0154 \pm 0.0186$ \\
300-349 & $1 \pm 0$ & $1.021 \pm 0.024$ & $1.0141 \pm 0.0157$ & $1.0296 \pm 0.0337$ & $1.0141 \pm 0.0157$ & $1.028 \pm 0.0183$ & $1.0404 \pm 0.0373$ & $1.0141 \pm 0.0157$ & $3.2411 \pm 0.5721$ & $1.0858 \pm 0.0588$ & $1.0141 \pm 0.0157$ \\
\bottomrule
\end{tabular}
}
\end{table}
\begin{table}[H]
\centering
\scriptsize
\caption{Effects of SPD on social welfare among scenarios with more than one PNE.}
\label{tab:correlation-filtered-social}
\resizebox{\textwidth}{!}{%
\begin{tabular}{@{}lrrrrrrrrrrr@{}}
\toprule
SPD & Optimal Social Welfare & Min-Max & Min-Sum & Max-Min Improvement & Nash & Kalai\_Smorodinsky & Egalitarian & Utilitarian & SIP|SIP & PoA & PoS \\
\midrule
0-10 & $1 \pm 0$ & $1.0471 \pm 0.0599$ & $1.0175 \pm 0.0229$ & $1.0337 \pm 0.0342$ & $1.0245 \pm 0.0259$ & $1.0289 \pm 0.027$ & $1.035 \pm 0.0362$ & $1.0175 \pm 0.0229$ & $1.8324 \pm 0.5566$ & $1.0825 \pm 0.0779$ & $1.0175 \pm 0.0229$ \\
11-20 & $1 \pm 0$ & $1.0376 \pm 0.0536$ & $1.0177 \pm 0.03$ & $1.0446 \pm 0.0581$ & $1.0233 \pm 0.0321$ & $1.0353 \pm 0.0349$ & $1.0527 \pm 0.0598$ & $1.0177 \pm 0.03$ & $1.8333 \pm 0.5978$ & $1.0887 \pm 0.0691$ & $1.0177 \pm 0.03$ \\
21-30 & $1 \pm 0$ & $1.0578 \pm 0.0665$ & $1.0154 \pm 0.0249$ & $1.0562 \pm 0.0736$ & $1.0232 \pm 0.0273$ & $1.0421 \pm 0.0398$ & $1.0696 \pm 0.0823$ & $1.0154 \pm 0.0249$ & $1.7211 \pm 0.4385$ & $1.1253 \pm 0.0988$ & $1.0154 \pm 0.0249$ \\
31-40 & $1 \pm 0$ & $1.0511 \pm 0.0702$ & $1.019 \pm 0.0263$ & $1.0672 \pm 0.0789$ & $1.0337 \pm 0.0326$ & $1.0431 \pm 0.0352$ & $1.0747 \pm 0.0927$ & $1.019 \pm 0.0263$ & $1.6647 \pm 0.3439$ & $1.1221 \pm 0.118$ & $1.019 \pm 0.0263$ \\
41-50 & $1 \pm 0$ & $1.0525 \pm 0.067$ & $1.0161 \pm 0.0247$ & $1.0626 \pm 0.0726$ & $1.0394 \pm 0.0508$ & $1.0452 \pm 0.0419$ & $1.0646 \pm 0.0704$ & $1.0161 \pm 0.0247$ & $1.6027 \pm 0.4042$ & $1.1063 \pm 0.0839$ & $1.0161 \pm 0.0247$ \\
51-60 & $1 \pm 0$ & $1.0564 \pm 0.0683$ & $1.0273 \pm 0.0485$ & $1.0555 \pm 0.0676$ & $1.0365 \pm 0.0489$ & $1.0416 \pm 0.0488$ & $1.0515 \pm 0.0573$ & $1.0273 \pm 0.0485$ & $1.5379 \pm 0.2965$ & $1.0849 \pm 0.0672$ & $1.0273 \pm 0.0485$ \\
61-70 & $1 \pm 0$ & $1.0222 \pm 0.029$ & $1.0074 \pm 0.0188$ & $1.0466 \pm 0.0541$ & $1.0189 \pm 0.0275$ & $1.0245 \pm 0.0239$ & $1.0474 \pm 0.0531$ & $1.0074 \pm 0.0188$ & $1.5202 \pm 0.2824$ & $1.0717 \pm 0.0478$ & $1.0074 \pm 0.0188$ \\
\bottomrule
\end{tabular}
}
\end{table}
\subsubsection{Individual path times in all scenarios}
\phantom{a}
\begin{table}[H]
\centering
\scriptsize
\caption{Effects of Cooperation Density on individual path times across all tested scenarios, including those with a single PNE.}
\label{tab:density-non-filtered-individual}
\resizebox{\textwidth}{!}{%
\begin{tabular}{@{}lrrrrrrrrrr@{}}
\toprule
Cooperation Density & SIP & SCP & Optimal Social Welfare & Min-Max & Min-Sum & Max-Min-Surplus & Nash & Kalai\_Smorodinsky & Egalitarian & Utilitarian \\
\midrule
0.1 & $1 \pm 0$ & $0.8463 \pm 0.071$ & $0.9044 \pm 0.0768$ & $0.9091 \pm 0.0766$ & $0.9091 \pm 0.0766$ & $0.9094 \pm 0.0761$ & $0.9091 \pm 0.0766$ & $0.9092 \pm 0.0764$ & $0.9094 \pm 0.0761$ & $0.9091 \pm 0.0766$ \\
0.2 & $1 \pm 0$ & $0.7247 \pm 0.0842$ & $0.8097 \pm 0.1115$ & $0.817 \pm 0.1103$ & $0.8154 \pm 0.1099$ & $0.8168 \pm 0.1086$ & $0.8157 \pm 0.1094$ & $0.8168 \pm 0.1087$ & $0.8183 \pm 0.1077$ & $0.8154 \pm 0.1099$ \\
0.3 & $1 \pm 0$ & $0.6179 \pm 0.0874$ & $0.7045 \pm 0.1231$ & $0.711 \pm 0.118$ & $0.7103 \pm 0.1185$ & $0.7114 \pm 0.1179$ & $0.7103 \pm 0.1178$ & $0.7113 \pm 0.1173$ & $0.7116 \pm 0.1173$ & $0.7102 \pm 0.1183$ \\
0.4 & $1 \pm 0$ & $0.516 \pm 0.0833$ & $0.6007 \pm 0.1143$ & $0.6076 \pm 0.1158$ & $0.6057 \pm 0.1151$ & $0.6094 \pm 0.1103$ & $0.6062 \pm 0.1125$ & $0.608 \pm 0.1109$ & $0.6111 \pm 0.1088$ & $0.6057 \pm 0.115$ \\
0.5 & $1 \pm 0$ & $0.4283 \pm 0.0955$ & $0.5188 \pm 0.1525$ & $0.5273 \pm 0.1484$ & $0.5253 \pm 0.1494$ & $0.5276 \pm 0.1475$ & $0.5254 \pm 0.1479$ & $0.527 \pm 0.1474$ & $0.5292 \pm 0.147$ & $0.5252 \pm 0.1493$ \\
0.6 & $1 \pm 0$ & $0.3518 \pm 0.0875$ & $0.459 \pm 0.1498$ & $0.4667 \pm 0.1496$ & $0.4657 \pm 0.1497$ & $0.4678 \pm 0.1479$ & $0.4657 \pm 0.149$ & $0.4685 \pm 0.1476$ & $0.4704 \pm 0.1467$ & $0.4656 \pm 0.1496$ \\
0.7 & $1 \pm 0$ & $0.303 \pm 0.0637$ & $0.4194 \pm 0.1405$ & $0.4273 \pm 0.1413$ & $0.4242 \pm 0.141$ & $0.4267 \pm 0.1414$ & $0.4243 \pm 0.1405$ & $0.4283 \pm 0.1381$ & $0.4288 \pm 0.1389$ & $0.4242 \pm 0.1411$ \\
\bottomrule
\end{tabular}
}
\end{table}
\begin{table}[H]
\centering
\scriptsize
\caption{Effects of Cooperation Magnitude on individual path times across all tested scenarios, including those with a single PNE.}
\label{tab:magnitude-non-filtered-individual}
\resizebox{\textwidth}{!}{%
\begin{tabular}{@{}lrrrrrrrrrr@{}}
\toprule
Cooperation Magnitude & SIP & SCP & Optimal Social Welfare & Min-Max & Min-Sum & Max-Min-Surplus & Nash & Kalai\_Smorodinsky & Egalitarian & Utilitarian \\
\midrule
2 & $1 \pm 0$ & $0.724 \pm 0.0304$ & $0.8019 \pm 0.075$ & $0.805 \pm 0.0748$ & $0.8037 \pm 0.0751$ & $0.8047 \pm 0.0739$ & $0.804 \pm 0.0746$ & $0.8045 \pm 0.0746$ & $0.8049 \pm 0.0737$ & $0.8038 \pm 0.0755$ \\
4 & $1 \pm 0$ & $0.496 \pm 0.0497$ & $0.6 \pm 0.113$ & $0.6068 \pm 0.1186$ & $0.6032 \pm 0.1179$ & $0.6048 \pm 0.1151$ & $0.6031 \pm 0.116$ & $0.6055 \pm 0.1157$ & $0.6065 \pm 0.1141$ & $0.6032 \pm 0.1167$ \\
6 & $1 \pm 0$ & $0.3897 \pm 0.0563$ & $0.4996 \pm 0.1218$ & $0.5037 \pm 0.122$ & $0.5021 \pm 0.122$ & $0.5039 \pm 0.1215$ & $0.5023 \pm 0.1216$ & $0.5046 \pm 0.1213$ & $0.5056 \pm 0.1201$ & $0.5022 \pm 0.1217$ \\
8 & $1 \pm 0$ & $0.3284 \pm 0.0625$ & $0.4355 \pm 0.1449$ & $0.4418 \pm 0.1394$ & $0.4391 \pm 0.1401$ & $0.442 \pm 0.1376$ & $0.4391 \pm 0.1395$ & $0.4428 \pm 0.1372$ & $0.446 \pm 0.1368$ & $0.4391 \pm 0.1401$ \\
10 & $1 \pm 0$ & $0.2866 \pm 0.0616$ & $0.4074 \pm 0.1695$ & $0.4193 \pm 0.1718$ & $0.4141 \pm 0.1709$ & $0.4172 \pm 0.1691$ & $0.4141 \pm 0.1707$ & $0.4184 \pm 0.1689$ & $0.4197 \pm 0.1677$ & $0.4143 \pm 0.1711$ \\
12 & $1 \pm 0$ & $0.2558 \pm 0.064$ & $0.3695 \pm 0.1284$ & $0.3774 \pm 0.1295$ & $0.3744 \pm 0.1302$ & $0.3786 \pm 0.1281$ & $0.3746 \pm 0.1289$ & $0.3789 \pm 0.1271$ & $0.3809 \pm 0.1267$ & $0.3744 \pm 0.1302$ \\
14 & $1 \pm 0$ & $0.2297 \pm 0.0592$ & $0.3357 \pm 0.1241$ & $0.3456 \pm 0.1285$ & $0.3426 \pm 0.1275$ & $0.3448 \pm 0.1256$ & $0.3426 \pm 0.1272$ & $0.3459 \pm 0.1255$ & $0.3468 \pm 0.1247$ & $0.3425 \pm 0.1275$ \\
\bottomrule
\end{tabular}
}
\end{table}
\begin{table}[H]
\centering
\scriptsize
\caption{Effects of Path Length on individual path times across all tested scenarios, including those with a single PNE.}
\label{tab:extent-non-filtered-individual}
\resizebox{\textwidth}{!}{%
\begin{tabular}{@{}lrrrrrrrrrr@{}}
\toprule
Path Length & SIP & SCP & Optimal Social Welfare & Min-Max & Min-Sum & Max-Min-Surplus & Nash & Kalai\_Smorodinsky & Egalitarian & Utilitarian \\
\midrule
0-49 & $1 \pm 0$ & $0.3391 \pm 0.1714$ & $1.0258 \pm 2.0233$ & $0.7957 \pm 0.2067$ & $0.7925 \pm 0.2101$ & $0.796 \pm 0.2023$ & $0.7939 \pm 0.2033$ & $0.7965 \pm 0.2001$ & $0.7986 \pm 0.1977$ & $0.7925 \pm 0.2101$ \\
50-99 & $1 \pm 0$ & $0.3133 \pm 0.0938$ & $0.5839 \pm 0.2185$ & $0.6068 \pm 0.2121$ & $0.6027 \pm 0.2119$ & $0.6052 \pm 0.2085$ & $0.6028 \pm 0.2111$ & $0.6066 \pm 0.2075$ & $0.6071 \pm 0.2068$ & $0.6033 \pm 0.2127$ \\
100-149 & $1 \pm 0$ & $0.2921 \pm 0.075$ & $0.4545 \pm 0.1507$ & $0.4692 \pm 0.1565$ & $0.4632 \pm 0.1542$ & $0.4679 \pm 0.1527$ & $0.4633 \pm 0.1527$ & $0.4701 \pm 0.1494$ & $0.4718 \pm 0.1505$ & $0.4632 \pm 0.154$ \\
150-199 & $1 \pm 0$ & $0.2783 \pm 0.0584$ & $0.3962 \pm 0.1244$ & $0.4061 \pm 0.1256$ & $0.4025 \pm 0.1234$ & $0.4061 \pm 0.1219$ & $0.4024 \pm 0.1226$ & $0.4067 \pm 0.1204$ & $0.4071 \pm 0.1205$ & $0.4025 \pm 0.1235$ \\
200-249 & $1 \pm 0$ & $0.2698 \pm 0.0613$ & $0.3699 \pm 0.132$ & $0.3792 \pm 0.1322$ & $0.3764 \pm 0.1315$ & $0.3796 \pm 0.1296$ & $0.3764 \pm 0.1306$ & $0.3801 \pm 0.129$ & $0.3822 \pm 0.1288$ & $0.3765 \pm 0.1314$ \\
250-299 & $1 \pm 0$ & $0.2717 \pm 0.0595$ & $0.3605 \pm 0.1256$ & $0.369 \pm 0.1292$ & $0.3664 \pm 0.1289$ & $0.3688 \pm 0.1271$ & $0.3662 \pm 0.1282$ & $0.3695 \pm 0.1267$ & $0.3708 \pm 0.1263$ & $0.3661 \pm 0.1286$ \\
300-349 & $1 \pm 0$ & $0.2676 \pm 0.0563$ & $0.3387 \pm 0.1146$ & $0.3444 \pm 0.1098$ & $0.3427 \pm 0.1104$ & $0.3458 \pm 0.1093$ & $0.3427 \pm 0.1101$ & $0.3457 \pm 0.1084$ & $0.348 \pm 0.1081$ & $0.3427 \pm 0.1103$ \\
\bottomrule
\end{tabular}
}
\end{table}
\begin{table}[H]
\centering
\scriptsize
\caption{Effects of Cooperation SPD on individual path times across all tested scenarios, including those with a single PNE.}
\label{tab:correlation-non-filtered-individual}
\resizebox{\textwidth}{!}{%
\begin{tabular}{@{}lrrrrrrrrrr@{}}
\toprule
SPD & SIP & SCP & Optimal Social Welfare & Min-Max & Min-Sum & Max-Min-Surplus & Nash & Kalai\_Smorodinsky & Egalitarian & Utilitarian \\
\midrule
0-10 & $1 \pm 0$ & $0.3013 \pm 0.088$ & $0.7662 \pm 0.3309$ & $0.7624 \pm 0.2265$ & $0.7522 \pm 0.2266$ & $0.7519 \pm 0.2247$ & $0.7509 \pm 0.2251$ & $0.7532 \pm 0.2235$ & $0.7521 \pm 0.224$ & $0.7521 \pm 0.2266$ \\
11-20 & $1 \pm 0$ & $0.298 \pm 0.0765$ & $0.7823 \pm 0.3023$ & $0.7924 \pm 0.2386$ & $0.786 \pm 0.2406$ & $0.7877 \pm 0.2363$ & $0.7856 \pm 0.2391$ & $0.7881 \pm 0.2361$ & $0.7892 \pm 0.2347$ & $0.7855 \pm 0.2404$ \\
21-30 & $1 \pm 0$ & $0.3039 \pm 0.09$ & $0.7821 \pm 0.2585$ & $0.8026 \pm 0.2046$ & $0.7895 \pm 0.214$ & $0.7956 \pm 0.2017$ & $0.7896 \pm 0.2083$ & $0.7944 \pm 0.2013$ & $0.7983 \pm 0.1965$ & $0.7895 \pm 0.214$ \\
31-40 & $1 \pm 0$ & $0.2958 \pm 0.0841$ & $0.8236 \pm 0.2562$ & $0.8437 \pm 0.1949$ & $0.837 \pm 0.2047$ & $0.8433 \pm 0.1914$ & $0.838 \pm 0.1974$ & $0.84 \pm 0.1947$ & $0.844 \pm 0.1887$ & $0.837 \pm 0.2047$ \\
41-50 & $1 \pm 0$ & $0.2983 \pm 0.0841$ & $0.8032 \pm 0.272$ & $0.8404 \pm 0.2019$ & $0.8325 \pm 0.2162$ & $0.8392 \pm 0.1998$ & $0.8352 \pm 0.204$ & $0.8369 \pm 0.2025$ & $0.8399 \pm 0.1983$ & $0.8321 \pm 0.2163$ \\
51-60 & $1 \pm 0$ & $0.2944 \pm 0.0965$ & $0.8758 \pm 0.264$ & $0.9065 \pm 0.167$ & $0.901 \pm 0.1734$ & $0.9037 \pm 0.1687$ & $0.9016 \pm 0.1706$ & $0.9029 \pm 0.1687$ & $0.9032 \pm 0.1684$ & $0.9011 \pm 0.1736$ \\
61-70 & $1 \pm 0$ & $0.2971 \pm 0.0885$ & $0.9148 \pm 0.2704$ & $0.9313 \pm 0.1393$ & $0.9283 \pm 0.1441$ & $0.9309 \pm 0.1365$ & $0.929 \pm 0.1399$ & $0.9298 \pm 0.1388$ & $0.931 \pm 0.1361$ & $0.9283 \pm 0.1441$ \\
\bottomrule
\end{tabular}
}
\end{table}
\subsubsection{Social welfare in all scenarios}
\phantom{a}
\begin{table}[H]
\centering
\scriptsize
\caption{Effects of Cooperation Density on social welfare across all tested scenarios, including those with a single PNE.}
\label{tab:density-non-filtered-social}
\resizebox{\textwidth}{!}{%
\begin{tabular}{@{}lrrrrrrrrrrr@{}}
\toprule
Cooperation Density & Optimal Social Welfare & Min-Max & Min-Sum & Max-Min Improvement & Nash & Kalai\_Smorodinsky & Egalitarian & Utilitarian & SIP|SIP & PoA & PoS \\
\midrule
0.1 & $1 \pm 0$ & $1.0055 \pm 0.0115$ & $1.0055 \pm 0.0115$ & $1.0059 \pm 0.0119$ & $1.0055 \pm 0.0115$ & $1.0056 \pm 0.0115$ & $1.0059 \pm 0.0119$ & $1.0055 \pm 0.0115$ & $1.1129 \pm 0.0819$ & $1.006 \pm 0.012$ & $1.0055 \pm 0.0115$ \\
0.2 & $1 \pm 0$ & $1.0098 \pm 0.0197$ & $1.0081 \pm 0.018$ & $1.0101 \pm 0.0199$ & $1.0086 \pm 0.018$ & $1.0099 \pm 0.0186$ & $1.0121 \pm 0.0218$ & $1.0081 \pm 0.018$ & $1.2579 \pm 0.1599$ & $1.0157 \pm 0.0278$ & $1.0081 \pm 0.018$ \\
0.3 & $1 \pm 0$ & $1.0109 \pm 0.0197$ & $1.0102 \pm 0.0196$ & $1.0121 \pm 0.0213$ & $1.0104 \pm 0.0198$ & $1.0116 \pm 0.0199$ & $1.0123 \pm 0.0214$ & $1.0102 \pm 0.0196$ & $1.4557 \pm 0.2137$ & $1.0167 \pm 0.0252$ & $1.0102 \pm 0.0196$ \\
0.4 & $1 \pm 0$ & $1.0117 \pm 0.0184$ & $1.0089 \pm 0.0143$ & $1.0167 \pm 0.0298$ & $1.01 \pm 0.0157$ & $1.0133 \pm 0.0166$ & $1.0199 \pm 0.0304$ & $1.0089 \pm 0.0143$ & $1.7142 \pm 0.3021$ & $1.0278 \pm 0.0412$ & $1.0089 \pm 0.0143$ \\
0.5 & $1 \pm 0$ & $1.0182 \pm 0.0371$ & $1.0148 \pm 0.034$ & $1.0197 \pm 0.0373$ & $1.0153 \pm 0.034$ & $1.0185 \pm 0.0337$ & $1.0231 \pm 0.0385$ & $1.0148 \pm 0.034$ & $2.0435 \pm 0.4506$ & $1.0338 \pm 0.0494$ & $1.0148 \pm 0.034$ \\
0.6 & $1 \pm 0$ & $1.0193 \pm 0.0339$ & $1.0172 \pm 0.0318$ & $1.0243 \pm 0.0372$ & $1.0176 \pm 0.0322$ & $1.0242 \pm 0.0341$ & $1.0308 \pm 0.0451$ & $1.0172 \pm 0.0318$ & $2.3437 \pm 0.5641$ & $1.0645 \pm 0.09$ & $1.0172 \pm 0.0318$ \\
0.7 & $1 \pm 0$ & $1.0207 \pm 0.0327$ & $1.0149 \pm 0.0184$ & $1.0217 \pm 0.0294$ & $1.0152 \pm 0.0186$ & $1.0248 \pm 0.024$ & $1.0277 \pm 0.0314$ & $1.0149 \pm 0.0184$ & $2.5726 \pm 0.6138$ & $1.066 \pm 0.0798$ & $1.0149 \pm 0.0184$ \\
\bottomrule
\end{tabular}
}
\end{table}
\begin{table}[H]
\centering
\scriptsize
\caption{Effects of Cooperation Magnitude on social welfare across all tested scenarios, including those with a single PNE.}
\label{tab:magnitude-non-filtered-social}
\resizebox{\textwidth}{!}{%
\begin{tabular}{@{}lrrrrrrrrrrr@{}}
\toprule
Cooperation Magnitude & Optimal Social Welfare & Min-Max & Min-Sum & Max-Min Improvement & Nash & Kalai\_Smorodinsky & Egalitarian & Utilitarian & SIP|SIP & PoA & PoS \\
\midrule
2 & $1 \pm 0$ & $1.004 \pm 0.0184$ & $1.0026 \pm 0.0176$ & $1.0039 \pm 0.0181$ & $1.0029 \pm 0.0177$ & $1.0035 \pm 0.0177$ & $1.0041 \pm 0.0182$ & $1.0026 \pm 0.0176$ & $1.255 \pm 0.0766$ & $1.006 \pm 0.0198$ & $1.0026 \pm 0.0176$ \\
4 & $1 \pm 0$ & $1.011 \pm 0.036$ & $1.0061 \pm 0.033$ & $1.0091 \pm 0.0351$ & $1.0062 \pm 0.033$ & $1.0096 \pm 0.0331$ & $1.012 \pm 0.0362$ & $1.0061 \pm 0.033$ & $1.7072 \pm 0.2326$ & $1.0233 \pm 0.0443$ & $1.0061 \pm 0.033$ \\
6 & $1 \pm 0$ & $1.0089 \pm 0.0213$ & $1.0066 \pm 0.0137$ & $1.0113 \pm 0.0209$ & $1.007 \pm 0.0138$ & $1.0115 \pm 0.0141$ & $1.0152 \pm 0.0245$ & $1.0066 \pm 0.0137$ & $2.0852 \pm 0.3849$ & $1.0339 \pm 0.0489$ & $1.0066 \pm 0.0137$ \\
8 & $1 \pm 0$ & $1.0173 \pm 0.0248$ & $1.0118 \pm 0.019$ & $1.0213 \pm 0.0284$ & $1.0124 \pm 0.0191$ & $1.0211 \pm 0.0229$ & $1.0313 \pm 0.0403$ & $1.0118 \pm 0.019$ & $2.4446 \pm 0.4936$ & $1.0583 \pm 0.0702$ & $1.0118 \pm 0.019$ \\
10 & $1 \pm 0$ & $1.0281 \pm 0.0579$ & $1.0192 \pm 0.048$ & $1.0286 \pm 0.0558$ & $1.0192 \pm 0.048$ & $1.029 \pm 0.0496$ & $1.0355 \pm 0.0579$ & $1.0192 \pm 0.048$ & $2.699 \pm 0.6755$ & $1.0696 \pm 0.0808$ & $1.0192 \pm 0.048$ \\
12 & $1 \pm 0$ & $1.0214 \pm 0.0313$ & $1.0154 \pm 0.0211$ & $1.0282 \pm 0.0353$ & $1.0163 \pm 0.0212$ & $1.0273 \pm 0.0273$ & $1.0357 \pm 0.0419$ & $1.0154 \pm 0.0211$ & $2.9286 \pm 0.7054$ & $1.0687 \pm 0.0762$ & $1.0154 \pm 0.0211$ \\
14 & $1 \pm 0$ & $1.0292 \pm 0.0426$ & $1.0219 \pm 0.0309$ & $1.0319 \pm 0.0399$ & $1.0222 \pm 0.0311$ & $1.032 \pm 0.0339$ & $1.0379 \pm 0.0447$ & $1.0219 \pm 0.0309$ & $3.2267 \pm 0.7949$ & $1.0745 \pm 0.0867$ & $1.0219 \pm 0.0309$ \\
\bottomrule
\end{tabular}
}
\end{table}
\begin{table}[H]
\centering
\scriptsize
\caption{Effects of Path Length on social welfare across all tested scenarios, including those with a single PNE.}
\label{tab:extent-non-filtered-social}
\resizebox{\textwidth}{!}{%
\begin{tabular}{@{}lrrrrrrrrrrr@{}}
\toprule
Path Length & Optimal Social Welfare & Min-Max & Min-Sum & Max-Min Improvement & Nash & Kalai\_Smorodinsky & Egalitarian & Utilitarian & SIP|SIP & PoA & PoS \\
\midrule
0-49 & $1 \pm 0$ & $1.0811 \pm 0.1156$ & $1.0766 \pm 0.1115$ & $1.089 \pm 0.1142$ & $1.0846 \pm 0.1116$ & $1.0869 \pm 0.1098$ & $1.0935 \pm 0.1129$ & $1.0766 \pm 0.1115$ & $1.4459 \pm 0.3397$ & $1.1041 \pm 0.1167$ & $1.0766 \pm 0.1115$ \\
50-99 & $1 \pm 0$ & $1.0522 \pm 0.048$ & $1.0472 \pm 0.0457$ & $1.0571 \pm 0.0602$ & $1.0478 \pm 0.0459$ & $1.0558 \pm 0.0493$ & $1.0612 \pm 0.0627$ & $1.0472 \pm 0.0457$ & $1.862 \pm 0.4749$ & $1.0915 \pm 0.0981$ & $1.0472 \pm 0.0457$ \\
100-149 & $1 \pm 0$ & $1.0321 \pm 0.0447$ & $1.0223 \pm 0.0281$ & $1.0372 \pm 0.0405$ & $1.0235 \pm 0.028$ & $1.0378 \pm 0.0326$ & $1.0466 \pm 0.0494$ & $1.0223 \pm 0.0281$ & $2.3512 \pm 0.5157$ & $1.0929 \pm 0.0841$ & $1.0223 \pm 0.0281$ \\
150-199 & $1 \pm 0$ & $1.0263 \pm 0.0355$ & $1.02 \pm 0.0229$ & $1.0326 \pm 0.0395$ & $1.0205 \pm 0.0234$ & $1.0308 \pm 0.0281$ & $1.0356 \pm 0.0415$ & $1.02 \pm 0.0229$ & $2.6701 \pm 0.5452$ & $1.0802 \pm 0.0889$ & $1.02 \pm 0.0229$ \\
200-249 & $1 \pm 0$ & $1.0266 \pm 0.0458$ & $1.0207 \pm 0.0402$ & $1.032 \pm 0.0492$ & $1.0211 \pm 0.0403$ & $1.031 \pm 0.0414$ & $1.0396 \pm 0.0538$ & $1.0207 \pm 0.0402$ & $2.8879 \pm 0.6148$ & $1.0786 \pm 0.0857$ & $1.0207 \pm 0.0402$ \\
250-299 & $1 \pm 0$ & $1.0232 \pm 0.0378$ & $1.0172 \pm 0.033$ & $1.0268 \pm 0.0403$ & $1.0174 \pm 0.0331$ & $1.0265 \pm 0.0339$ & $1.0323 \pm 0.0424$ & $1.0172 \pm 0.033$ & $2.9731 \pm 0.6604$ & $1.0714 \pm 0.0769$ & $1.0172 \pm 0.033$ \\
300-349 & $1 \pm 0$ & $1.0188 \pm 0.0258$ & $1.0141 \pm 0.0208$ & $1.0246 \pm 0.0329$ & $1.0142 \pm 0.0208$ & $1.0235 \pm 0.0231$ & $1.0318 \pm 0.0368$ & $1.0141 \pm 0.0208$ & $3.1339 \pm 0.6427$ & $1.0623 \pm 0.0609$ & $1.0141 \pm 0.0208$ \\
\bottomrule
\end{tabular}
}
\end{table}
\begin{table}[H]
\centering
\scriptsize
\caption{Effects of SPD on social welfare across all tested scenarios, including those with a single PNE.}
\label{tab:correlation-non-filtered-social}
\resizebox{\textwidth}{!}{%
\begin{tabular}{@{}lrrrrrrrrrrr@{}}
\toprule
SPD & Optimal Social Welfare & Min-Max & Min-Sum & Max-Min Improvement & Nash & Kalai\_Smorodinsky & Egalitarian & Utilitarian & SIP|SIP & PoA & PoS \\
\midrule
0-10 & $1 \pm 0$ & $1.0435 \pm 0.0778$ & $1.0352 \pm 0.0729$ & $1.0397 \pm 0.0734$ & $1.0371 \pm 0.0727$ & $1.0384 \pm 0.0727$ & $1.0401 \pm 0.0736$ & $1.0352 \pm 0.0729$ & $1.5433 \pm 0.6042$ & $1.0534 \pm 0.0841$ & $1.0352 \pm 0.0729$ \\
11-20 & $1 \pm 0$ & $1.0371 \pm 0.0619$ & $1.0308 \pm 0.0573$ & $1.0394 \pm 0.0633$ & $1.0326 \pm 0.0574$ & $1.0364 \pm 0.0575$ & $1.0419 \pm 0.0641$ & $1.0308 \pm 0.0573$ & $1.492 \pm 0.6542$ & $1.0533 \pm 0.0708$ & $1.0308 \pm 0.0573$ \\
21-30 & $1 \pm 0$ & $1.0472 \pm 0.0731$ & $1.0328 \pm 0.0645$ & $1.0467 \pm 0.0753$ & $1.0355 \pm 0.0642$ & $1.0419 \pm 0.0658$ & $1.0512 \pm 0.0791$ & $1.0328 \pm 0.0645$ & $1.4069 \pm 0.4315$ & $1.0702 \pm 0.0931$ & $1.0328 \pm 0.0645$ \\
31-40 & $1 \pm 0$ & $1.0486 \pm 0.0785$ & $1.0391 \pm 0.0712$ & $1.0534 \pm 0.0813$ & $1.0435 \pm 0.071$ & $1.0463 \pm 0.0711$ & $1.0557 \pm 0.086$ & $1.0391 \pm 0.0712$ & $1.3215 \pm 0.3558$ & $1.0697 \pm 0.0999$ & $1.0391 \pm 0.0712$ \\
41-50 & $1 \pm 0$ & $1.0691 \pm 0.0983$ & $1.0581 \pm 0.0956$ & $1.0721 \pm 0.0991$ & $1.0651 \pm 0.0962$ & $1.0669 \pm 0.0945$ & $1.0727 \pm 0.0985$ & $1.0581 \pm 0.0956$ & $1.3495 \pm 0.3958$ & $1.0852 \pm 0.1024$ & $1.0581 \pm 0.0956$ \\
51-60 & $1 \pm 0$ & $1.0575 \pm 0.0815$ & $1.0531 \pm 0.0801$ & $1.0573 \pm 0.0815$ & $1.0545 \pm 0.0798$ & $1.0552 \pm 0.0796$ & $1.0567 \pm 0.0803$ & $1.0531 \pm 0.0801$ & $1.2175 \pm 0.2997$ & $1.0618 \pm 0.082$ & $1.0531 \pm 0.0801$ \\
61-70 & $1 \pm 0$ & $1.0474 \pm 0.0711$ & $1.045 \pm 0.0716$ & $1.0514 \pm 0.0727$ & $1.0469 \pm 0.0712$ & $1.0478 \pm 0.0707$ & $1.0515 \pm 0.0725$ & $1.045 \pm 0.0716$ & $1.1595 \pm 0.228$ & $1.0555 \pm 0.0723$ & $1.045 \pm 0.0716$ \\
\bottomrule
\end{tabular}
}
\end{table}

\subsection{Execution Environment}
The experiments were conducted on a MacBook Pro with an Apple M1 Pro processor and 16 GB of RAM, running macOS 26.5. The implementation was written in Python 3.11.3 and relies primarily on the Python standard library, with NetworkX used for graph representation and graph algorithms. Matplotlib was used to generate the plots and simulation visualizations. No GPU acceleration was used.
Randomness was used only for sampling experimental scenarios from the benchmark maps, using Python's built-in \texttt{random} module. No fixed random seed was used. However, the actual sampled instances are publicly available as serialized pickle files representing the generated graphs, together with the experimental data and aggregated results, at \url{https://iscmpp.info/}.

\end{document}